\documentclass[acmsmall,nonacm]{acmart}
\usepackage{amsthm,amsmath}
\usepackage{todonotes}
\usepackage{hyperref}
\usepackage{tikz}
\usetikzlibrary{shapes,fit}

\newcommand{\datarule}{{\,:\!\!-\,}} %

\newtheorem{theorem}{Theorem}
\newtheorem{proposition}[theorem]{Proposition}
\newtheorem{lemma}[theorem]{Lemma}
\newtheorem{corollary}[theorem]{Corollary}
\newtheorem{observation}[theorem]{Observation}
\newtheorem{conjecture}{Conjecture}
\newtheorem{claim}{Claim}

\theoremstyle{definition}
\newtheorem{definition}[theorem]{Definition}
\newtheorem{example}[theorem]{Example}

\newcommand{\qstar}{Q^\star}
\newcommand{\qstarp}{\overline{Q}^\star}
\newcommand{\Qsf}{Q^\text{sf}}

\newcommand{\softO}{\widetilde{O}}

\newcommand{\dom}{\mathrm{dom}}
\newcommand{\var}{\mathrm{var}}
\newcommand{\aut}{\mathrm{aut}}
\newcommand{\id}{\mathrm{id}}
\newcommand{\Nid}{\hom(A_{Q}, D, \mathfrak c)^{\id}}
\newcommand{\Naut}{\hom(A_{Q}, D, \mathfrak c)^{\aut}}

\newcommand{\eps}{\varepsilon}
\newcommand{\I}{\mathcal{I}}
\newcommand{\A}{\mathcal{A}}

\newcommand{\Fp}{\mathbb{F}_p}

\newcommand{\ignore}[1]{}

\newcommand{\pb}[1]{\textsf{#1}}

\begin{document}
	\title{Tight Fine-Grained Bounds for Direct Access on Join Queries}
	\thanks{\emph{Karl Bringmann:} This work is part of the project TIPEA that has received funding from the European Research Council (ERC) under the European Unions Horizon 2020 research and innovation programme (grant agreement No. 850979). \emph{Nofar Carmeli and Stefan Mengel:} This work has been funded by the French government under
		management of Agence Nationale de la Recherche as part of the
		“Investissements d’avenir” program, reference ANR-19-P3IA-0001
		(PRAIRIE 3IA Institute) and
		by the ANR project EQUUS ANR-19-CE48-0019.}
	\author{Karl Bringmann}
	\affiliation{%
		\institution{Saarland University and Max Planck Institute for Informatics}
		\department{Saarland Informatics Campus}
		\city{Saarbrücken}
		\country{Germany}}
	
	\author{Nofar Carmeli}
	\affiliation{%
		\institution{LIRMM, Inria, University of Montpellier, CNRS}
		\country{France}
	}
	
	\author{Stefan Mengel}
	\affiliation{%
		\institution{Univ. Artois, CNRS, Centre de Recherche en Informatique de Lens (CRIL)}
		\city{Lens}
		\country{France}}

	\begin{abstract}
		We consider the task of lexicographic direct access to query answers. That is, we want to simulate an array containing the answers of a join query sorted in a lexicographic order chosen by the user. A recent dichotomy showed for which queries and orders this task can be done in polylogarithmic access time after quasilinear preprocessing, but this dichotomy does not tell us how much time is required in the cases classified as hard. We determine the preprocessing time needed to achieve polylogarithmic access time for all join queries and all lexicographical orders. To this end, we propose a decomposition-based general algorithm for direct access on join queries. We then explore its optimality by proving lower bounds for the preprocessing time based on the hardness of a certain online Set-Disjointness problem, which shows that our algorithm's bounds are tight for all lexicographic orders on join queries. Then, we prove the hardness of Set-Disjointness based on the Zero-Clique Conjecture which is an established conjecture from fine-grained complexity theory.
		Interestingly, while proving our lower bound, we show that self-joins do not affect the complexity of direct access (up to logarithmic factors).
		Our algorithm can also be used to solve queries with projections and relaxed order requirements, though in these cases, its running time is not necessarily optimal.
		We also show that similar techniques to those used in our lower bounds can be used to prove that, for enumerating answers to Loomis-Whitney joins, it is not possible to significantly improve upon trivially computing all answers at preprocessing. This, in turn, gives further evidence (based on the Zero-Clique Conjecture) to the enumeration hardness of self-join free cyclic joins with respect to linear preprocessing and constant delay.
	\end{abstract}
	
	
	\maketitle
	
	\section{Introduction}
	
	Algorithms for direct access allow to access answers to a query on a database essentially as if they are materialized and stored in an array: given an index $j\in \mathbb{N}$, the algorithm returns the $j$th answer (or an out-of-bounds error) in very little time. To make this possible, the algorithm first runs a preprocessing on the database that then allows answering arbitrary access queries efficiently.
	As the number of answers to a database query may be orders-of-magnitude larger than the size of the database, the goal is to avoid materializing all the answers during preprocessing and only simulate the array.
	
	The direct access task (previously also called $j$th answer and random access) was introduced by Bagan et al.~\cite{BaganDGO08}. They also mentioned that this task can be used for uniform sampling (by first counting and drawing a random index) or for enumerating all query answers (by consecutively accessing all indices). They devised an algorithm that runs in only linear preprocessing time and average constant time per access call for a large class of queries (first-order queries) on databases of bounded degree. Direct access algorithms were also considered for queries in monadic second order logic on databases of bounded treewidth~\cite{Bagan09}. To reason about general databases and still expect extremely efficient algorithms, another approach is to restrict the class of queries instead.
	Follow-up work gave linear preprocessing and logarithmic access algorithms for subclasses of conjunctive queries over general databases~\cite{bb:thesis, CarmeliRandom, Keppeler2020Answering}. There, it was also explained how direct access can be used to sample without repetitions~\cite{CarmeliRandom}.
	
	Though all of the algorithms we mentioned above simulate storing the answers in a lexicographic order (whether they state it or not),
	one shortcoming they have in common is that the specific lexicographic order cannot be chosen by the user (but rather it depends on the query structure).
	Allowing the user to specify the order is desirable because then direct access can be used for additional tasks that are sensitive to the order, such as median finding and boxplot computation.
	Progress in this direction was recently made by Carmeli et al.~\cite{CarmeliTGKR20} who identified which combinations of a conjunctive query and a lexicographic order can be accessed with linear preprocessing time and logarithmic access time. They identified an easy-to-check substructure of the query, called \emph{disruptive trios}, whose (non-)existence distinguishes the tractable cases (w.r.t.~the time guarantees we mentioned) from the intractable ones.
	In particular, 
	if we consider acyclic join queries, they suggested an algorithm that works for any lexicographic order that does not have disruptive trios with respect to the query. If the join query is also self-join free, they proved conditional lower bounds stating that for all other variable orders, direct access with polylogarithmic direct access requires superlinear time preprocessing.
	These hardness results assume the hardness of Boolean matrix multiplication.
	Given the known hardness of self-join free cyclic joins, if we also assume the hardness of hyperclique detection in hypergraphs, this gives a dichotomy for all self-join free join queries and lexicographic orders: they are tractable if and only if the query is acyclic with no disruptive trio with respect to the order.
	
	What happens if the query and order we want to compute happen to fall on the intractable side of the dichotomy? Also, what is the role of self-joins in the complexity of direct access to join queries? These questions were left open by previous work, and we aim to understand how much preprocessing is needed to achieve polylogarithmic access time for each combination of query, potentially containing self-joins, and order.
	To this end, we introduce \emph{disruption-free decompositions} of a query with respect to a variable order. These can be seen as hypertree decompositions of the queries, induced by the desired variable orders, that resolve incompatibilities between the order and the query. Practically, these decompositions specify which relations should be joined at preprocessing in order to achieve an equivalent acyclic join query with no disruptive trios. We can then run the known direct access algorithm~\cite{CarmeliTGKR20} with linear time preprocessing on the result  of this preprocessing to get an algorithm for the query and order at hand with logarithmic access time.
	The cost of our preprocesing phase is therefore dominated by the time it takes to join the selected relations.
	We define the \emph{incompatibility number} of a query and order and show that the preprocessing time of our solution is polynomial where the exponent is this number. Intuitively, the incompatibility number is $1$ when the query is acyclic and the order is compatible, and this number grows as the incompatibility between the query and order grows.
	
	Next, we aspire to know whether our solution can be improved. We show that, somewhat surprisingly, self-joins have no influence whatsoever on the complexity of direct access to join queries: if we change some relation symbols in a join query, this leaves the preprocessing and access times unchanged up to polylogarithmic factors. Note that this is completely different than the situation for related query evaluation tasks, in particular enumeration, where self-joins are known to change the complexity~\cite{berkholz2020tutorial,CarmeliS22}. The situation for direct access is instead similar to that for counting, where it is known that the complexity depends only on the underlying hypergraph, and thus self-joins play no role in the complexity analysis~\cite{DalmauJ04}. In fact, our results regarding the elimination of self-joins in the case of direct access are obtained by adapting similar proofs for counting complexity~\cite{DalmauJ04,ChenM15}.
	Having dealt with self-joins, for the rest of our proofs we can consider only self-join free queries.
	
	Though we can easily show that no other decomposition can be better than the specific decomposition we propose, we can still wonder whether a better algorithm can be achieved using an alternative technique. Thus, we set out to prove conditional lower bounds.
	Such lower bounds show hardness independently of a specific algorithmic technique, but instead they assume that some known problem cannot be solved significantly faster than by the state-of-the art algorithms.
	We show that the incompatibility number corresponds in a sense to the number of leaves of the largest star query that can be embedded in our given query. 
	We then prove lower bounds for queries that allow embedding stars through a reduction from online $k$-Set-Disjointness. In this problem, we are given during preprocessing $k$ sets of subsets of the universe, and then we need to answer queries that specify one subset from each set and ask whether these subsets are disjoint.
	On a technical level, in case the query is acyclic, the link of these hardness results with the existence of disruptive trios is that the latter correspond exactly to the possibility to embed a star with two leaves. 
	
	Using known hardness results for $2$-Set-Disjointness, our reduction shows that the acyclic hard cases in the known dichotomy need at least quadratic preprocessing, unless both the 3-SUM Conjecture and the APSP Conjecture fail. These are both central, well-established conjectures in fine-grained complexity theory, see e.g.~the survey~\cite{williams2018some}.
	To have tighter lower bounds for the case that the incompatibility number is not $2$, we show the hardness of $k$-Set-Disjointness through a reduction from the \emph{\pb{Zero-$k$-Clique} Conjecture}.
	This conjecture postulates that deciding whether a given edge-weighted $n$-node graph contains a $k$-clique of total edge weight 0 has no algorithm running in time $O(n^{k-\eps})$ for any $\eps > 0$.
	For $k=3$ this conjecture is implied by the 3SUM Conjecture and the APSP Conjecture, so it is very believable. For $k>3$ the \pb{Zero-$k$-Clique} Conjecture is a natural generalization of the case $k=3$, and it was recently used in several contexts~\cite{LincolnWW18,AbboudBDN18,BringmannGMW20,AbboudWW14,BackursDT16,BackursT17}.
	Assuming the \pb{Zero-$k$-Clique} Conjecture, we prove that the preprocessing time of our decomposition-based algorithm is (near-)optimal. 
	
	To conclude, our \textbf{main result} is as follows:
	a join query, potentially with self-joins, and an ordering of its variables with incompatibility number $\iota$ admit a lexicographic direct access algorithm with preprocessing time $O(|D|^\iota)$ and logarithmic access time.
	Moreover, if the Zero-$k$-Clique conjecture holds for all $k$, there is no lexicographic direct access algorithm for this query and order with preprocessing time $O(|D|^{\iota-\epsilon})$ and polylogarithmic access time for any $\varepsilon > 0$.
	
	We also consider several extensions of this main result: we show that if we require the answers to be ordered in any lexicographic order---but do not specify which one---then the best running time we can expect is determined by the fractional hypertree width of the query. However, if we make no restriction on the order of the answers (i.e., not requiring it to be lexicographic), then this complexity bound can be improved. We show this by considering the $4$-cycle query that has faster direct access than the one determined by its fractional hypertree width for lexicographic order. Finally, we extend our setting to partial lexicographic orders and queries with projections and show how our algorithm can handle these cases. However, we again show that this algorithm is not optimal for certain queries, leaving the identification of optimal runtime bounds open.
	
	As we develop our lower bound results, we notice that our techniques can also be used in the context of constant delay enumeration of query answers.
	We show that, assuming the Zero-$k$-Clique Conjecture, the preprocessing of any constant delay enumeration algorithm for the $k$-variable Loomis-Whitney join is roughly at least $\Omega(|D|^{1+1/(k-1)})$ which tightly matches the trivial algorithm in which the answers are materialized during preprocessing using a worst-case optimal join algorithm~\cite{NgoPRR12,NgoPRR18,Veldhuizen14}. From the lower bound for Loomis-Whitney joins, we then infer the hardness of other cyclic joins using a construction by Brault-Baron~\cite{bb:thesis}. Specifically, we conclude that the self-join free join queries that allow constant delay enumeration after linear processing are exactly the acyclic ones, unless the Zero-$k$-Clique Conjecture fails for some $k$. This dichotomy was already established based on the hardness of hyperclique detection in hypergraphs~\cite{bb:thesis, berkholz2020tutorial}; we here give more evidence for this dichotomy by showing it also holds under a different complexity assumption.
	
	A preliminary version of this manuscript appeared
	in a conference proceedings~\cite{confversion}. The current version contains full proofs omitted in the previous version, as well as two new sections: Section~\ref{sec:selfjoins} and Section~\ref{sec:relaxations}, discussing the implications of our main result for queries with self-joins, queries with projections, and relaxed order requirements.
	
	This manuscript is organized as follows. 
	Preliminary definitions and results are given in Section~\ref{sec:prelims}.
	Section~\ref{sec:algorithm} defines disruption-free decompositions and the incompatibility number and provides the direct-access algorithm.
	Section~\ref{sct:lowerdirect} proves the hardness of direct access for self-join free join queries assuming the hardness of set-disjointness, while Section~\ref{sec:lowersetdisjointness} proves the hardness of set-disjointness assuming the more established Zero-$k$-Clique Conjecture.
	Section~\ref{sec:selfjoins} proves the equivalence in direct-access complexity between queries with and without self-joins.
	Then, Section~\ref{sec:together} puts the previous sections together, and summarizes tight complexity bounds for join queries.
	Relaxed order requirements and queries with projections are discussed in Section~\ref{sec:relaxations}.
	Finally, Section~\ref{sec:enumeration} proves the enumeration lower bounds for cyclic queries, assuming the Zero-$k$-Clique Conjecture.

	\section{Preliminaries}\label{sec:prelims}
	
	\subsection{Databases and Queries}
	
	A \emph{join query} $Q$ is an expression of the form
	$Q(\vec{u}) \datarule R_1(\vec{v}_1), \ldots, R_\ell(\vec{v}_n)$,
	where each $R_i$ is a relation symbol, $\vec{v}_1, \ldots, \vec{v}_n$ are tuples of variables, and $\vec{u}$ is a tuple of all variables in $\bigcup_{i=1}^{n}\vec{v}_i$.
	Each $R_i(\vec{v}_i)$ is called an \emph{atom} of $Q$. We also denote by $\var(Q)$ the set $\vec{u}$ of all variables appearing in $Q$.
	A query is called \emph{self-join free} if no relation symbol appears in two different atoms.
	If the same relation symbol appears in two different atoms, $R_i=R_j$, then $\vec{v}_i$ and $\vec{v}_j$ must have the same arity.
	A \emph{conjunctive query} (or, for short, a \emph{query}) is defined as a join query, with the exception that $\vec{u}$ may contain any subset of the query variables. We say that the query variables that do not appear in $\vec{u}$ are \emph{projected}. Defined this way, join queries are simply conjunctive queries without projections.
	
	The input to a query $Q$ is a \emph{database} $D$ which assigns every relation symbol $R_i$ in the query with a relation $R_i^D$: a finite set of tuples of constants, each tuple having the same arity as $\vec{v}_i$.
	The \emph{domain} $\dom(D)$ of $D$ is the set of all constants appearing in tuples in the relations of $D$.
	The \emph{size} $|D|$ of $D$ is defined as the total number of tuples in all of its relations. (More generally, $|E|$ denotes the cardinality of a set $E$.)
	An \emph{answer} to a query $Q$ over a database $D$ is determined by a mapping $\mu$ from the variables of $Q$ to the constants of $D$ such that $\mu(\vec{v}_i)\in R_i^D$ for every atom $R_i(\vec{v}_i)$ in $Q$. The answer is then $\mu(\vec{u})$. The set of all answers to $Q$ over $D$ is denoted $Q(D)$.
	
	A \emph{lexicographic order} of a join query $Q$ is specified by a permutation $L$ of the query variables. We assume databases come with a linear  order on their constants. Then, $L$ defines a linear order over $Q(D)$: the order between two distinct answers is defined to be the order between their assignments to the first variable in $L$ on which their assignments differ. 
	
	\subsection{Query Answering Tasks}
	
	The problems we consider are typically defined by a query $Q$, and the input is a database $D$ for $Q$. We always work in the model of \emph{data complexity}, meaning the size of the query is considered constant. In particular, the $O$-notation may hide constant factors that depend on $Q$.
	We consider three types of tasks:
	\begin{itemize}
		\item \emph{Testing} is the task where, after a preprocessing phase, the user can specify a tuple $\vec{a}$ of constants, and one has to determine whether it is a query answer (i.e., whether $\vec{a}\in Q(D)$).
		\item \emph{Enumeration} is the task of listing all query answers in $Q(D)$. We measure the efficiency of an enumeration algorithm with two parameters: the time before the first answer, which we call \emph{preprocessing time}, and the time between successive answers, which we call \emph{delay}.
		\item \emph{Direct-access} in lexicographic order is a task defined by a query $Q$ and a permutation $L$ of its free variables. After a preprocessing phase, the user can specify an index $j$ and expects the $j$th answer in $Q(D)$ according to the lexicographic order corresponding to $L$ or an out-of-bounds error if there are less than $j$ answers. We call the time it takes to provide an answer given an index the \emph{access time}.
	\end{itemize}
	
	An \emph{exact reduction} from a query $Q_1$ to a query $Q_2$ is a mapping computable in linear
	time which associates to each database $D_1$ for $Q_1$ a database $D_2$ for $Q_2$ such that
	there is a bijection $\tau$ from $Q_2(D_2)$ to $Q_1(D_1)$, and $\tau$ can be computed in constant time.
	Consider two queries $Q_1$ and $Q_2$ with the same free variables. An exact reduction via a bijection $\tau$ from $Q_1$ to $Q_2$ is called \emph{lex-preserving} if, for every permutation $L$ of the free variables, and for every pair $(a_1,a_2)$ of answers in $Q_2(D_2)$, we have that $a_1 \prec a_2$ if and only if $\tau(a_1) \prec \tau(a_2)$, where $\prec$ represents the lexicographic order specified by $L$.
	Note that if there is a lex-preserving exact reduction from $Q_1$ to $Q_2$, and $Q_2$ has a direct access algorithm for some lexicographic order $L$ with preprocessing time in $\Omega(|D_2|)$, then $Q_1$ has a direct access algorithm for $L$ with the same preprocessing and access time complexities.
	
	\subsection{Hypergraphs}
	
	A \emph{hypergraph} $H=(V,E)$ consists of a finite set $V$ of vertices and a set $E$ of edges, i.e., subsets of $V$. Given a set $S\subseteq V$, the hypergraph $H[S]$ induced by $S$ is $(S, E_S)$ with $E_S = \{e \cap S\mid e \in E\}$. A super-hypergraph $H'$ of $H$ is a hypergraph $(V, E')$ such that $E\subseteq E'$. By $N_H(v)$ we denote the set of neighbors of a vertex $v$ in $H$, i.e., $N_H(v) := \{ u\in V\mid u\ne v, \exists e\in E: u,v\in e\}$. For $S\subseteq V$, we define $N_H(S) :=\bigcup_{v\in S}N_H(v) \setminus S$. In this notation we sometimes leave out the subscript $H$ when it is clear from the context.
	
	A \emph{walk} between vertices $s$ and $t$ in a hypergraph $H= (V,E)$ is a sequence $s=v_1, \ldots, v_\ell=t$ of vertices such that for every $i\in [\ell-1]$ we have that there is an edge $e\in E$ with $v_iv_{i+1}\in E$. Here we allow $s=t$ in which case the walk consists of the single vertex $v_1=s=t$. The set $V^s$ of vertices reachable from $s$ in $H$ is defined to consist of all vertices $t$ such that there is a walk between $s$ and $t$ in $H$. Note that $s$ itself is reachable from $s$, so in any case $s\in V^s$. The \emph{connected component} of $s$ is defined as $H[V^s]$.
	
	A hypergraph $H$ is called \emph{acyclic} if we can eliminate all of its vertices by iteratively applying the following two rules\footnote{This is called the GYO-algorithm, see e.g.~Algorithm~6.4.4 in the standard textbook~\cite{AbiteboulHV95}.}:
	\begin{itemize}
		\item if there is an edge $e\in E$ that is completely contained in another edge $e'\in E$, delete $e$ from $H$, and
		\item if there is a vertex $v\in V$ that is contained in a single edge $e\in E$, delete $v$ from $H$, i.e., delete $v$ from $V$ and from $e$.
	\end{itemize}
	An order in which the vertices can be eliminated in the above procedure is called an \emph{elimination order} for $H$.
	Note that it might be possible to apply the above rules on several vertices or edges at the same moment. In that case, the order in which they are applied does not change the final result of the process. 
	
	Given a hypergraph and a permutation of its vertices, a \emph{disruptive trio} consists of three vertices $v_1,v_2,v_3$ such that $v_3$ appears after $v_1$ and $v_2$, the vertices $v_1$ and $v_2$ are not neighbors, but $v_3$ is a neighbor to both $v_1$ and $v_2$.
	A permuation $L$ of the vertices of a hypergraph $H$ is the reverse of an elimination order for $H$ if and only if $H$ is acyclic and $L$ contains no disruptive trio with $H$~\cite{bb:thesis}.
	
	A \emph{fractional edge cover} of $H=(V,E)$ is defined as a mapping $\mu: E\rightarrow [0,1]$ such that for every $v\in V$ we have $\sum_{e: v\in e} \mu(e) \ge 1$. The \emph{weight} of $\mu$ is defined as $\mu(E) := \sum_{e\in E} \mu(e)$. The \emph{fractional edge cover number} of $H$ is $\rho^*(H) := \min_\mu \mu(E)$ where the minimum is taken over all fractional edge covers of $H$. We remark that $\rho^*(H)$ and an optimal fractional edge cover of $H$ can be computed efficiently by linear programming. 
	
	Every join query $Q$ has an underlying hypergraph $H$, where the vertices of $H$ correspond to the variables of $Q$ and the edges of $H$ correspond to the variable scopes of the atoms of $Q$. We use $Q$ and $H$ interchangeably in our notation.
	
	\subsection{Known Algorithms}
	
	Here, we state some results that we will use in the remainder of this paper. All running time bounds are in the word-RAM model with $O(\log(n))$-bit words and unit-cost operations. Note that the values stored in the registers are bounded by $n^c$ for some constant $c$, and this model allows sorting $m$ such values in time $O(c(m+n))$ using radix sort, see e.g.~\cite[Section 6.3]{CormenLR20}.
	
	\begin{theorem}[\cite{CarmeliTGKR20}]\label{thm:linearcase}
		If an acyclic join query $Q$ and an ordering $L$ of its variables have no disruptive trios, then lexicographic direct access for $Q$ and $L$ can be solved with $O(|D|)$ preprocessing and $O(\log(|D|))$ access time.
	\end{theorem}
	
	A celebrated result by Atserias, Marx and Grohe~\cite{AtseriasGM13} shows that join queries of fractional edge cover number $k$, can, on any database~$D$, result in query results of size at most $O(|D|^k)$, and this bound is tight for every query for some databases. The upper bound is made algorithmic by so-called worst-case optimal join algorithms.
	
	\begin{theorem}[\cite{NgoPRR12,Veldhuizen14,NgoPRR18}]\label{thm:worst-case-joins}
		There is an algorithm that for every join query $Q$ of fractional edge cover number $\rho^*(Q)$, given a database $D$, computes $Q(D)$ in time $O(|D|^{\rho^*(Q)})$.
	\end{theorem}
	
	\subsection{Fine-Grained Complexity}

	Fine-grained complexity theory aims to find the exact exponent of the best possible algorithm for any problem, see e.g.~\cite{williams2018some} for a recent survey. Since unconditional lower bounds of this form are currently far out of reach, fine-grained complexity provides conditional lower bounds that hold when assuming a conjecture about some central, well-studied problem. 
	
	Some important reductions and algorithms in fine-grained complexity are randomized, i.e., algorithms are allowed to make random choices in their run and may return wrong results with a certain probability, see e.g.~\cite{mitzenmacher2017probability} for an introduction. Throughout the paper, when we write ``randomized algorithm'' we always mean a randomized algorithm with success probability at least $2/3$. It is well-known that the success probability can be boosted to any $1-\delta$ by repeating the algorithm $O(\log(1/\delta))$ times and returning the majority result. In particular, we can assume to have success probability at least $1-1/n^{10}$, at the cost of only a factor $O(\log(n))$. Our reductions typically worsen the success probability by some amount, but using boosting it can be improved back to $1-1/n^{10}$; we do not make this explicit in our proofs.
	
	We stress that randomization is only used in our hardness reductions, while all our algorithmic results are deterministic.

	We will base our lower bounds on the following problem which is defined for every fixed constant $k\in \mathbb{N}, k \ge 3$.
	
	\begin{definition}
		In the \pb{Zero-$k$-Clique} problem, given an $n$-node graph $G=(V,E)$ with edge-weights $w:E\rightarrow \{-n^c,\ldots, n^c\}$ for some constant $c\ge 1$, the task is to decide whether there are vertices $v_1,\ldots,v_k \in V$ such that they form a $k$-clique (i.e.~$\{v_i,v_j\} \in E$ for all $1 \le i < j \le k$) and their total edge-weight is 0 (i.e.~$\sum_{1 \le i < j \le k} w(v_i,v_j) = 0$). In this case we say that $v_1,\ldots,v_k$ is a zero-clique.
	\end{definition}
	
	We remark that by hashing techniques we can assume $c \le 10k$, see e.g.~\cite[Lemma~3.2]{AbboudBDN18}. Due to this bound, we may assume that all numbers that we encounter in the remainder of this paper have bit-size $O(\log(n))$ and thus fit into a constant number of memory cells. As a further consequence, all arithmetic operations on all numbers can be done in constant time in the RAM model. We thus tacitly ignore the size of numbers and the cost of operations on them in the remainder of this paper.
	
	The following conjectures have been used in several places, see e.g.~\cite{LincolnWW18,AbboudBDN18,BringmannGMW20,AbboudWW14,BackursDT16,BackursT17}.
	The first conjecture is for a fixed $k$, the second postulates hardness for all $k$.

	\begin{conjecture}[Zero-$k$-Clique Conjecture]\label{con:zeroclique}
		For no constant $\varepsilon >0$ \pb{Zero-$k$-Clique} has a randomized algorithm running in time $O(n^{k-\varepsilon})$.
	\end{conjecture}
	
	\begin{conjecture}[Zero-Clique Conjecture]
		For every $k \ge 3$ the Zero-$k$-Clique Conjecture is true.
	\end{conjecture}
	
	It is known that the \pb{Zero-$3$-Clique} Conjecture, also called Zero-Triangle Conjecture, is implied by two other famous conjectures: the 3SUM Conjecture~\cite{WilliamsW13} and the APSP Conjecture~\cite{WilliamsW18}. Since we do not use these conjectures directly in this paper, we do not formulate them here and refer the interested reader to the survey~\cite{williams2018some}.
	
	We remark that instead of \pb{Zero-$k$-Clique} some references work with the \pb{\emph{Exact-Weight}-$k$-Clique} problem, where we are additionally given a target weight $t$ and want to find a $k$-clique of weight $t$. Both problems are known to have the same time complexity up to constant factors, see e.g.~\cite{AbboudBDN18}.
	
	A related problem is \pb{\emph{Min}-$k$-Clique},  where we are looking for the $k$-clique of minimum weight. The \pb{Min-$k$-Clique} Conjecture postulates that this problem also cannot be solved in time $O(n^{k-\eps})$. It is known that the \pb{Min-$k$-Clique} Conjecture implies the \pb{Zero-$k$-Clique} Conjecture, see e.g.~\cite{AbboudBDN18}. 
	
	As common in fine-grained complexity, we use $\softO$-notation for many of our results which is similar to $O$-notation but suppresses polylogarithmic factors. More precisely, we say that an algorithm runs in time $\softO(g(n))$ if there is a constant $c$ such that it runs in time $O(\log^c(n)\cdot g(n))$.
	
	\section{Disruption-Free Decompositions and the Direct-Access Algorithm}
	\label{sec:algorithm}
	
	In this section, we give an algorithm that, for every join query and desired lexicographic order, provides direct access to the query result on an input database. In particular, we propose to add new atoms to the query such that the resulting query has no disruptive trios with respect to the order. Then, any direct-access algorithm that assumes acyclicity and no disruptive-trios can be applied, provided we can compute a database for the new query that yields the same answers. 
	We show that the new query is essentially a generalized hypertree decomposition of optimal fractional hypertree width out of all decompositions with the required properties.
	This shows that the suggested solution here is the best we can achieve using a decomposition, but it does not mean we cannot do better using a different method---the latter question is studied in the later sections of this paper.

	\subsection{Disruption-Free Decompositions}
	
	We describe a process that iteratively eliminates disruptive trios in a query by adding new atoms.
	
	\begin{definition}[Disruption-Free Decomposition]\label{def:disruption-free}
		Let $Q$ be a join query and $L=(v_1,\ldots,v_\ell)$ be an ordering of its variables. Let $H_\ell$ be the hypergraph of $Q$, and for $i=\ell, \ldots, 1$ construct hypergraph $H_{i-1}$ from $H_i$ by adding an edge $e_i := \{v_i\}\cup\{v_j\mid j<i\text{ and }v_j\in N_{H_i}(v_i)\}$. The disruption-free decomposition of $Q$ for $L$ is then defined to be $H_0$.
	\end{definition}
	
	\begin{example}\label{ex:1}
		Consider the query $Q(v_1, v_2, v_3, v_4, v_5) \datarule R_1(v_1, v_5),$ $R_2(v_2, v_4), R_3(v_3, v_4), R_4(v_3,v_5)$  with the order $(v_1, v_2, v_3, v_4, v_5)$ of its variables. Its hypergraph is shown in Figure~\ref{fig:ex1}. In the first step of the construction of Definition~\ref{def:disruption-free}, we add the edge $\{v_5\}\cup N(v_5)= \{v_1, v_3, v_5\}$. Similarly, in the second step, we add $\{v_2, v_3, v_4\}$. For the third step, note that $N(v_3) = \{v_1, v_2, v_4, v_5\}$ due to the edges we have added before. Out of these neighbors, only $v_1$ and $v_2$ come before $v_3$ in the order. So we add the edge $\{v_1, v_2, v_3\}$. Finally, for $v_2$, we add the edge $\{v_1, v_2\}$ and for $v_1$ the singleton edge $\{v_1\}$.
	\end{example}
	
	\begin{figure}
		\begin{tikzpicture}[scale=.8]
			\path (0,0) node[draw,circle,dashed] (x1) {$v_1$}
			(0,4) node (x5) {$v_5$}
			(2,2) node (x3) {$v_3$}
			(4,4) node (x4) {$v_4$}
			(4,0) node (x2) {$v_2$};
			\draw (x1) -- (x5);
			\draw (x5) -- (x3);
			\draw (x3) -- (x4);
			\draw (x4) -- (x2);
			\node[ellipse, draw, dashed, fit=(x1) (x3) (x5),inner sep=1mm](FIt1) {};
			\node[ellipse, draw, dashed, fit=(x2) (x3) (x4),inner sep=1mm](FIt2) {};
			\node[ellipse, draw, dashed, fit=(x1) (x3) (x2),inner sep=1mm](FIt3) {};
			\draw[dashed] (x1) -- (x2);
		\end{tikzpicture}
		\caption{The hypergraph of Example~\ref{ex:1}. The original (hyper)graph of the query is drawn in full edges. The edges that we add in the construction are dashed.}\label{fig:ex1}
	\end{figure}
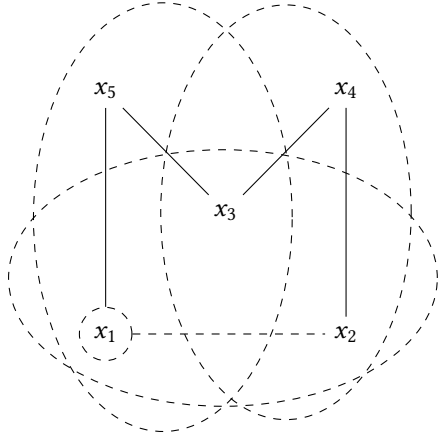
	
	\begin{proposition}
		The disruption-free decomposition of a join query $Q$ for $L$ is an acyclic super-hypergraph $H_0$ of the hypergraph of $Q$ without any disruptive trios with respect to $L$.
	\end{proposition}
	\begin{proof}
		Let $L=(v_1,\ldots,v_\ell)$. We also use the notation $e_i$ as defined in Definition~\ref{def:disruption-free}.
		To prove that~$H_0$ is acyclic, we will show that the reverse of $L$ is an elimination order.
		First, set $i=\ell$. Since~$e_i$ contains all neighbors of $v_i$ (that were not yet eliminated), we have that~$e_i$ contains all other edges containing $v_i$, and we can eliminate all of these edges. Afterwards, $v_i$ appears only in $e_i$, so we can eliminate $v_i$.
		Note that we did not eliminate any edges $e_j$ with $j<i$ since all edges we eliminated contain $v_i$ (and $e_j$ does not contain $v_i$ by construction when $j<i$).
		Thus, we can repeat these steps with $i-1$ instead of $i$ until all vertices are eliminated. It follows that, as claimed, the reverse of $L$ is an elimination order for $H_0$, and thus $H_0$ is indeed acyclic.
		
		We next show that $H_0$ has no disruptive trio with respect to $L$. By way of contradiction, assume $H_0$ has a disruptive trio $v_i,v_j,v_k$ with $i<j<k$. Then, by definition, in $H_0$ the vertices $v_i$ and $v_j$ are both neighbors of $v_k$, but $v_i$ and $v_j$ are not neighbors. In the construction of $H_0$ from Definition~\ref{def:disruption-free}, no edge containing $v_k$ is added after introducing $e_k$, and $e_k$ does not introduce new neighbors to~$v_k$. We conclude that $v_i$ and $v_j$ were already neighbors of $v_k$ right before introducing $e_k$. But then $\{v_i,v_j\}\subseteq e_k$, and thus $v_i$ and $v_j$ are neighbors in $H_0$. This contradicts the assumption that $v_i$, $v_j$, $v_k$ forms a disruptive trio.
	\end{proof}
	
	It will be useful in the remainder of this section to have a non-iterative definition of disruption-free decompositions.
	Let $S_i$ be the vertices in the connected component of $v_i$ in $Q[\{v_i,\ldots,v_\ell\}]$. In particular, we have that $v_i\in S_i$.
	
	\begin{lemma}\label{lem:alternativecharacterizationdecomposition}
		The edges introduced in Definition~\ref{def:disruption-free} are $e_i
		=\{v_i\}\cup\{v_j\mid j<i\text{ and }v_j\in N_Q(S_i)\}$ for $i\in\{1,\ldots,\ell\}$.
	\end{lemma}
	\begin{example}
		Let us illustrate Lemma~\ref{lem:alternativecharacterizationdecomposition} with the query of Example~\ref{ex:1}: For the variable $v_5$ we have $S_5 = \{v_5\}$. Consequently, $e_5 = \{v_5\} \cup N_Q(v_5)= \{v_1, v_3, v_5\}$ which is exactly the edge that is added for $v_5$. The case for $v_4$ is analogous. For $v_3$ we have $S_3 = \{v_3, v_4, v_5\}$. Consequently $e_3 := \{v_3\} \cup N_Q(\{v_3, v_4, v_5\}) = \{v_1, v_2, v_3\}$. We have $S_2 = \{v_2, v_3, v_4, v_5\}$, so we add the edge $e_2 = \{v_1, v_2\}$. For $v_1$ we add again the singleton edge $\{v_1\}$.
	\end{example}
	
	\begin{proof}[Proof of Lemma~\ref{lem:alternativecharacterizationdecomposition}]
		Let $e_i'=\{v_i\}\cup\{v_j\mid j<i\text{ and }v_j\in N_Q(S_i)\}$. We prove that $e_i=e_i'$ by induction on decreasing $i$. The claim holds by definition for $e_\ell$ since $S_\ell=\{v_\ell\}$, and $H_\ell$ is the hypergraph of $Q$, and so $N_Q(S_i) = N_{H_\ell}(v_\ell)$.
		
		For the induction step, we first show $e_i\subseteq e_i'$.
		Let $v_j\in e_i$. If $j=i$, we know that $v_j=v_i\in e_i'$ by definition. Otherwise, $j<i$ and $v_j\in N_{H_i}(v_i)$. If $v_j\in N_Q(v_i)$, we have that $v_j\in N_Q(S_i)$ because $v_i\in S_i$, and so $v_j\in e_i'$.
		Otherwise, the reason $v_j$ is a neighbor of $v_i$ in $H_i$ is because there is some $k>i>j$ such that $\{v_i,v_j\}\subseteq e_k$. From the induction hypothesis, we have that $e_k'= e_k$, and it follows that $v_i, v_j\in N_{Q}(S_k)$. Since $v_i\in N_{Q}(S_k)$, we have that $S_k\subseteq S_i$. It follows that $v_j\in N_{Q}(S_k)\subseteq N_{Q}(S_i)$ and thus $v_j\in e_i'$. Since in all cases $v_j\in e_i'$, we get that $e_i\subseteq e_i'$.
		
		It remains to show that $e_i' \subseteq e_i$. 
		To this end, let $v_j\in e_i'$. If $j=i$, we know that $v_j=v_i\in e_i$ by definition.
		Otherwise, $j<i$ and $v_j\in N_Q(S_i)$.
		If $v_j\in N_Q(v_i)$, then also $v_j\in N_{H_i}(v_i)$, and so $v_j\in e_i$.
		Otherwise, if we remove $v_i$ from $S_i$, it is split into connected components, one of them containing a neighbor of $v_j$.
		Consider the smallest variable $v_k$ in such a connected component. Then, this connected component is $S_k$, and we have that $v_i,v_j\in N_Q(S_k)$ where $k>i>j$.
		It follows that $\{v_i, v_j\}\subseteq e_k'$, and from the induction hypothesis we get that $\{v_i, v_j\}\subseteq e_k$. Thus,~$v_j$ is a neighbor of~$v_i$ when introducing $e_i$, and so $v_j\in e_i$. Consequently, $e_i'\subseteq e_i$.
	\end{proof}
	
	\subsection{The Direct-Access Algorithm}
	
	The idea of our direct access algorithm is to use the disruption-free decomposition. More precisely, we define a new query $Q'$ from $Q$ such that the hypergraph of $Q'$ is the disruption-free decomposition of $Q$. Then, given an input database $D$ for $Q$, we compute a new database $D'$ for $Q'$ such that $Q(D) = Q'(D')$. Since $Q'$ has no disruptive trio, we can then use the algorithm from~\cite{CarmeliTGKR20} on $Q'$ and $D'$ to allow direct access.
	A key component to making this approach efficient is the efficient computation of $D'$. To measure its complexity, we introduce the following notion.
	
	\begin{definition}[Incompatibility Number]
		Let $Q$ be a join query with hypergraph $H$ and let $L=(v_1,\ldots,v_\ell)$ be an ordering of its variables.
		Let $\rho^*(H[e_i])$ with $i\in [\ell]$ be the fractional edge cover number of $H[e_i]$, where $e_i$ is as defined in Definition~\ref{def:disruption-free}.
		We call $\max_{i=1,\ldots,\ell}\rho^*(H[e_i])$ the \emph{incompatibility number} of $Q$ and $L$.
	\end{definition}
	Note that we assume that queries have at least one atom, so the incompatibility number of any query and any order is at least $1$.
	The incompatibility number can also be seen as the \emph{fractional width} of the disruption-free decomposition, and we can show that this decomposition has the minimum fractional width out of all decompositions with the properties we need, see Section~\ref{sec:decompositions} for details.
	We now show that the incompatibility number lets us state an upper bound for the computation of a database $D'$ with the properties claimed above.
	
	\begin{theorem}\label{thm:cover-preprocess}
		Given a join query $Q$ and an ordering $L$ of its variables with incompatibility number $\iota$, lexicographic direct access with respect to $L$ can be achieved with $O(|D|^\iota)$ preprocessing and logarithmic access time.
	\end{theorem}
	\begin{proof}
		We construct $Q'$ by adding for every edge $e_i$ an atom $R_i$ whose variables are those in $e_i$. We compute a relation $R_i^{D'}$ as follows: let $\mu$ be a fractional edge cover of $H[e_i]$ of weight $\rho^*(H[e_i])\le \iota$. Let $\bar e_1, \ldots, \bar e_r$ be the edges of $H[e_i]$ that have positive weight in $\mu$. For every $\bar e_j$ there is an edge $e_j'$ of $H$ such that $\bar e_j = e_j'\cap e_i$. Let $\bar R_j(\bar X_j)$ be the projection to $e_i$ of an atom corresponding to $e_j'$. Let $D^*$ be the extension of $D$ that contains the corresponding projected relations for the $\bar R_j(\bar X_j)$. We set $R_i^{D'} := Q_i(D^*)$ where $Q_i(e_i) \datarule \bar R_1(X_1), \ldots, \bar R_r(X_r)$.
		Then clearly for all tuples $\vec t$ in~$Q(D)$, the tuple we get by projecting $\vec t$ to $e_i$ lies in $R_i^{D'}$.
		As a consequence, we can construct $D'$ by adding all relations $R_i^{D'}$ to $D$ to get $Q(D) = Q'(D')$. Moreover, the hypergraph of $Q'$ is the disruption-free decomposition of the hypergraph of $Q$, so we can apply the algorithm from~\cite{CarmeliTGKR20} for direct access.
		
		It remains to show that all $R_i^{D'}$ can be constructed in time $O(|D|^\iota)$. To this end, consider the join query $Q_i(e_i)$ from before. By definition, its variable set is $e_i$, and $\mu$ is a fractional edge cover of weight $\rho^*(H[e_i])\le \iota$. Thus, we can use a worst-case optimal join algorithm from~Theorem~\ref{thm:worst-case-joins} to compute $R_i^{D'}$ in time $O(|D|^\iota)$.
	\end{proof}

	\subsection{Disruption-Free Decompositions and Fractional Hypertree Width}\label{sec:decompositions}
	
	In this section, we will relate disruption-free decompositions to fractional hypertree decompositions and the incompatibility number to fractional width.
	
	Let $H=(V,E)$ be a hypergraph. A \emph{hypertree decomposition} $\mathcal{D}$ of $H$ is defined to be an acyclic hypergraph $(V,B)$ such that for every edge $e\in E$ there is a $b\in B$ with $e\subseteq b$.\footnote{Hypertree decompositions appear in previous work under several names, such as \emph{generalized hypertree decompositions} and \emph{fractional hypertree decompositions} depending on the way the bags get covered. Usually, the definitions are more involved as they also contain a tree structure with the bags as nodes and an edge cover of the bags. This is a simplified definition containing only what we need here.} The sets $b\in B$ are called the \emph{bags} of the decomposition.
	The \emph{fractional width} of $\mathcal D$ is defined as $\max_{b\in B} \rho^*(H[b])$.
	The \emph{fractional hypertree width} of $H$ is defined to be the minimal fractional width of any hypertree decomposition of $H$~\cite{GroheM14}.
	
	Note that the fractional edge cover number is monotone in the sense that $\rho^*(H[b])\le \rho^*(H[b'])$ whenever $b\subseteq b'$, and so the fractional width of a hypertree decomposition is determined by its maximal edges with respect to set inclusion.
	Next, we observe that the incompatibility number is the fractional width of the disruption-free decomposition seen as a hypertree decomposition.
	
	\begin{proposition}
		Let $Q$ be a join query and $L$ be an ordering of its variables. The disruption-free decomposition of $Q$ and $L$ is a hypertree decomposition of~$Q$, and its fractional width is the incompatibility number of $Q$ and $L$.
	\end{proposition}
	\begin{proof}
		Consider an edge $e$ of the disruption-free decomposition $H_0$. First, we claim that $e$ is contained in some edge $e_i$ from Definition~\ref{def:disruption-free}. Let $v_i$ be the latest vertex (according to $L$) in an edge $e$ of $H_0$. Then, $e$ could not have been introduced after the introduction of $e_i$. Thus, either $e=e_i$ or $e$ existed right before the introduction of $e_i$, and then $e\subseteq e_i$ by the definition of $e_i$.
		Due to the monotonicity of the fractional edge cover number, $\rho^*(H[e])\le \rho^*(H[e_i])$, and this is bounded by the incompatibility number by its definition. We also have that $\iota=\rho^*(H[e_j])$ for some $j$, and thus $\iota$ is exactly the fractional width of $H_0$.
	\end{proof}
	
	\begin{observation}\label{observ:incompatibility-width}
		Let $Q$ be a join query and $L$ be an ordering of its variables. Then the incompatibility number of $Q$ and $L$ is at least the fractional hypertree width of $Q$.
	\end{observation}
	
	Of course there are other hypertree decompositions $\mathcal D$ of $Q$ that we could have used for a direct access algorithm. The only property that we need is that $\mathcal D$ has no disruptive trio for $L$. If this is the case, then inspection of the proof of Theorem~\ref{thm:cover-preprocess} shows that we get an alternative algorithm whose preprocessing depends on the fractional width of $\mathcal D$.
	We will see next that this approach cannot yield a better running time than Theorem~\ref{thm:cover-preprocess}. To this end,
	we prove that any alternative decomposition contains the disruption-free decomposition in a way.
	
	\begin{lemma}\label{lem:optimaldecomposition}
		Let $Q$ be a join query, $L$ be an ordering of its variables, and $\mathcal D$ be a hypertree decomposition of $Q$ without disruptive trios with respect to $L$. Then, for every edge $e$ of the disruption-free decomposition of $Q$, there exists a bag $b$ of $\mathcal D$ such that $e\subseteq b$.
	\end{lemma}
	\begin{proof}
		Let $H$ be the hypergraph of $Q$.
		Assume by way of contradiction that there is an edge of the disruption-free decomposition $H_0$ of $Q$ for $L$ that is not contained in any bag of $\mathcal D$.
		Since $\mathcal D$ is a hypertree decomposition of $H$, for every edge $e$ of $H$, the decomposition $\mathcal D$ contains, by definition, a bag $b$ such that $e\subseteq b$. Thus, the edge of $H_0$ not contained in any bag of $\mathcal D$ must be of the form $P\cup\{v\}$ where
		$P$ consists of the preceding neighbors of $v$ at the moment of creation of the edge; here, when we say `preceding', we mean the neighbors that come before $v$ in $L$.
		Consider the non-covered edge $e$ with the largest $v$.
		We show next that every pair of vertices in $e$ appears in a common bag in $\mathcal D$.
		It is known that acyclic hypergraphs are \emph{conformal}~\cite{Brault-Baron16}; that is, any set of pairwise neighbors must be contained in an edge.
		Thus, $\mathcal D$ has a bag containing $e$, which is a contradiction.
		
		Since $e$ is the largest non-covered edge, we know that $\mathcal D$ contains, for all edges $e'$ that were present at the moment of creation of $e$, a bag $b$ with $e'\subseteq b$. Thus, for every vertex $v'$ in $P$, there is a bag of $\mathcal D$ that contains both $v$ and $v'$.
		Now consider $v_1,v_2\in P$.
		If $v_1,v_2$ are not in a common bag $b'$ of $\mathcal D$, then $v_1,v_2,v$ is a disruptive trio of $\mathcal D$, which is a contradiction to the conditions of the lemma. So $v_1, v_2$ must appear in a common bag.
		It follows that the vertices of $e$ are all pairwise neighbors in $\mathcal D$ as we wanted to show.
	\end{proof}

	Due to the monotonicity of the fractional edge cover number, we directly get that the disruption-free decomposition is optimal in the following sense.
	
	\begin{proposition}\label{prop:incompatibility-min-width}
		Let $Q$ be a join query and $L$ be an ordering of its variables. The disruption-free decomposition of $Q$ and $L$ has the minimal fractional width taken over all hypertree decompositions of~$Q$ that have no disruptive trio with respect to $L$.
	\end{proposition}
	
	Note that, in general, finding an optimal fractional hypertree decomposition is known to be hard~\cite{GottlobLPR21}. However, in our case, decompositions are only useful if they eliminate disruptive trios. If we restrict the decompositions in that way, it is much easier to find a good decomposition compared to the general case: since the optimal decomposition is induced by the order, we get it in polynomial time using the procedure of Definition~\ref{def:disruption-free}.
	
	\section{Set-Disjointness-Based Hardness for Self-join free Direct Access}\label{sct:lowerdirect}
	
	In this section, we show lower bounds for lexicographically ranked direct access for all self-join free queries and all variable orders. We first reduce general direct access queries to the special case of star queries, which we introduce in Section~\ref{sec:lowerdirect_star}.
	Then we show that lower bounds for direct access to star queries follow from lower bounds for \pb{$k$-Set-Disjointness}, see Section~\ref{sec:lowerdirect_setdisj}. 
	Later, in Section~\ref{sec:lowersetdisjointness} we prove a lower bound for \pb{$k$-Set-Disjointness} based on the \pb{Zero-Clique} Conjecture, and in Section~\ref{sec:selfjoins} we remove the assumption that $Q$ is self-join free.

	\subsection{From \boldmath$k$-Star Queries to Direct Access} \label{sec:lowerdirect_star}
	
	A crucial role in our reduction is played by the $k$-star query $\qstar_k$:
	\begin{align*}
		\qstar_k(x_1, \ldots, x_k,z) \datarule R_1(x_1, z), \ldots, R_k(x_k, z).
	\end{align*}
	We say that a variable order $L$ is \emph{bad for $\qstar_k$} if $z$ is the last variable in $L$. We will also use the variant $\qstarp_k$ of this query in which the variable $z$ is projected away:
	\begin{align*}
		\qstarp_k(x_1, \ldots, x_k) \datarule R_1(x_1, z), \ldots, R_k(x_k, z).
	\end{align*}

	In the case that the query is acyclic, we can show that the incompatibility number is always an integer, and then the reduction is simpler. We thus start with this case as a warm-up for the general construction.
	
	\begin{lemma}\label{thm:acyclic-fractionalstar}
		Let $Q$ be an acyclic self-join-free join query and $L$ be an ordering of its variables with incompatibility number $\iota$. Then $\iota$ is an integer, and if $Q$ has a lexicographic direct access algorithm according to $L$ with preprocessing time $p(|D|)=\Omega(|D|)$ and access time $r(|D|)$, then so has $\qstar_\iota$ for a bad lexicographic order.
	\end{lemma}
	\begin{proof}
		We use the notation $S_i$ and $e_i$ for $i\in [\ell]$ as in Lemma~\ref{lem:alternativecharacterizationdecomposition}. Recall that $S_i$ is the vertex set of the connected component of $v_i$ in $H[\{v_i, \ldots, v_\ell\}]$ where $H$ is the hypergraph corresponding to $Q$.
		Since every induced subhypergraph of an acyclic hypergraph is also acyclic (this is well known and follows easily from the definition of acyclicity by elimination orders), we have that $H[e_i]$ is acyclic for all $i\in\{1,\ldots,\ell\}$.
		Let $r$ be such that $\rho^*(H[e_r])= \iota = \max_{i\in [\ell]}\left(\rho^*(H(e_i))\right)$.
		In acyclic hypergraphs, the fractional edge cover number is an integer, and there is an \emph{independent set} of this size~\cite{DurandM14}, i.e., a subset of vertices such that every edge contains at most one vertex of this subset.
		Let $\{u_1,\ldots, u_\iota\}$ be such an independent set in $H[e_r]$.
		
		We construct a database $D$ for $Q$, given an input database $D^\star$ for $\qstar_\iota$. 
		We first give roles to the  variables of~$Q$ that correspond to the variables of $\qstar_\iota$.
		First, we assign for every $i\in [\iota]$ the role $x_i$ to every $u_i$. We next add the role $z$ for every $v\in S_r$. Note that when doing so, the variable $v_r$, the only one in both $e_r$ and $S_r$, may have both the role $x_i$ for some $i\in[\iota]$ and the role $z$, while the other variables have at most one role.
		
		To define the relations of $D$, consider an atom $R(\vec{w})$ of $Q$.
		Since the variables with the roles of the form $x_i$ are an independent set, the roles of variables in this atom cannot contain $x_i$ and $x_j$ with $i\neq j$, and so there is an atom $R_{j}(x_{j},z)$ in $Q^\star_\iota$ whose variables contain all roles played by the variables $\vec{w}$.
		We use $R_j^{D^\star}$ to define $R^D$.
		For every tuple $(c_j,c_z)\in R_j^{D^\star}$, we add a tuple $\vec t$ to $R^D$ as follows: for all $w_i$ in $\vec{w}$, if the variable $w_i$ has only the role $x_j$, then $t_i=c_j$; if it has only the role $z$, then $t_i=c_z$; if it has both roles $x_j$ and $z$, then $t_i=(c_j,c_z)$; if it has no role, then $t_i$ is set to a constant $\bot$.
		Note that $|D|=O(|D^\star|)$ and that the construction only requires linear time. 
		
		The construction yields a bijection $\tau$ between $Q(D)$ and $\qstar_k(D^\star)$ as follows:	if a variable with a single role $y$ is assigned a value $c$ in an answer $a\in Q(D)$, then $\tau(a)$ assigned $y$ with $c$; similarly, if a variable with two roles $x_j$, $z$ is assigned a value $(c_j,c_z)$, then $\tau(a)$ assigned $x_j$ with $c_j$ and $z$ with $c_z$.
		This bijection is well-defined because all variables that have the role $z$ must take the same value on the corresponding coordinate by construction, since $S_r$ is connected in $H$ and the construction assigns different coordinates with the role $z$ in the same atom with the same value.
		Moreover, for every atom $R_j(x_j, z)$ of $\qstar_\iota$, there is an atom of $Q$ containing (not necessarily different) variables $v_x$ and $v_z$ such that $v_x$ has the role $x_j$ and $v_z$ has the role $z$. By construction, it then follows that the variables $v_x$ and $v_z$ take values that are consistent with the relation $R_j^{D^\star}$. This directly gives the desired bijection.
		Note that computing $\tau(a)$ given $a \in Q(D)$ can be done in constant time.
		
		Next we claim that this construction yields a bad lexicographic order for $\qstar_\iota(D^\star)$. That is, ordering the solutions to $Q(D)$ by $L$ and then applying $\tau$ on each solution gives the solutions to $\qstar_\iota(D^\star)$ sorted in a bad lexicographic order. This is the case since the variables with the role $z$ appear in $L$ after all variables with a role of the form $x_i$ with some $i\in [\iota]$, with the exception of $v_r$ that may have two roles, in which case the construction assigns it the values corresponding to $z$ after those corresponding to $x_i$.
		
		Using this construction, a lexicographic direct access for $Q$ and $L$ can be used for lexicographic direct access for $\qstar_\iota$ and a bad lexicographic order.
	\end{proof}
	
	\begin{example}
		Let us illustrate the role assignment from the proof of Lemma~\ref{thm:acyclic-fractionalstar} with the query of Example~\ref{ex:1}. The incompatibility number here is $\iota=3$, and it is witnessed by the edge $e_3=\{v_1,v_2,v_3\}$. An independent set of $H[e_3]$ of size $3$ is $\{v_1,v_2,v_3\}$. Thus, the variables $v_1$, $v_2$ and $v_3$ are assigned the roles $x_1$, $x_2$, and $x_3$ respectively. In addition, $v_3$, $v_4$ and $v_5$ (these are the variables of $S_3$) are each assigned the role $z$. Using this role assignment, $Q^\star_3$ is reduced to $Q$.
	\end{example}
	
	In general, the incompatibility number may not be an integer; we will see such a case in Example~\ref{ex:fractionalIncompatible}. We will next show how we can generalize the construction from Lemma~\ref{thm:acyclic-fractionalstar} to also handle this case.
	
	\begin{lemma}\label{lem:fractionalstar}
		Let $Q$ be a self-join-free join query and $L=(v_1,\ldots,v_\ell)$ be an ordering of its variables. Let $\iota$ be the incompatibility number of $Q$ and $L$ and assume $\iota > 1$. If there is an $\eps > 0$ such that for all $\delta > 0$ there is a direct access algorithm for $Q$ and $L$ with preprocessing time $O(|D|^{\iota - \varepsilon})$ and access time $O(|D|^{\delta})$, then there is a $k\in \mathbb{N}$, $k>2$ and $\eps'>0$ such that for all $\delta'>0$ there is a direct access algorithm for $\qstar_k$ with respect to a bad ordering with preprocessing time $O(|D^\star|^{k-\varepsilon'})$ and access time $O(|D^\star|^{\delta'})$.
	\end{lemma}
	\begin{proof}
		We use a reduction similar to that in the proof of Lemma~\ref{thm:acyclic-fractionalstar}, except we use a \emph{fractional} independent set in order to handle the fractional incompatibility number.
		A \emph{fractional independent set} in a hypergraph $H=(V,E)$ is a mapping $\phi: V\rightarrow [0,1]$ such that for every $e\in E$ we have that $\sum_{v\in e}\phi(v) \le 1$. The weight of $\phi$ is defined to be $\phi(V)= \sum_{v \in V} \phi(v)$. The fractional independent set number of $H$ is defined as $\alpha^*(H) := \max_\phi \phi(V)$ where the maximum is taken over all fractional independent sets of $H$. 
		As already noted in~\cite{GroheM14}, using linear programming duality, we have for every hypergraph $H$ and every vertex set $S$ that $\alpha^*(H[S]) = \rho^*(H[S])$, see e.g.~\cite[Chapter 82.2]{Schrijver03} for details. We use the same notation on join queries, with the meaning of applying it to the underlying hypergraph.
		
		Let $r$ be such that $\rho^*(H[e_r])= \iota = \max_{i\in [\ell]}\left(\rho^*(H(e_i))\right)$. Then, we know that there is a fractional independent set $\phi: e_r\rightarrow [0,1]$ such that $\sum_{v\in e_r}\phi(v) = \iota$. Since $\phi$ is the solution of a linear program with integer coefficients and weights, all values $\phi(v)$ are rational numbers. Let $\lambda$ be the least common multiple of the denominators of $\{\phi(v)\mid v\in e_r\}$. Now define a new weight function $\phi':e_r\rightarrow \{0, \ldots, \lambda\}$ by $\phi'(v) := \lambda \phi(v)$. Let $k:= \lambda \iota$, and consider the query $\qstar_k$.
		
		As in the proof of Lemma~\ref{thm:acyclic-fractionalstar}, we again assign variables of $Q$ with roles that correspond to the variables of $\qstar_k$. The difference in this proof is that now every variable may play several roles of the form $x_i$.
		Indeed, every $v\in e_r$ takes $\phi'(v)$ roles out of the variables $x_1, \ldots, x_k$.
		More specifically, let $\{u_1,\ldots,u_m\}\subseteq e_r$ be the variables that $\phi'$ assigns a non-zero value. Then, for each $i\in[m]$, $u_i$ is assigned the roles $\{x_j|\sum_{t<i}{\phi'(u_t)} < j \leq \sum_{t\leq i}{\phi'(u_t)}\}$.
		Note that the number of roles we distributed so far is indeed $\sum_{v\in e_r} \phi'(v) = \lambda \sum_{v\in e_r} \phi_v = \lambda \iota = k$. We next add the role $z$ for every $v\in S_r$. Note that when doing so, the variable $v_r$, the only one in both $e_r$ and $S_r$, may have both roles of the form $x_i$ and the role $z$.  For an illustration of this role assignment, see Example~\ref{ex:fractionalIncompatible}.
		
		We now construct a database $D$ for $Q$ given an input database $D^\star$ for $\qstar_k$. For every variable $y$ of $\qstar_k$, denote by $\dom^\star(y)$ its active domain in $D^\star$. We fix all variables of $Q$ that have no role to a constant $\bot$, and we will ignore these variables in the rest of the proof, so we only need to consider variables in $e_r \cup S_r$. 
		For each variable $v\in e_r\cup S_r$ we define the domain as follows: let $y_1, \ldots, y_s$ be the roles of $v$ given in the order defined by $L$. Then the domain of $v$ is $\dom(v):=\dom^\star(y_1)\times \ldots \times \dom^\star(y_s)$. We order $\dom(v)$ lexicographically.
		To define the relations of $D$,
		let $x_{j_1}, \ldots, x_{j_t}$ be the roles of the form $x_i$ played by variables in $R(v_{i_1}, \ldots, v_{i_s})$. Then we compute the sub-join $Q_J(D^\star)$ where $Q_J(x_{j_1}, \ldots, x_{j_t},z) \datarule R_{j_1}(x_{j_1},z), \ldots,$ $R_{j_t}(x_{j_t}, z)$. Note that there is a function $\nu : Q_J(D^\star)\rightarrow \dom(v_{i_1}) \times \ldots \times \dom(v_{i_s})$ that consists of ``packing'' the values for the different variables of $\qstar_k$ into the variables of $Q$ according to the roles in the following sense: for every tuple $\vec t$, the function $\nu$ maps each variable $v_i$ with role set $\mathcal R$ to the tuple we get from $\vec t$ by deleting the coordinates not in $\mathcal R$. The relation $R^D$ is then simply $\nu(Q_J(D^\star))$.
		
		The construction yields a bijection $\tau$ between $Q(D)$ and $\qstar_k(D^\star)$: if a coordinate corresponding to a role $y$ is assigned a value $c$ in an answer $a\in Q(D)$, then $\tau(a)$ assigns $c$ to $y$. The same arguments given in the proof of Lemma~\ref{thm:acyclic-fractionalstar} show that $\tau$ is well defined and is indeed a bijection and that this construction yields a bad lexicographic order.
		Note that the bijection $\tau$ can again be computed in constant time.
		
		It remains to analyze the time complexity of the construction of $D$.
		We show that the size of $D$ and the time required for its construction are in $O(|D^\star|^\lambda)$.
		For every atom $A$ in $Q$, its variables in total have at most $\lambda$ roles of the form $x_i$. To see this, let $v_{i_1}, \ldots, v_{i_s}$ be the variables of $A$. Then the overall number of non-$z$-roles is
		$
		\sum_{j=1}^s \phi'(v_{i_j}) = \sum_{j=1}^s \lambda \phi(v_{i_j}) \le \lambda,
		$
		where the inequality comes from the fact that $\phi$ is a fractional independent set of $H$ and $\{v_{i_1}, \ldots, v_{i_s}\}$ is an edge of $H$. Hence, the join $Q_J(D^\star)$ involves at most $\lambda$ atoms of $\qstar_k$, and it can be computed in time $O(|D^\star|^\lambda)$.
		
		Assume that for some $\eps>0$ and all $\delta>0$ there is a direct access algorithm for $Q$ and $L$ with preprocessing time $O(|D|^{\iota - \varepsilon})$ and access time $O(|D|^{\delta})$.
		We set $\eps' := \min(\lambda
		\eps, \lambda (\iota-1))$. Note that indeed $\eps'> 0$ since  $\iota>1$. Using the facts that $|D|=O(|D^\star|^\lambda)$ and $k=\lambda\iota$, we get that the preprocessing is 
		$
		O(|D^\star|^\lambda + |D|^{\iota - \varepsilon})
		= O(|D^\star|^{\lambda\iota - \lambda(\iota-1)} + |D^\star|^{\lambda \iota - \lambda \varepsilon}) = O(|D^\star|^{k - \varepsilon'}).
		$
		Then, given $\delta'$, we set $\delta := \delta' / \lambda$. We get that the access time is $O(|D|^\delta) = O(|D^\star|^{\lambda \delta}) = O(|D^\star|^{\delta'})$.
	\end{proof}
	\begin{example}\label{ex:fractionalIncompatible}
		Let us illustrate the role assignment from the proof of Lemma~\ref{lem:fractionalstar}. Consider the query obtained from that of Example~\ref{ex:1} by adding the atoms $R_5(v_1,v_2),R_6(v_2,v_3),R_7(v_1,v_3)$. This query does not have disruptive trios, but it is cyclic. The edges $e_i$ introduced by the distruption-free decomposition are the same as in the previous example. The incompatibility number here is $\iota=\frac{3}{2}$.
		One witness for this incompatibility number is the edge $e_3=\{v_1,v_2,v_3\}$ using the fractional independent set of $H[e_3]$ that assigns each of its vertices a weight of $\frac{1}{2}$. As a result, we take $\lambda=2$, and the variables $v_1$, $v_2$ and $v_3$ are assigned the roles $\{x_1,x_2\}$, $\{x_3,x_4\}$, and $\{x_5,x_6\}$ respectively. In addition, $v_3$, $v_4$ and $v_5$ (these are the variables of $S_3$) are each assigned the role $z$. Using this role assignment, $Q^\star_6$ is reduced to $Q$.
		An alternative choice of witness for the incompatibility number is the edge $e_5=\{v_1,v_3,v_5\}$ using the fractional independent set of $H[e_5]$ that assigns each of its vertices a weight of $\frac{1}{2}$.
		Here the variables $v_1$, $v_3$ and $v_5$ are assigned the roles $\{x_1,x_2\}$, $\{x_3,x_4\}$, and $\{x_5,x_6,z\}$ respectively, while the variables $v_2$ and $v_4$ have no roles. This choice gives another way of reducing $Q^\star_6$ to $Q$.
	\end{example}

	\subsection{From Projected Stars to Star Queries} \label{sec:testing}
	
	\begin{proposition}\label{prop:testing}
		Let $L$ be a bad order on $x_1, \ldots, x_k,z$. If there is an algorithm that solves the direct access problem for $\qstar_k(x_1, \ldots, x_k, z)$ and $L$ with preprocessing time $p(|D|)$ and access time $r(|D|)$, then there is an algorithm for the testing problem for $\qstarp_k(x_1, \ldots, x_k)$ with preprocessing time $p(|D|)$ and query time $O(r(|D|)\log(|D|))$.
	\end{proposition}
	\begin{proof}
		If an answer $(a_1, \ldots, a_k)$ is in $\qstarp_k(D)$, then answers of the form $(a_1, \ldots, a_k, b)$ (for any $b$) form a contiguous sequence in the set $\qstar_k(D)$ ordered according to the $L$-lexicographic order, since the position of $z$ is the last variable in $L$. Thus, a simple binary search using the direct access algorithm allows to test whether such a tuple $(a_1, \ldots, a_k,b)\in \qstar_k(D)$ exists, using a logarithmic number of direct access calls.
	\end{proof}
	
	\subsection{From Set-Disjointness to Projected Stars} \label{sec:lowerdirect_setdisj}
	
	We observe (Lemma~\ref{lem:reductiontosetdisjointness}) that the testing problem for the projected star query $\qstarp_k$ is equivalent to the following problem:
	
	\begin{definition} \label{def:setdisjointness}
		In the \pb{$k$-Set-Disjointness} problem, we are given an instance $\I$ consisting of a universe $U$ and families $\mathcal{A}_1, \ldots, \mathcal{A}_k \subseteq 2^U$ of subsets of $U$. We denote the sets in family $\A_i$ by $S_{i,1},\ldots,S_{i,|\A_i|}$.
		The task is to preprocess $\I$ into a data structure that can answer queries of the following form: Given indices $j_1,\ldots,j_k$, decide whether $S_{1,j_1} \cap \ldots \cap S_{k,j_k} = \emptyset$.
		
		We denote the input size by $\|\I\| := \sum_i \sum_{S \in \A_i} |S|$.
	\end{definition}
	
	\begin{example}
		Since we will see several problems later on that have the same type of inputs, let us take some time for an example. Fix the domain $U:=\{1, \ldots, 5\}$ and say $k=3$. Let $\mathcal A_1$ consist of the sets $S_{1,1}:=\{1,3,5\}$ and $S_{1,2}:=\{1,2,4\}$. Let $\mathcal A_2$ consist of $S_{2,1}:=\{1,4\}$, $S_{2,2}:=\{2,4\}$, and $S_{2,3}:=\{1,2,3,4,5\}$. Finally, let $\mathcal A_3$ consist of $S_{3,1}:= \{3,4,5\}$ and $S_{3,2}:=\{4\}$. The size of the resulting instance $\I$ is then~$\|\I\| = 19$.
		
		On input $(2,3,2)$ we have to consider $S_{1,2}\cap S_{2,3}\cap S_{3,2} = \{4\}$, so an algorithm for \pb{$k$-Set-Disjointness} has to report that the queried set intersection is non-empty. For the input $(1,1,1)$, we have that $S_{1,1}\cap S_{2,1}\cap S_{3,1}= \emptyset$, so the correct answer is that the queried set intersection is empty.
	\end{example}
	
	\begin{lemma}\label{lem:reductiontosetdisjointness}
		If $\qstarp_k$ has a testing algorithm with preprocessing time $p(|D|) = \Omega(|D|)$ and access time $r(|D|)$, then \pb{$k$-Set-Disjointness} has an algorithm with preprocessing time $p(\|\I\|)$ and query time $r(\|\I\|)$.
	\end{lemma}
	\begin{proof}
		For each family $\A_i$, consider the relation \[R_i^D:= \{(j, v)\mid v \in S_{i, j}\}.\] 
		Let $D$ be the resulting database, and note that $|D| = \|\I\|$ and that $D$ is computable from $\I$ in linear time. The equivalence $v\in S_{i,j}\Leftrightarrow (j,v)\in R_i^D$ implies the equivalence $v\in S_{1, j_1}\cap \ldots \cap S_{k, j_k}\Leftrightarrow D\models R_1(j_1, v) \land \ldots \land R_k(j_k, v)$, and then for every query $(j_1, \ldots, j_k)$ of the \pb{$k$-Set-Disjointness} problem, we have that $S_{1, j_1}\cap \ldots \cap S_{k, j_k} \ne \emptyset$ if and only if $(j_1, \ldots, j_k)\in \qstarp_k(D)$.
	\end{proof}

	\section{Hardness of Set Disjointness} \label{sec:lowersetdisjointness}
	
	In this section, we will show lower bounds for \pb{$k$-Set Disjointness}, as defined in  Definition~\ref{def:setdisjointness}. In combination with the reductions from the previous section, this will give us lower bounds for direct access in Section~\ref{sec:together}. The main result of this section is the following.
	
	\begin{theorem}\label{thm:lowerksetdisjointness}
		If there is $k\in \mathbb{N}, k \ge 2$ and $\varepsilon >0$ such that for all $\delta > 0$ there is an algorithm for \textsf{$k$-Set-Disjointness} that, given an instance $\mathcal{I}$, answers queries in time $O(\|\mathcal{I}\|^\delta)$ after preprocessing time $O(\|\mathcal{I}\|^{k-\varepsilon})$, then the \textsf{Zero-$(k+1)$-Clique} Conjecture is false.
	\end{theorem}
	
	\subsection*{Case \boldmath$k=2$}
	
	Kopelowitz, Pettie, and Porat~\cite{KopelowitzPP16} established hardness of \textsf{2-Set-Disjointness} under the \textsf{3-SUM} Conjecture, and Vassilevska Williams and Xu~\cite{WilliamsX20} showed that the same hardness holds under the \textsf{Zero-3-Clique} Conjecture (and thus also under the \textsf{APSP} Conjecture).
	
	\begin{theorem}[Corollary 3.12 in \cite{WilliamsX20}]\label{thm:3sum}
		Assuming the \textsf{Zero-3-Clique} Conjecture (or the \textsf{$3$-SUM} or \textsf{APSP} Conjecture), for any $0<\gamma < 1$ and any $\varepsilon > 0$, on \textsf{2-Set-Disjointness} instances $\mathcal{I}$ with $n$ sets, each of size $O(n^{1-\gamma})$, and universe size $|U| = \Theta(n^{2-2\gamma})$, the total time of preprocessing and $\Theta(n^{1+\gamma})$ queries cannot be $O(n^{2-\varepsilon})$.
	\end{theorem}
	
	This implies Theorem~\ref{thm:lowerksetdisjointness} for $k=2$.
	
	\begin{corollary}
		If there is $\varepsilon >0$ such that for all $\delta > 0$ there is an algorithm for \textsf{2-Set-Disjointness} that, given an instance $\mathcal{I}$, answers queries in time $O(\|\mathcal{I}\|^\delta)$ after preprocessing time $O(\|\mathcal{I}\|^{2-\varepsilon})$, then the \textsf{Zero-3-Clique} Conjecture is false (and thus also the \textsf{3-SUM} and \textsf{APSP} Conjectures are false).
	\end{corollary}
	\begin{proof}
		Assume that we can solve \textsf{2-Set-Disjointness} with preprocessing time $O(\|\mathcal{I}\|^{2-\varepsilon})$ and query time $O(\|\mathcal{I}\|^\delta)$ where $\delta := \eps/10$. 
		Set $\gamma := 1 - 4\delta$, and consider an instance $\mathcal{I}$ consisting of $n$ sets, each of size $O(n^{1-\gamma})$, so the input size is $\|\mathcal{I}\| = O(n^{2-\gamma})$. Over $\Theta(n^{1+\gamma})$ queries, the total query time is $O(n^{1+\gamma} \|\I\|^{\delta}) \le O(n^{1+\gamma+2\delta}) = O(n^{2-2\delta})$. The preprocessing time on $\mathcal{I}$ is $O(\|\mathcal{I}\|^{2-\varepsilon}) = O(n^{(2-\gamma)(2-10\delta)}) = O(n^{(1+4\delta)(2-10\delta)}) = O(n^{2-2\delta})$. Hence, the total time of the preprocessing and $\Theta(n^{1+\gamma})$ queries on $\mathcal{I}$ is $O(n^{2-\varepsilon'})$ for $\varepsilon' := 2\delta$. This contradicts Theorem~\ref{thm:3sum} (for parameters $\gamma,\varepsilon'$).
	\end{proof}
	
	In particular, by combining our previous reductions, for $\qstar_2$ with a bad order $L$ there is no algorithm with polylogarithmic access time and preprocessing $O(|D|^{2-\varepsilon})$, assuming any of the three conjectures mentioned above.
	
	\subsection*{Generalization to \boldmath$k\ge 2$}
	
	In the following, we show how to rule out preprocessing time $O(\|\I\|^{k-\varepsilon})$ for \textsf{$k$-Set-Disjointness} for any $k \ge 2$. To this end, we make use of the \textsf{Zero-$(k+1)$-Clique} Conjecture. Our chain of reductions closely follows the proof of the case $k=2$ 
	by Vassilevska Williams and Xu~\cite{WilliamsX20}, but also uses additional ideas required for the generalization to larger $k$.
	
	\subsection{$k$-Set-Intersection}
	
	We start by proving a lower bound for the following problem of listing many elements in a set intersection.
	
	\begin{definition} \label{def:setintersection}
		In the \pb{$k$-Set-Intersection} problem, we are given an instance $\I$ consisting of a universe $U$ and families $\mathcal{A}_1, \ldots, \mathcal{A}_k \subseteq 2^U$ of subsets of $U$. We denote the sets in family $\A_i$ by $S_{i,1},\ldots,S_{i,|\A_i|}$.
		The task is to preprocess $\I$ into a data structure that can answer queries of the following form: Given indices $j_1,\ldots,j_k$ and a number $T$, compute the set $S_{1,j_1} \cap \ldots \cap S_{k,j_k}$ if it has size at most $T$; if it has size more than $T$, then compute any $T$ elements of this set.
		
		We denote the number of sets by $n := \sum_i |\A_i|$ and the input size by $\|\I\| := \sum_i \sum_{S \in \A_i} |S|$. We call $|U|$ the universe size.
	\end{definition}
	
	The main result of this section shows that \pb{$k$-Set-Intersection} is hard in the following sense. Later we will use this result to show hardness for \pb{$k$-Set-Disjointness}.
	
	\begin{theorem}\label{thm:intersection}
		Assuming the \pb{Zero-$(k+1)$-Clique} Conjecture, for every constant $\varepsilon> 0$ there exists a constant $\delta > 0$ such that no randomized algorithm solves \pb{$k$-Set-Intersection} on $n$ sets, universe size $|U|=n$, and $T = \Theta(n^{1-\eps/2})$ in preprocessing time $O(n^{k+1-\varepsilon})$ and query time $O(Tn^\delta)$.
	\end{theorem}
	
	We prove Theorem~\ref{thm:intersection} in the rest of this section.
	The idea of this proof is as follows. We take a finite range that is big enough to contain all possible clique weights, and we spread the edge weights evenly over this range using random choices in a way that maintains the zero-cliques. We split this range into many intervals of small size, and we compute the set $S$ of all tuples of $k+1$ intervals for which there exists a choice of an element from each interval such that all chosen elements sum up to zero. We have that for every zero-clique $(v_1,\ldots,v_{k+1})$, there exists an interval tuple in $S$ such that the weight of the clique formed by $(v_1,\ldots,v_k)$ falls within the first interval, and the weights of the edges between $v_{k+1}$ and the other vertices of the clique fall within the other $k$ intervals. In this way, all zero-cliques lie within a tuple of intervals in $S$. For every interval tuple in $S$, we use the $k$-Set-Intersection algorithm to find many vertex tuples that fall within this interval tuple. For each such candidate vertex tuple, we check whether it indeed forms a zero-clique. This algorithm finds a zero-clique if one exists with a constant success probability. If the algorithm does not detect a zero-clique, it declares that there are none. Repeating this process with different random choices yields an algorithm that succeeds with high probability. We next describe our construction in more detail.
	
	\paragraph{Preparations.}
	First note that finding a zero-$k$-clique in a general graph is equivalent to finding a zero-$k$-clique in a complete $k$-partite graph. This can be shown using a reduction in which each vertex $v$ is duplicated $k$ times, giving vertices of the form $(v,i)$ with $1\le i \le k$, and each edge 
	$\{u, v\}$ is replaced by the $k^2-k$ edges $\{(u,i), (v,j)\}$ for $1\le i,j\le k$ and $i\ne j$; then, to make this $k$-partite graph complete, edges of very large weight are added where needed.
	\begin{observation}[\cite{AbboudBDN18}]\label{obs:kpartite}
		If the \pb{Zero-$k$-Clique} Conjecture is true, then it is also true restricted to complete $k$-partite graphs.
	\end{observation}
	Due to this observation, we can assume that we are given a complete $(k+1)$-partite graph $G=(V,E)$ with $n$ vertices and with color classes $V_1, \ldots, V_{k+1}$ and a weight function $w$ on the edges. We denote by $w(u,v)$ the weight of the edge from $u$ to $v$, and more generally by $w(v_1,\ldots,v_\ell) := \sum_{1 \le i < j \le \ell} w(v_i,v_j)$ the total weight of the clique $(v_1,\ldots,v_\ell)$.
	
	For our construction it will be convenient to assume that the edge weights lie in a finite field. 
	Recall that in the \pb{Zero-$(k+1)$-Clique} problem we can assume that all weights are integers between $-n^c$ and $n^c$ for some constant $c$.
	We first compute a prime number between $10(k+1)^2 n^{c}$ and $100 (k+1)^2 n^{c}$ by a standard algorithm: Pick a random number in that interval, check whether it is a prime, and repeat if it is not. By the prime number theorem, a random number in that interval is a prime with probability $\Theta(1/\log(n))$. It follows that in expected time $O(\textup{polylog}(n))$ we find a prime $p$ in that interval.
	
	Having found the large prime $p$, we consider all edge weights as given in the finite field $\mathbb{F}_p$. Note that $p$ is bigger than the sum of weights of any $(k+1)$-clique in $G$, so the $(k+1)$-cliques of weight $0$ are the same over $\mathbb{Z}$ and $\mathbb{F}_p$. Thus, in the remainder we assume all arithmetic to be done over $\mathbb{F}_p$.

	In our analysis, it will be necessary to assume that the input weights are evenly distributed in~$\mathbb{F}_p$. This is generally of course not the case, so we have to distribute them more evenly. To this end, we use a hashing technique that redistributes the weights close to randomly while maintaining the zero-cliques. Concretely, we define a new weight function $w':E\rightarrow \mathbb{F}_p$ by choosing independently and uniformly at random
	from $\mathbb{F}_p$:
	\begin{itemize}
		\item one value $x$,
		\item for all $v\in V_{k+1}$ and all $j\in [k-1]$ a value $y_v^j$.
	\end{itemize}
	We define a new weight function $w'$ by setting for every $i,j \in [k+1]$ with $i\ne j$ and every $v \in V_i$ and $u \in V_j$:
	\begin{align} \label{eq:def_new_weights}
		w'(v, u) = x \cdot w(v, u) + 
		\begin{cases} y_{u}^1, & \text { if } i=1, j=k+1\\ 
			y_{u}^i - y_u^{i-1}, & \text { if } 2 \le i < k, j = k+1\\ 
			- y_u^{k-1}, & \text { if } i=k, j=k+1\\ 
			0, & \text{ otherwise }
		\end{cases}
	\end{align}
	Note that in every $(k+1)$-clique $(v_1,\ldots,v_{k+1}) \in V_1\times \ldots \times V_{k+1}$ the sum $w'(v_1,\ldots, v_{k+1})$ contains every term $y^i_{v_{k+1}}$ for $i \in [k-1]$ once positively and once negatively, so for the overall weight we have $w'(v_1,\ldots, v_{k+1}) = x\cdot w(v_1,\ldots, v_{k+1})$. In particular, every zero-clique according to $w$ is also a zero-clique according to $w'$, and we continue the construction with the weight function $w'$.
	
	We next split $\mathbb{F}_p$ into $n^\rho$ disjoint intervals of size $\lceil p/n^\rho\rceil$ or $\lfloor p/n^\rho\rfloor$, where $\rho$ is a constant that we will choose later such that $p\geq n^\rho$. In the following, $I_i$ always denotes an interval resulting from this splitting of $\mathbb{F}_p$.
	We denote by $S$ the set of all tuples $(I_0,\ldots,I_k)$ of intervals in $\mathbb{F}_p$ such that $0 \in \sum_{i=0}^{k} I_i$. Note that there are only $O(n^{\rho k})$ tuples in $S$ and they can be enumerated efficiently: given $I_0, \ldots, I_{k-1}$, the set $-\sum_{i=0}^{k-1} I_i$ is covered by $\Theta(k)= \Theta(1)$ intervals $I_k$, and these intervals can be computed efficiently.
	
	For any clique $C = (v_1,\ldots,v_{k+1}) \in V_1 \times \ldots \times V_{k+1}$ there is a unique tuple $(I_0, \ldots, I_k)$ of intervals such that  
	\begin{align}\label{eq:intervalconditions}
		\forall i\in [k]\colon  w'(v_i, v_{k+1}) \in I_i \;\text{ and }\; w'(v_1,\ldots,v_k) \in I_0. 
	\end{align}
	We call $(I_0, \ldots, I_k)$ the \emph{weight interval tuple} of $C$.
	If $C$ is a zero-clique, then its weight interval tuple must be in $S$. 
	
	\paragraph{The Reduction.}
	With these preparations, we are now ready to present our reduction. For every weight interval tuple $(I_0,\ldots,I_k) \in S$ we construct an instance of \pb{$k$-Set-Intersection} as follows. For any $i \in [k]$ we construct
	\begin{align*}
		\mathcal{A}_i := \{S_{i,v}\mid v\in V_i\} \;\text{ where }\; S_{i,v} := \{ u\in V_{k+1}\mid w'(v,u)\in I_i\}.
	\end{align*}
	Note for every tuple $(v_1, \ldots, v_k)\in V_1\times \ldots V_k$ (which plays the role of $j_1, \ldots, j_k$) the equivalence
	\begin{align*}
		u\in S_{1, v_1}\cap \ldots \cap S_{k, v_k} \Leftrightarrow \forall i \in [k] w'(v_i, u)\in I_i.
	\end{align*}
	Moreover, we consider the set of queries 
	\begin{align*}
		Q := \{(v_1,\ldots,v_k) \in V_1\times \ldots \times V_k \mid w'(v_1,\ldots,v_k) \in I_0 \}.
	\end{align*}
	We first run the preprocessing of \pb{$k$-Set-Intersection} on the instance $\I = (\A_1,\ldots,\A_k)$ and then query each $(v_1, \ldots, v_k) \in Q$ with parameter $T := \lceil 100\cdot 3^k n^{1-k\rho} \rceil$ (so we compute the entire set $S_{1, v_1}\cap \ldots \cap S_{k, v_k}$ if it has at most $T$ elements, and otherwise we compute $T$ elements of this set). For each answer $v_{k+1}$ returned by some query $(v_1,\ldots,v_k)\in Q$, we check whether $(v_1,\ldots,v_{k+1})$ is a zero-clique with respect to $w$, and if it is, then we return this zero-clique.
	
	After repeating this construction for every weight interval tuple $(I_0,\ldots,I_k) \in S$, if we never found a zero-clique, then we return~`no'.
	
	\paragraph{Correctness.}
	Let us show correctness of the reduction. Since we explicitly test for each query answer whether it is a zero-clique, if $G$ contains no zero-clique then our reduction returns `no'.
	
	We next show that if there exists a zero-clique then our reduction returns a zero-clique with probability at least $.99$. To this end, we use the following claim.
	
	\begin{claim}\label{clm:secondcutoff}
		Let $(v_1,\ldots,v_{k+1}) \in V_1 \times \ldots \times V_{k+1}$ be a zero-clique.
		Then, with probability at least $.99$, there are less than $100\cdot 3^k n^{1-k\rho}$ vertices $u\in V_{k+1}$ such that $(v_1,\ldots,v_k,u)$ is not a zero-clique and has the same weight interval tuple as $(v_1,\ldots,v_{k+1})$.
	\end{claim}
	
	We first finish the correctness proof and later prove the claim. 
	Fix any zero-clique $(v_1,\ldots,v_{k+1})$, and denote its weight interval tuple by $(I_0,\ldots,I_k)$. Note that we query $(v_1,\ldots,v_k)$ in the instance constructed for $(I_0,\ldots,I_k)$. With probability at least $0.99$, this query has less than $100\cdot 3^k n^{1-k\rho}$ ``false positives'', that is, vertices $u \in V_{k+1}$ such that $(v_1,\ldots,v_k,u)$ has weight interval tuple $(I_0,\ldots,I_k)$ and $(v_1,\ldots,v_k,u)$ is not a zero-clique. Then by listing $T = \lceil 100\cdot 3^k n^{1-k\rho} \rceil$ answers for the query $(v_1,\ldots,v_k)$, at least one answer $v$ must correspond to a zero-clique. (That is, we do not necessarily find $v_{k+1}$, but we find some vertex $v$ forming a zero-clique together with $v_1,\ldots,v_k$.) 
	
	Hence, with probability at least $.99$ we find a zero-clique, if there exists one. This probability can be boosted to high probability (that is, probability $1-1/n^d$ for any desirable constant $d$) by repeating the reduction algorithm $O(d\log n)$ times and returning a zero-clique if at least one repetition finds a zero-clique. This repetition incurs a logarithmic factor in the runtime of our algorithm which we will ignore later on since it is dominated by the polynomial terms we analyze there.
	
	\begin{proof}[Proof of Claim~\ref{clm:secondcutoff}]
		To ease notation, let $v := v_{k+1}$.
		Let $(I_0, \ldots, I_k)$ be the weight interval tuple of $(v_1,\ldots,v_k,v)$, that is, $w'(v_i, v) \in I_i$ for $i\in [k]$, and $w'(v_1,\ldots,v_k) \in I_0$. Let $u\in V_{k+1}$ such that $(v_1,\ldots,v_k,u)$ is not a zero-clique. If $(v_1,\ldots,v_k,u)$ has the same weight interval tuple, then for every $i\in [k]$ we have $w'(v_i, u)\in I_i$. It then follows that 
		\begin{align}\label{eq:containedinbox}
			\forall i\in [k]\colon w'(v_i, v)- w'(v_i, u) \in [-pn^{-\rho}, pn^{-\rho}].
		\end{align}
		We want to bound the probability that this happens, over the random choices in the construction of $w'$. To this end, write $w'_i := w'(v_i, v)- w'(v_i, u)$ and $w_i := w(v_i, v)- w(v_i, u)$ for all $i \in [k]$. Expanding the definition of $w'$ (see (\ref{eq:def_new_weights})) we obtain
		\begin{align}
			w'_1 &= x \cdot w_1 + (y_v^1 - y_u^1),\notag\\
			w'_i &= x \cdot w_i -(y_v^{i-1} - y_u^{i-1}) + (y_v^i - y_u^{i}) \;\text{for each }i\in \{2, \ldots, k-1\}\notag\\
			w'_k &= x \cdot w_k -(y_v^{k-1} -y_u^{k-1}),\label{eq:bijection}
		\end{align}
		
		We claim that for fixed $w_1,\ldots,w_k$ and $y_v^1,\ldots,y_v^{k-1}$ the above equations induce a bijection sending $(w'_1,\ldots,w'_k)$ to $(x,y_u^1,\ldots,y_u^{k-1})$.
		Indeed, by the equations (\ref{eq:bijection}) we can compute $(w'_1,\ldots,w'_k)$ given $(x,y_u^1,\ldots,y_u^{k-1})$. For the other direction, we use that $(v_1,\ldots,v_k,v)$ is a zero-clique (that is, $w(v_1,\ldots,v_k,v) = 0$), and $(v_1,\ldots,v_k,u)$ is not a zero-clique (that is, $w(v_1,\ldots,v_k,u) \ne 0$). This implies 
		\begin{align*}
			0 \;\ne\; &w(v_1,\ldots,v_k,v) - w(v_1,\ldots,v_k,u) = \sum_{i=1}^k{\left(w(v_i,v) - w(v_i,u)\right)} = \sum_{i=1}^k w_i.
		\end{align*} 
		By adding up all equations (\ref{eq:bijection}), since each term $y_v^i$ and $y_u^i$ appears once positively and once negatively, we obtain $w'_1+\ldots+w'_k = x \cdot (w_1+\ldots+w_k)$. Hence, given $(w'_1,\ldots,w'_k)$, we can compute $x = (w'_1+\ldots+w'_k)/(w_1+\ldots+w_k)$. Note that here we do not divide by~0 since $\sum_{i=1}^k w_i \ne 0$. Next we can compute $y_u^1 = x \cdot w_1 + y_v^1 - w'_1$ by rearranging the first equation of (\ref{eq:bijection}) (recall that $w_1$ and $y_v^1$ are fixed, we are given $w'_1$, and we just computed $x$, so each variable on the right hand side is known). Similarly, from $x$ and $y_u^{i-1}$ we can compute $y_u^i = x \cdot w_i - (y_v^{i-1} - y_u^{i-1}) + y_v^i - w'_i$. Iterating over all $i$ computes $(x,y_u^1,\ldots,y_u^{k-1})$ given $(w'_1,\ldots,w'_k)$.
		
		Since we showed a bijection sending $(w'_1,\ldots,w'_k)$ to $(x,y_u^1,\ldots,y_u^{k-1})$, and since $x,y_u^1,\ldots,y_u^{k-1}$ are all chosen independently and uniformly random from $\mathbb{F}_p$, we obtain that $(w'_1,\ldots,w'_k)$ is uniformly random in $\Fp^k$. In particular, each $w'_i = w'(v_i, v)- w'(v_i, u)$ is independently and uniformly at random distributed in $\Fp$. Thus, since the intervals in expression~(\ref{eq:containedinbox}) have size at most $2pn^{-\rho} +1$, the probability of expression (\ref{eq:containedinbox}) being true (which is greater than or equal to the probability that $(v_1, \ldots, v_k, v)$ and $(v_1, \ldots, v_k, u)$ have the same weight interval tuple) is at most
		\begin{align*}
			\left(\frac{2pn^{-\rho}+1}{p}\right)^k \leq 3^k n^{-k\rho},
		\end{align*}
		where we use the fact that $p\geq n^\rho$.
		Thus, given $u\in V_{k+1}$ such that $(v_1,\ldots,v_k,u)$ is not a zero-clique, the probability that $(v_1,\ldots,v_k,u)$ has the same weight interval tuple as $(v_1,\ldots,v_{k+1})$ is at most $3^k n^{-k\rho}$.
		
		Let $X$ be the random variable that is the number of vertices  $u\in V_{k+1}$ satisfying the conditions in the claim (i.e., $(v_1,\ldots,v_k,u)$ is not a zero-clique and has the same weight interval tuple as $(v_1,\ldots,v_{k+1})$).
		As $|V_{k+1}|\leq n$ and the probability of each $u\in V_{k+1}$ to satisfy the conditions is at most $3^k n^{- k\rho}$, the expectation of $X$ is $E(X) \leq 3^k n^{1 - k\rho}$.
		We can apply Markov's inequality to get $\Pr(X\ge 100 E(X))\le \frac{1}{100}$, implying that $\Pr(X\ge 100\cdot3^k n^{1 - k\rho})\le \frac{1}{100}$. This proves the claim.
	\end{proof}
	
	\paragraph{Running Time.}
	It remains to analyze the running time and to set the parameter $\rho$. 
	Assume, to the contrary of the theorem statement, that for some $\eps > 0$ \pb{$k$-Set-Intersection} on $n$ sets, universe size $|U| = n$, and $T = \Theta(n^{1-\eps/2})$ can be solved with preprocessing time $O(n^{k+1-\eps})$ and query time $O(T n^{\delta})$ for $\delta := \frac{\eps}{4k}$. We assume w.l.o.g.~that $\eps \le 1$. Note that our constructed instances have universe size $|U| = |V_{k+1}| \le n$ (which can be embedded into a universe of size $|U| = n$), and $T = O(n^{1-k\rho})$, so by setting $\rho := \frac{\eps}{2k}$ the \pb{$k$-Set-Intersection} algorithm can be applied.
	Note that indeed $p\geq n^\rho$ as we promised earlier.
	
	Over $O(n^{k\rho})$ instances (one for every interval tuple in $S$), we have that the total preprocessing time is $O(n^{k+1+k\rho - \eps}) = O(n^{k+1-\eps/2})$, since $\rho = \frac{\eps}{2k}$.
	
	Let us count the total number of queries. Note that $(v_1,\ldots,v_k)$ is queried in the instance for $(I_0,\ldots,I_k) \in S$ if and only if $w'(v_1,\ldots,v_k) \in I_0$. Recall that for every $I_0,\ldots,I_{k-1}$ there are $O(1)$ choices for $I_k$ such that $(I_0,\ldots,I_k) \in S$. Hence, each $(v_1,\ldots,v_k)$ is queried in $O(n^{(k-1)\rho})$ instances, as there are $O(n^\rho)$ choices for each of $I_1,\ldots,I_{k-1}$ and $O(1)$ choices for $I_k$.
	The total number of queries is thus $O(n^{k + (k-1)\rho})$. Each query runs in time $O(T n^\delta) = O(n^{1-k\rho+\delta})$, so the total query time over all our constructed instances is $O(n^{k + (k-1)\rho} \cdot n^{1-k\rho+\delta}) = O(n^{k+1+\delta-\rho}) = O(n^{k+1-\rho/2})$ since we set $\delta = \rho/2$. 
	
	In total, we obtain running time $O(n^{k+1-\rho/2})$ for \pb{Zero-$(k+1)$-Clique}, contradicting the \pb{Zero-$(k+1)$-Clique} Conjecture. This concludes the proof of Theorem~\ref{thm:intersection}.\qed

	\subsection{Unique-\boldmath$k$-Set-Intersection}
	
	As the next step, we show hardness of \pb{$k$-Set-Intersection} for $T=1$, and even if we get the promise that there is a unique answer to a query. We formally define this problem as follows.
	
	\begin{definition}
		The \pb{Unique-$k$-Set-Intersection} problem has the same input $\I = (\A_1,\ldots,\A_k)$ with $\A_i = \{S_{i,1},\ldots,S_{i,|\A_i|}\} \subseteq 2^U$ as \pb{$k$-Set-Intersection} (cf.~Definition~\ref{def:setintersection}). In the query, we are given indices $j_1,\ldots,j_k$, and if the set $S_{1,j_1} \cap \ldots \cap S_{k,j_k}$ consists of exactly one element, then we return that element; otherwise we return the error element $\bot$. Again we write $n = \sum_i |\A_i|$.
	\end{definition}
	
	We show how to reduce \pb{$k$-Set-Intersection} to \pb{Unique-$k$-Set-Intersection}.
	
	\begin{lemma}\label{lem:tounique}
		For every $0<\theta \le 1$, every $\varepsilon>0$ and every $\delta>0$, if there is a randomized algorithm for \pb{Unique-$k$-Set-Intersection} on $n$ sets and universe size $O(n^{\theta})$ with preprocessing time $\softO(n^{k-\varepsilon})$ and query time $\softO(n^\delta)$, then there is a randomized algorithm for \pb{$k$-Set-Intersection} on $n$ sets, universe size $|U|=n$, and $T = \Theta(n^{1-\theta})$ with preprocessing time $\softO(T n^{k-\varepsilon})$ and query time $\softO(Tn^\delta)$.
	\end{lemma}
	\begin{proof}
		The proof idea is as follows. We create many instances of Unique-$k$-Set-Intersection by restricting the k-Set-intersection instance according to subsamples of the universe, and we return the union of the answers returned by the Unique-$k$-Set-Intersection algorithm over these instances. By properly choosing the number of instances and the subsampling probability, we can show that this process generates with high probability enough answers while maintaining the required running time. We next describe this idea in more detail.
		
		Let $\mathcal{I}$ be an instance of \pb{$k$-Set-Intersection} with families $\A_1,\ldots,\A_k$ of total size $n = \sum_i |\A_i|$ over a universe $U$ of size $|U|=n$, and let $T = \Theta(n^{1-\theta})$ be the query size bound.
		For every $\ell \in \{\log(T),\ldots,\log(4n)\}$ and every $j \in [100 T \ln(n)]$ we create an instance $\I_{\ell,j}$ of \pb{Unique-$k$-Set-Intersection} by subsampling with probability $2^{-\ell}$, i.e., starting from $\I$ and keeping each element in the universe $U$ with probability $2^{-\ell}$ and otherwise deleting it from all sets in $\A_1,\ldots,\A_k$.
		
		Denote by $U_{\ell,j}$ the surviving universe elements of $\I_{\ell,j}$. We want to show that with high probability $U_{j, \ell}$ does not contain more than $\Theta(n^{\theta})$ elements. Clearly, the probability that $U_{\ell,j}$ has more than $\Theta(n^{\theta})$ elements only increases when the probability that an individual element survives increases, so if $\ell$ decreases. Thus, we get 
		$
		\Pr(|U_{\ell,j}| \ge p) \le      \Pr(|U_{\log(T),j}| \ge p)
		$ for any value $p$.

		Note that in expectation the size of $U_{\log(T),j}$ is $E(|U_{\log(T),j}|) = |U| \cdot 2^{-{\log(T)}} = \frac{n}{T}= \Theta(n^{\theta})$.
		We will use the following basic Chernoff bound: let $X_1, \ldots, X_r$ be independent Bernoulli random variables and let $X$ be their sum.  Then we have $\Pr(X\ge (1+\delta) E(X)) \le \left(\frac{e^{\delta}}{(1+\delta)^{1+\delta}}\right)^{E(X)}$ for any $\delta>0$.
		Setting $\delta :=1$, we get in our case
		\begin{align*}
			\Pr(|U_{\ell,j}| \ge 2E(|U_{\log(T),j}|))\le
			\Pr(|U_{\log(T),j}| \ge 2E(|U_{\log(T),j}|))\le
			\left(\frac{e}{2^{2}}\right)^{E(|U_{\log(T),j}|)}\le
			0.7^{E(|U_{\log(T),j}|)}.
		\end{align*}
		Using a union bound over all $O(T\log^2(n))$ instances $\I_{\ell, j}$, we get that with high probability none of the $U_{\ell, j}$ will have more than $\Theta(n^\theta)$ elements.
		So in the remainder of this proof, we assume that this is the case (if we ever encounter a universe that is too large in the subsampling, we simply let the algorithm fail; since the probability is small, this is not a problem for the success probability).
		
		After the subsampling, we use the assumed algorithm for \pb{Unique-$k$-Set-Intersection} to preprocess all instances $\mathcal{I}_{\ell, j}$. Then, given a query $(j_1, \ldots, j_k)$, we run the algorithm for \pb{Unique-$k$-Set-Intersection} on all instances $\I_{\ell,j}$, filter out $\bot$-answers, remove duplicates, and return up to $T$ of the remaining answers.
		
		Let us analyze this algorithm. Since we filter out wrong answers, all answers we return are correct. It remains to show that with good probability we return enough answers. To this end, fix a query and assume that the instance $\mathcal{I}$ has $r>0$ answers for that query. Pick $\ell$ such that $r+T\in [2^{\ell-1}, 2^{\ell}]$ and consider what happens in the queries on the instances $\mathcal{I}_{\ell, j}$. For every answer $v$, consider the event $\mathcal{E}_v$ that $v$ is isolated by instance $\mathcal{I}_{\ell,j}$, that is, $v$ survives the subsampling, but none of the other $r-1$ solutions survive the subsampling. This event has probability
		\begin{align*}
			\Pr(\mathcal{E}_v)=\frac{1}{2^\ell} \left(1-\frac{1}{2^\ell}\right)^{r-1}\ge \frac{1}{2(r+T)} \left(1-\frac{1}{r+T}\right)^{r+T} \ge \frac{1}{6(r+T)},
		\end{align*}
		where in the first inequality we used that $r+T\leq 2^\ell \leq 2(r+T)$ and $T\ge 0$, and for the last inequality we assume that $r+T\ge 10$ in which case $\left(1-\frac{1}{r+T}\right)^{r+T}\ge \frac 1 3$.
		Assume first that $r \le 2T$, so the above probability is at least $\frac{1}{18T} \ge \frac{1}{20T}$. Then the probability that the answer $v$ is not isolated in any instance $\I_{\ell,j}$ over all $j \in [100T \ln(n)]$ is at most $(1-\frac{1}{20T})^{100T \ln(n)} = n^{100T \ln(1-\frac{1}{20T})}\le n^{-5}$ (using the fact $\ln(1-x) \le -x$).
		By a union bound over all $r\le |U| = n$ answers, with probability $\ge 1-n^{-4}$ each of the $r$ answers will be isolated in some instance $\I_{\ell,j}$, so the \pb{Unique-$k$-Set-Intersection} algorithm will return each answer at least once. Hence, we see all answers and can correctly return $\min\{r,T\}$ of them.
		
		Now consider the case $r > 2T$. We have $100T\ln(n)$ rounds where in round $j$ we make a random experiment, creating an instance $\I_{\ell,j}$ as above. This is essentially a coupon collector problem (see e.g.~\cite{mitzenmacher2017probability} for more background) which we analyze next. We will show that with high probability we found $T$ different answers after all rounds are done.
		If the event $\mathcal{E}_v$ happens, we say that the answer $v$ was found. Since the events $\mathcal{E}_v$ are disjoint, in each round we get an answer with probability at least
		\begin{align*}
			\frac{r}{6(r+T)} \ge \frac{r}{6(r+\frac{1}{2}r)} = \frac{1}{9} \ge \frac{1}{10}.
		\end{align*}
		In case we get an answer, it is a uniformly random one.
		Since $T < r/2$, as long as we have seen less than $T$ distinct answers, the probability that the next answer is a fresh one is at least half. In combination with the bound on getting an answer, in each round we find an answer that we have not seen before with probability at least $\frac{1}{20}$.
		We next treat the $100T\ln(n)$ rounds as $T$ batches of $100\ln(n)$ rounds. The probability of not getting a new answer in a batch of $100 \ln(n)$ rounds is at most $(1-\frac{1}{20})^{100\ln(n)} \le n^{-5}$. Taking the union bound over $T$ batches, we get at least $T$ distinct answers over all rounds with probability at least $1 - \frac{T}{n^5} \ge 1 - \frac{1}{n^4}$.
		
		In both cases, with high probability we return the correct answer.
		
		To analyze the running time, note that we construct $100T\ln(n) = \softO(T)$ instances $\mathcal{I}_{\ell, j}$, so the preprocessing takes time $\softO(Tn^{k- \varepsilon})$. Also, each query to an instance $\mathcal{I}_{\ell,j}$ takes time $\softO(n^\delta)$ and since we make $\softO(T)$ of them, the total time of all queries is $\softO(Tn^\delta)$. Finally, note that the different answers can easily be collected in, say, a binary search tree on which all operations are in time $O(\ln(T))= \softO(1)$, so the queries to $\mathcal{I}$ can overall be answered in time $\softO(Tn^\delta)$.
	\end{proof}
	
	\subsection{\boldmath$k$-Set-Disjointness}
	
	Recall that the \pb{$k$-Set-Disjointness} problem is defined similarly as \pb{$k$-Set-Intersection}, except that a query should just decide whether the intersection $S_{1,j_1}\cap \ldots \cap S_{k,j_k}$ is empty or not. Here an instance~$\I$ is given by a universe set $U$ and families $\A_1,\ldots,\A_k \subseteq 2^U$, where we write $\A_i = \{S_{i,1},\ldots,S_{i,|\A_i|}\}$. 
	Again we denote the number of sets by $n = \sum_i |\A_i|$.
	
	\begin{lemma}\label{lem:uniquetodisjointness}
		Let $\varepsilon, \delta > 0$ and $0 < \theta < 1$ be constants. If there is a randomized algorithm for \pb{$k$-Set-Disjointness} on $n$ sets and universe size $|U| = \Theta(n^\theta)$ with preprocessing $\softO(n^{k-\varepsilon})$ and query time $\softO(n^\delta)$, then there is a randomized algorithm for \pb{Unique-$k$-Set-Intersection} on $n$ sets and universe size $|U| = \Theta(n^\theta)$ with preprocessing time $\softO(n^{k-\varepsilon})$ and query time $\softO(n^\delta)$.
	\end{lemma}
	\begin{proof}
		The principle of our reduction is given by the following claim.
		\begin{claim}\label{clm:unique}
			Let $E$ be a set of words over $\{0,1\}$ of length $\ell$. For every $j\in [\ell]$ and $b\in \{0,1\}$, let $E^{j,b}$ denote the set of words in $E$ whose $j$th bit is different from $b$. Then $E$ contains exactly one word if and only if for all $j\in [\ell]$ exactly one of $E^{j,0}$ and $E^{j,1}$ is empty. Moreover, if $E=\{a\}$, then the $j$th bit $a_j$ of $a$ is the (unique) $b\in \{0,1\}$ such that $E^{j,b}=\emptyset$.
		\end{claim}
		\begin{proof}
			There are three cases depending on whether $E$ contains a single word, is empty or contains at least two words:
			\begin{itemize}
				\item If $E=\{a\}$, then for all $j$ and $b$ we have $E^{j,b} = \emptyset$ if $b=a_j$, and $a\in E^{j,b}\ne \emptyset$ if $b\ne a_j$.
				\item If $E=\emptyset$, then for all $j$ and $b$ we have $E^{j,b}= \emptyset$ since $E^{j,b}\subseteq E$ by construction.
				\item If $E$ contains two distinct words that differ on the $j$th bit, then $E^{j,0}$ and $E^{j,1}$ are non-empty.
			\end{itemize}
			This proves the claim.
		\end{proof}
		
		We now describe an algorithm for \pb{Unique-$k$-Set-Intersection}, using an algorithm for \pb{$k$-Set-Disjointness}, and whose correctness will be justified with the previous claim. Let $\mathcal{I}$ be an instance of \pb{Unique-$k$-Set-Intersection} consisting of sets $\mathcal{A}_1, \ldots, \mathcal{A}_k$. Assume w.l.o.g.~that the universe $U$ of $\mathcal{A}$ is the set $[|U|]$, so that each element of $U$ can be written with $\ell = \lfloor \log(|U|)\rfloor +1 = O(\log(n))$ bits, since $|U| = \Theta(n^\theta)$.
		
		We create $2\ell$ instances $\I^{j,b}$ of \pb{$k$-Set-Disjointness}, for $j\in [\ell]$ and $b\in \{0,1\}$, as follows. For any $j\in [\ell]$ and $b\in \{0,1\}$, remove every $v \in U$ whose $j$th bit is $b$ from the universe and from all sets in $\A_1,\ldots,\A_k$, and call the resulting sets $\mathcal{A}_i^{j, b}$. Define $\mathcal{I}^{j,b}$ to contain the sets $\mathcal{A}^{j,b}_i$ and the same queries as the original instance $\mathcal{I}$. Denote by $S_{i,h}^{j,b}$ the sets belonging to $\A_i^{j,b}$. We run the preprocessing of the \pb{$k$-Set-Disjointness} algorithm on each instance $\mathcal{I}^{j,b}$. Now, given a query $q = (h_1,\ldots,h_k)$, we ask the same query for each instance $\I^{j,b}$, for all $j\in [\ell]$ and $b\in \{0,1\}$. By definition, each answer $\alpha^{j,b}$ we obtain is the answer \emph{yes/no} to the question whether $E^{j,b}$ is empty, where $E^{j,b}:= S^{j,b}_{1, h_1}\cap \ldots \cap S_{k, h_k}^{j,b}$. With Claim~\ref{clm:unique}, we get that the following (two cases of the) answer to $q$ for the original instand $\I$ of \pb{Unique-$k$-Set-Intersection} is correct:
		\begin{itemize}
			\item if for all $j\in [\ell]$, we have exactly one of $E^{j,0}$or $E^{j,1}$ is empty, i.e., either $\alpha^{j,0} = \text{yes}$ or $\alpha^{j,1} = \text{yes}$, then $S^{j,b}_{1, h_1}\cap \ldots \cap S_{k, h_k}^{j,b} = \{a\}$ where the $j$th bit of $a$ is the unique $b$ such that $\alpha^{j,b} = \text{yes}$ and we return $a$;
			\item otherwise, $S^{j,b}_{1, h_1}\cap \ldots \cap S_{k, h_k}^{j,b}$ does not contain a single element, and we return $\bot$.
		\end{itemize}
		We have thus justified that the algorithm is correct. It is clear that it satisfies the claimed time bounds. 
	\end{proof}
	
	So far we measured running times in terms of the number of sets $n$. Next we show how to obtain lower bounds in terms of the input size $\|\I\| = \sum_i \sum_{S \in \A_i} |S|$ which is at most $\sum_i |\A_i|\cdot|U|$.
	
	\begin{lemma}\label{lem:densetosparse}
		For any constants $\eps,\delta > 0$, if there is an algorithm that solves \pb{$k$-Set-Disjointness} with preprocessing time $\softO(\|\I\|^{k-\varepsilon})$ and query time $\softO(\|\mathcal{I}\|^\delta)$, then for any constant $0 < \theta \le \frac{\eps}{2k}$ there is an algorithm that solves \pb{$k$-Set-Disjointness} on $n$ sets and universe size $|U| = \Theta(n^\theta)$ with preprocessing time $\softO(n^{k-\varepsilon/2})$ and query  time $\softO(n^{\delta(1+\theta)})$.
	\end{lemma}
	\begin{proof}
		For every instance $\mathcal{I}$ of \pb{$k$-Set-Disjointness} on universe size $|U| = \Theta(n^\theta)$, we have that $\|\mathcal I\|= O(n^{1+\theta})$ because $\|\I\| \le \sum_i |A_i|\cdot|U|\le n|U|$. Now on $\mathcal{I}$, the assumed algorithm for \pb{$k$-Set-Disjointness} takes preprocessing time 
		\begin{align*}
			\softO(\|\I\|^{k-\varepsilon}) = \softO(n^{k + k \theta - \varepsilon - \theta \varepsilon}) = \softO(n^{k - \varepsilon/2}),
		\end{align*}
		where we used $\theta \le \frac{\eps}{2k}$ in the last step.
		Observing that the query time is $\softO(\|\mathcal{I}\|^\delta)= \softO(n^{\delta(1+\theta)})$ completes the proof.
	\end{proof}
	
	Finally we are ready to prove the main result of this section: the hardness of \textsf{$k$-Set-Disjointness} for all $k \ge 2$ based on the \textsf{Zero-$(k+1)$-Clique} Conjecture.
	
	\begin{proof}[Proof of Theorem~\ref{thm:lowerksetdisjointness}]
		Almost all algorithms discussed in this proof are randomized; we omit this word for readability.
		
		Assume that there is a constant $\eps > 0$ such that for all $\delta > 0$ there is an algorithm for \pb{$k$-Set-Disjointness} with preprocessing time $O(\|\I\|^{k-\varepsilon})$ and query time $O(\|\mathcal{I}\|^\delta)$. 
		
		For $\eps' := \frac{\eps}{2k}$, Theorem~\ref{thm:intersection} shows that there is a constant $\delta' > 0$ such that \pb{$k$-Set-Intersection} on $n$ sets, universe size $|U|=n$, and $T = \Theta(n^{1-\eps'/2})$ has no algorithm with preprocessing time $O(n^{k+1-\varepsilon'})$ and query time $O(Tn^{\delta'})$, if we assume the \pb{Zero-$(k+1)$-Clique} Conjecture.
		
		Using this $\delta'$, we set $\delta := \frac{\delta'}{2(1+\eps'/2)}$. Since we can solve \pb{$k$-Set-Disjointness} with preprocessing time $O(\|\I\|^{k-\varepsilon})$ and query time $O(\|\mathcal{I}\|^\delta)$, by Lemma~\ref{lem:densetosparse} for $\theta := \eps'/2$ there is an algorithm for \pb{$k$-Set-Disjointness} on $n$ sets and universe size $\Theta(n^{\eps'/2})$ with preprocessing time $\softO(n^{k-\varepsilon/2})$ and query time $\softO(n^{\delta(1+\eps'/2)}) = \softO(n^{\delta'/2})$. By Lemma~\ref{lem:uniquetodisjointness}, we get an algorithm with the same time bounds for \pb{Unique-$k$-Set-Intersection} on $n$ sets and universe size $\Theta(n^{\eps'/2})$. Then, using Lemma~\ref{lem:tounique}, we get an algorithm for \pb{$k$-Set-Intersection} on $n$ sets, universe size $|U|=n$, and $T = \Theta(n^{1-\eps'/2})$ with preprocessing time $\softO(T n^{k-\varepsilon/2})$ and query time $\softO(Tn^{\delta'/2}) = O(T n^{\delta'})$.
		We simplify the preprocessing time to $\softO(T n^{k-\varepsilon/2}) = \softO(n^{k+1-\eps'/2-\varepsilon/2}) = O(n^{k+1-\eps'})$, where we used $\frac{\eps'}{2}<\eps' = \frac{\eps}{2k} < \frac{\eps}{2}$. 
		The obtained preprocessing time $O(n^{k+1-\eps'})$ and query time $O(T n^{\delta'})$ contradict the statement that we got from Theorem~\ref{thm:intersection} in the last paragraph, so the \pb{Zero-$(k+1)$-Clique} Conjecture is false.
	\end{proof}
	
	\section{Removing Self-joins in Direct Access}\label{sec:selfjoins}
	
	In Section~\ref{sec:lowerdirect_star} we showed how to simulate a $k$-star query by embedding it into another query $Q$ of high incompatibility number. The construction required that $Q$ was self-join free. However, we would like to show lower bounds also for queries that have self-joins. To this end, in this section, we show how to simulate direct access to queries without self-joins by direct access to queries with self-joins in such a way that the incompatibility number---and in fact also the underlying hypergraph---does not change. This will allow us to show hardness of queries with self-joins later on. To state our result formally, we start with a definition.
	Let $Q$ be a join query, possibly with self-joins.
	We say that a query $\Qsf$ is a \emph{self-join free version} of a query $Q$ if it can be obtained from $Q$ by replacing the relation symbols such that each relation symbol appears at most once in $\Qsf$.
	The main result of this section is the following.

	\begin{theorem}\label{thm:selfjoins}
		Let $\Qsf$ be a self-join-free version of a join query $Q$, and let $L$ be a variable order for $Q$ and $\Qsf$. 
		If there is a lexicographic direct access algorithm for $Q$ by $L$ with preprocessing time $p(|D|) \in \Omega(|D|)$ and access time $r(|D|)$, then there is a lexicographic direct access algorithm for $\Qsf$ by $L$ with preprocessing time $O(p(|D|))$ and access time $O(r(|D|) \log^2 |D|)$. The same sentence holds after exchanging $Q$ and $\Qsf$.
	\end{theorem}
	
	One of the directions of Theorem~\ref{thm:selfjoins} is trivial as an algorithm for a self-join free query can be used to solve queries with self-joins with only linear overhead at preprocessing by duplicating the relevant relations, and the $\log^2(|D|)$ factor of the access time is not required. However, the other direction is by no means trivial. In particular, the same claim for the task of enumeration instead of direct access does not hold~\cite{berkholz2020tutorial,CarmeliS22}.
	We devote the rest of this section to proving this direction.
	
	\subsection{Colored and Self-Join-Free Versions}
	
	To prove that direct access for $Q$  implies direct access for $\Qsf$, it will be convenient to work with an intermediate query.
	We define the \emph{colored version} $Q^c$ of a query $Q$ to be the query obtained from $Q$ by adding the unary atom $R_x(x)$ for each variable $x$.
	The colored version is useful because it is equivalent to any self-join free version:
	
	\begin{lemma}\label{lem:color-sjf}
		Let $Q^c$ and $\Qsf$ be the colored version and a self-join-free version of the same join query $Q$.
		There are lex-preserving exact reductions in both directions between $Q^c$ and $\Qsf$.	
	\end{lemma}
	\begin{proof}
		Without loss of generality, let $\Qsf$ denote the self-join free version of $Q$ obtained by replacing each atom $R(x_1, \ldots, x_k)$ by $R_{\underline{x_1, \ldots, x_k}}(x_1, \ldots, x_k)$ (where $R_{\underline{x_1, \ldots, x_k}}$ is a new relation symbol associated to $R$ and the list of variables $x_1, \ldots, x_k$). The easy direction is the following exact reduction from $Q^c$ to $\Qsf$. Given a database $D^c$ for $Q^c$, we construct in linear time the database $D^\text{sf}$ as follows.
		For every atom $R(x_1, \ldots, x_k)$ of $Q$, we set
		\begin{align*}R_{\underline{x_1, \ldots, x_k}}^{D^\text{sf}}:= \{ (a_1, \ldots, a_k)\mid D^c \models R(a_1, \ldots, a_k)\land \bigwedge_{i\in [k]}R_{x_i}(a_i)\}.\end{align*}
		Obviously, this entails the equality $Q^c(D^c) = \Qsf(D^{\text{sf}})$ which shows the first direction.
		
		Let us construct the reduction from $\Qsf$ to $Q^c$. (Note that this is the reduction we will need in the following.) Given a database $D^{\text{sf}}$ for $\Qsf$, we construct in linear time the database $D^c$ defined as follows. We take $\dom(D^c):= \var(Q)\times \dom(D^{\text{sf}})$. For each $x\in \var(Q)$, we set
		\begin{align*}
			R^{D^c}_x := \{(x,b)\mid b\in \dom(D^{\text{sf}})\},
		\end{align*}
		and for each atom $R(x_1, \ldots, x_k)$ of $Q$, we include in $R^{ D^c}$
		\begin{align*}
			\{((x_1, a_1), \ldots, (x_k, a_k))\mid D^{\text{sf}} \models R_{\underline{x_1, \ldots, x_k}}(a_1, \ldots, a_k)\}.
		\end{align*}
		That is, $R^{ D^c}$ is defined by the union of the corresponding relations in $\Qsf$.
		
		For every tuple $(a_1, \ldots, a_k)\in (\dom(D^{\text{sf}}))^k$, we get the equivalence
		\begin{align*}
			D^{\text{sf}}\models R_{\underline{x_1, \ldots, x_k}}(a_1, \ldots, a_k) \Leftrightarrow D^c \models R((x_1, a_1), \ldots, (x_k,a_k))\land \bigwedge_{i\in [k]}R_{x_i}((x_i,a_i)).
		\end{align*}
		This immediately implies for every assignment $a$ of the variables of $\Qsf$ in $\dom(D^{\text{sf}})$ the equivalence $a\in \Qsf(D^{\text{ sf}}) \Leftrightarrow \bar a \in Q^c(D^c)$ where $\bar a$ is the mapping that assigns to every $x\in \var(Q)$ the value $(x, a(x))$. It is easy to verify that the function $a\mapsto \bar a$ is a bijection between $\Qsf(D^{\text{sf}})$ and $Q^c(D^c)$, computable in constant time and preserving any lexicographic order.
	\end{proof}

	Lemma~\ref{lem:color-sjf} implies that if there is a direct-access algorithm for $Q^c$ in some lexicographic order $L$, then there is a direct-access algorithm for $\Qsf$ in order $L$ with the same time guarantees. With this at hand, it is left to prove that efficient direct access for any query implies the same for its colored version. 
	
	\subsection{Counting Under Prefix Constraints and Direct Access}
	
	A natural approach for lexicographic direct access is to decide on the assignment to the variables one at a time in the required order. This approach has a strong connection to the task of \emph{counting} the possible assignments when a prefix of the variables is already determined, and it will be convenient in our proof to go through this intermediate task.
	
	Consider the following notion of constraints on answers. Let $Q$ be a join query, $L$ an order of its variables, and $D$ a database. We define a \emph{prefix constraint} $\mathfrak c$ on a prefix $x_1, \ldots, x_r$ of $L$ to be a mapping $\mathfrak c$ on variables $\{x_1, \ldots, x_r\}$ such that for $i\in [r-1]$ we have $\mathfrak c(x_i)\in \dom(D)$ and $\mathfrak c(x_r)$ is an interval in $\dom(D)$. We denote by $\var(\mathfrak c)$ the set $\{x_1, \ldots, x_r\}$ of variables on which $\mathfrak c$ is defined. To make the notation for prefix constraints more symmetric, we treat elements of $\dom(D)$ as intervals of length $1$, so $\mathfrak c$ maps all $x_i$ in the prefix to intervals, but for $i\in [r-1]$ these intervals have length $1$. We say that an answer $a\in Q(D)$ satisfies the prefix constraint $\mathfrak c$ if and only if for all $i\in [r]$ we have $a(x_i)\in \mathfrak c(x_i)$.
	\emph{Counting under prefix constraints} is a task defined by a query and an order where, after a preprocessing phase on a database, the user can specify a prefix constraint $\mathfrak c$ and expect the number of query answers that satisfy $\mathfrak c$. We call the time it takes to provide an answer given a prefix constraint the \emph{counting time}.
	
	Prefix constraints will be useful in the proof of Theorem~\ref{thm:selfjoins} due to the following equivalence between direct access and counting.
	
	\begin{proposition}\label{prop:prefixcounting}
		Let $Q$ be a join query and $L$ an ordering of its variables. 
		\begin{itemize}
			\item
			If lexicographic direct access for $Q$ and $L$ can be done with preprocessing time $p(|D|) \in \Omega(|D|)$ and access time $r(|D|)$, then counting under prefix constraints for $Q$ and $L$ can be done with preprocessing time $O(p(|D|))$ and counting time $O(r(|D|) \log |D|)$.
			\item
			If counting under prefix constraints for $Q$ and $L$ can be done with preprocessing time $p(|D|) \in \Omega(|D|)$ and counting time $r(|D|)$, then lexicographic direct access for $Q$ and $L$ can be done with preprocessing time $O(p(|D|))$ and access time $O(r(|D|) \log |D|)$.
		\end{itemize}
	\end{proposition}
	\begin{proof}
		To go from a direct access algorithm to a counting algorithm, note that the answers satisfying a prefix constraint appear consecutively in the ordered list of query answers. We use binary search to find the indices of the first and the last answers satisfying the given prefix constraint. The difference between these indices is the count we are looking for.
		
		For the other direction, given an index we want to access, we iteratively fix the assignments to the query variables in the specified order $L$. For each variable, we perform binary search on $\dom(D)$ by using prefix constraints to count the answers in the two halves of the interval and choose the relevant half.
	\end{proof}

	\subsection{Counting Under Prefix Constraints for the Colored Version}
	
	The final and most involved component of our proof of Theorem~\ref{thm:selfjoins} is to use an algorithm for counting under prefix constraints for any join query to solve the same task for its colored version.
	
	\begin{lemma}\label{lem:colorjoins}
		Let $Q$ be a join query and $L$ an ordering of its variables.
		If there is an algorithm for counting under prefix constraints for $Q$ and $L$, then there is an algorithm for counting under prefix constraints for the colored version $Q^c$ and $L$ with the same time guarantees.	
	\end{lemma}
	
	Notice that, by combining Lemma~\ref{lem:color-sjf} with Lemma~\ref{lem:colorjoins}, we get that an algorithm for counting under prefix constraints for a join query implies an algorithm for counting under prefix constraints for any of its self-join free versions.
	
	We remark that Lemma~\ref{lem:colorjoins} cannot simply be proven by constructing an equivalent database for $Q$ given a database for $Q^c$ by filtering the original relations using the new unary relations and then removing the unary relations, due to self-joins. As an example, if the query is 
	$Q(x,y) \datarule R(x), R(y)$, the colored version is $Q(x,y) \datarule R(x), R(y), R_x(x), R_y(y)$, and we cannot simply eliminate the last two atoms of the colored version by filtering the other atoms because we cannot have two copies of $R$ each filtered according to a different variable.
	We next prove Lemma~\ref{lem:colorjoins}.
	
	It will be convenient to have a more symmetric perspective on join queries. To this end, we remind the reader of some basics on homomorphisms between finite structures. Let $A$ and $B$ be two finite structures over the same vocabulary. Then $h:\dom(A)\rightarrow \dom(B)$ is called a \emph{homomorphism} from $A$ to $B$ if and only if, for every relation $R^A$ of $A$ and every tuple $(a_1, \ldots, a_k)\in R^A$, we have $(h(a_1), \ldots, h(a_k))\in R^B$. To simplify notation, we make the convention $h(a) := (h(a_1), \ldots, h(a_k))$ for $a= (a_1, \ldots, a_k)$. An \emph{automorphism} of $A$ is a bijective homomorphism from $A$ to itself. We will use the fact that homomorphisms compose, i.e., if $h_1$ is a homomorphism from $A$ to $B$ and $h_2$ is a homomorphism from $B$ to $C$, then $(h_2 \circ h_1)$ defined for every $a\in \dom(A)$ by $(h_2 \circ h_1)(a):= h_2(h_1(a))$ is a homomorphism from $A$ to $C$.
	
	Every query $Q$ is assigned a finite structure $A_Q$ on domain $\var(Q)$ as follows: for every $k$-ary relation symbol $R$ of $Q$, the structure $A_Q$ has a $k$-ary relation $R^{A_Q}$ which contains exactly the tuples $(x_{i_1}, \ldots, x_{i_k})$ such that $R(x_{i_1}, \ldots, x_{i_k})$ is an atom of $Q$. Then $h$ is an answer to $Q$ on a database $D$ if and only if it is a homomorphism from $A_Q$ to $D$.
	The notions on prefix constraints carry over from query answers to homomorphisms in the obvious way.
	
	\begin{example}\label{ex:queryStructure}
		Consider the join query $Q(x,y,z) \datarule R(x), R(y), R(z)$ which we will use as a running example in this section. The associated structure $A_Q$ has the domain $\{x, y, z\}$ and the single relation $R^{ A_Q}= \{(x),(y),(z)\}$. There are $3^3$ homomorphisms from $A_Q$ to itself (in this example, any variable can map to any variable), but there are only $3!$ automorphisms. In particular, if we consider the order $L=(x,y,z)$ of the variables and we denote by $(a,b,c)$ the mapping $\{x\mapsto a, y\mapsto b, z\mapsto c\}$ for any values $a,b,c$, we get that the automorphisms correspond to the vectors of permutations of the vertices.
	\end{example}
	
	\paragraph*{Notation.} For a join query $Q$ to a database $D$, an order $L$ of $\var(Q)$, and a prefix constraint $\mathfrak{c}$, we let $\hom(A_Q, D, \mathfrak{c})$ denote the set of homomorphisms from $A_Q$ to $D$ (corresponding to the set of answers in $Q(D)$) which satisfy $\mathfrak{c}$.
	
	Let $D^c$ be a database on which we want to solve the counting problem for $Q^c$.
	We construct a new database $D$ over the relations of $Q$ (without the additional unary relations) as follows by tagging the values of the database by variable names. 
	For every relation symbol $R$ of arity $k$ of $Q$, the database $D$ contains the relation:
	\begin{align*}
		R^D:= \{((a_1, b_1), \ldots, (a_k, b_k))\mid (a_1, \ldots, a_k)\in R^{A_Q}, (b_1, \ldots, b_k)\in R^{D^c}, \forall i\in [k]: b_i \in R_{a_i}^{D^c}\}.
	\end{align*}
	The domain is $\dom(D) = \{(a, b) \in \var(Q)\times \dom(D^c)\mid b \in R_a^{D^c}\}$ which we order lexicographically, so 
	$(a,b) \prec (a', b') \Leftrightarrow a \prec_L a' \lor (a = a' \land b \prec b')$, where $\prec_L$ is the order relation of~$L$.
	
	\begin{example}
		Consider again the query from Example~\ref{ex:queryStructure}.
		The structure $A_{Q^c}$ for the colored version of the query has the relations $R^{A_{Q^c}}=\{(x),(y),(z)\}$, $R_x^{A_{Q^c}}=\{(x)\}$, $R_y^{A_{Q^c}}=\{(y)\}$, and $R_z^{A_{Q^c}}=\{(z)\}$.
		Given a database $D^c$ with $R^{D^c}=\{(a),(b),(c)\}$, $R_x^{D^c}=\{(a),(c)\}$, $R_y^{D^c}=\{(a),(b)\}$, and $R_z^{D^c}=\{(a)\}$, we construct the database $D$ with $R^D=\{(x,a),(x,c),(y,a),(y,b),(z,a)\}$.
	\end{example}
	
	Let $\pi_1:\dom(D)\rightarrow \var(Q)$ be the projection on the first coordinate, i.e., for all $(a,b)\in \dom(D)$ we have $\pi_1(a,b) = a$. Analogously, we define $\pi_2: \dom(D)\rightarrow \dom(D^c)$ as the projection to the second coordinate. We will use the following simple observation throughout the remainder of this proof.
	
	\begin{observation}\label{obs:automorphism}
		Let $h$ be a homomorphism from $A_Q$ to $D$. Then we have $(\pi_1\circ h)(\var(Q)) = \var(Q)$ if and only if $\pi_1\circ h$ is an automorphism.
	\end{observation}
	
	\begin{example}		
		Consider $h_1=((y,a),(x,a),(z,a))$ and $h_2=((y,a),(y,b),(z,a))$ in our running example, following the homomorphism notation given in Example~\ref{ex:queryStructure}. We have that $\pi_1\circ h_1$ is an automorphism, while $\pi_1\circ h_2$ is not.
	\end{example}
	
	Now let $\mathfrak c$ be a prefix constraint on a prefix $x_1, \ldots, x_r$ of $L$., i.e., we want to count the number of answers $h$ in $Q^c(D^c)$ which satisfy $\mathfrak c$. Note that this subset of the answers can equivalently be seen as the set of homomorphisms $h$ from $A_{Q^c}$ to $D^c$ which satisfy $\mathfrak c$. We denote this set by $\hom(A_{Q^c}, D^c, \mathfrak c)$.
	We also denote by $\Nid$ the set of homomorphisms $h'$ from $A_Q$ to $D$ for which $\pi_1\circ h'$ is the identity, denoted $\id$, on $\var(Q)$, and $\pi_2\circ h'$ satisfies $\mathfrak c$.  We show that there is a bijection between these two sets of homomorphisms.
	
	\begin{claim}\label{clm:bijection}
		\[|\hom(A_{Q^c}, D^c, \mathfrak c)| = |\Nid|\]
	\end{claim}
	\begin{example}
		In our running example, consider the prefix constraint $\mathfrak c$ assigning $a$ to $x$.
		We have that $\Nid=\{((x,a),(y,a),(z,a)),((x,a),(y,b),(z,a))\}$, and for the colored version we have that $\hom(A_{Q^c}, D^c, \mathfrak c)=\{(a,a,a),(a,b,a)\}$.
	\end{example}
	\begin{proof}[Proof of Claim~\ref{clm:bijection}]
		We show a bijection between $\hom(A_{Q^c}, D^c, \mathfrak c)$ and $\Nid$.
		For every $h\in \hom(A_{Q^c}, D^c, \mathfrak c)$, define $P(h):= h'$ where for every $a\in \var(Q)$ we set $h'(a):= (a, h(a))$. We claim that $P$ is the required bijection.
		
		We first show that $P$ is a mapping from $\hom(A_{Q^c}, D^c, \mathfrak c)$ to $\Nid$.
		We show that $h'$ is a homomorphism from $A_Q$ to $D$. Consider any relation $R^{A_{Q}}$ of $A_Q$ and a tuple $(a_1, \ldots, a_k)\in R^{A_{Q}}$. Then $h'(a_1, \ldots, a_k) = ((a_1, h(a_1)), \ldots, (a_k, h(a_k)))$ by definition. We have that $(h(a_1), \ldots, h(a_k))\in R^{D^c}$, because $h$ is a homomorphism.
		By definition of $Q^c$, we have that $a_i\in R_{a_i}^{A_{Q^c}}$ for all $1\le i\le k$. As $h$ is a homomorphism from $A_{Q^c}$ to $D^c$, we have that $h(a_i)\in R_{a_i}^{D_c}$. By definition of $D$, we conclude that $h'(a_1, \ldots, a_k)\in R^D$, and thus $h'$ is indeed a homomorphism.
		By definition of $h'$, we have $\pi_1(h'(a))= a$ for every $a\in \var(Q)$, so $\pi_1\circ h'$ is the identity on $\var(Q)$. Moreover, since $h\in \hom(A_{Q^c}, D^c, \mathfrak c)$, we have that $h$ satisfies $\mathfrak c$ and with $h = \pi_2\circ h'$ it follows that $\pi_2\circ h'$ satisfies $\mathfrak c$. So we have that $h'\in \Nid$.
		
		It remains to show that $P$ is a bijection. 
		To show that $P$ is injective, let $h_1, h_2\in \hom(A_{Q^c}, D^c, \mathfrak c)$ be two different functions. There is $a\in \var(Q)$ such that $h_1(a) \ne h_2(a)$. We get that $P(h_1)(a) = (a, h_1(a)) \ne (a, h_2(a)) = P(h_2)(a)$, so $P(h_1)\ne P(h_2)$, and thus $P$ is indeed injective.
		To show that $P$ is surjective, let $h'\in \Nid$. Then, by definition of $D$, there is a function $h:\var(Q)\rightarrow \dom(D^c)$ such that for every $a\in \var(Q)$ we have $h'(a) = (a, h(a))$. We claim that $h\in \hom(A_{Q^c}, D^c, \mathfrak c)$. This will prove surjectivity since $P(h)= h'$. First note that $h=\pi_2\circ h'$ and, by definition of $\Nid$, $\pi_2\circ h'$ satisfies $\mathfrak c$, so $h$ satisfies $\mathfrak c$. To see that $h$ is a homomorphism, consider first the unary atoms of $Q^c$ of the form $R_a(a)$. Since for every $a\in \var(Q)$ we have $h'(a)= (a, h(a))\in \dom(D)$, we have by definition of $\dom(D)$ that $h(a)\in R_{a}^{D^c}$ and thus the requirement for homomorphisms is verified for $R_a$. Now let $R$ be one of the other relations of $A_{Q^c}$, so $R$ is a relation of $A_Q$ and $R^{A_{Q^c}}= R^{A_Q}$. Let $(a_1, \ldots, a_k)\in R^{A_Q}$, then $h'(a_1, \ldots, a_k)= ((a_1, h(a_1)), \ldots, (a_k, h(a_k)))\in R^D$. From the definition of $D$, it follows then directly that $h(a_1, \ldots, a_k)\in R^{D^c}$ which shows that $h$ is a homomorphism as claimed.
	\end{proof}
	
	We conclude that it is enough to compute $|\Nid|=|\hom(A_{Q^c}, D^c, \mathfrak c)|$ to answer the counting problem for $\mathfrak c$. To do this, we find its connection to other counting tasks that we can directly compute.
	Let $\aut(A_Q, \mathfrak c)$ denote the set of automorphisms of $A_Q$ that are the identity on all $a\in \var(\mathfrak c)$. Moreover, let $\Naut$ be the set of homomorphisms $h'$ from $A_Q$ to $D$ such that $\pi_1\circ h'\in \aut(A_Q, \mathfrak c)$ and $\pi_2\circ h'$ satisfies $\mathfrak c$. In particular, if $h'\in \Naut$, then $(\pi_1\circ h')(\var(Q))= \var(Q)$.

	\begin{claim}\label{clm:sizes-eq}
		\begin{align*}
			|\Naut| = |\Nid|\cdot|\aut(A_Q, \mathfrak c)|.	
		\end{align*}
	\end{claim}
	\begin{example}
		We have seen that $|\Nid|=\{((x,a),(y,a),(z,a)),((x,a),(y,b),(z,a))\}$ in our running example. We also have that
		$\aut(A_Q, \mathfrak c)=\{(x,y,z),(x,z,y)\}$, and $\Naut$ is $\{((x,a),(y,a),(z,a)),((x,a),(y,b),(z,a)),((x,a),(z,a),(y,a)),((x,a),(z,a),(y,b))\}$. We get that $|\Naut| = 4 = 2\cdot 2= |\Nid|\cdot|\aut(A_Q, \mathfrak c)|$.
	\end{example}
	\begin{proof}[Proof of Claim~\ref{clm:sizes-eq}]
		We first show 
		$\Naut = \{h\circ g\mid h\in \Nid, g\in \aut(A_Q, \mathfrak c)\}.$
		
		For the containment ``$\supseteq$'', let $h\in \Nid$ and $g\in \aut(A_Q, \mathfrak c)$.
		Note first that $h\circ g$ is a homomorphism from $A_Q$ to $D$.
		Next, as $\pi_1\circ h$ is the identity, we have that $\pi_1\circ h\circ g\in \aut(A_Q, \mathfrak c)$.
		By definition of $\aut(A_Q, \mathfrak c)$, we have $g|_{\var(\mathfrak c)}= \id$, so $(\pi_2\circ h\circ g)|_{\var(\mathfrak c)} = (\pi_2\circ h)|_{\var(\mathfrak c)}$, and $\pi_2\circ h\circ g$ satisfies $\mathfrak c$ because $h\in\Nid$. Thus, $h\circ g\in \Naut$ which proves the containment.
		
		For the direction ``$\subseteq$'', let $h'\in \Naut$. We define $g := \pi_1\circ h'\in \aut(A_Q, \mathfrak c)$, so $g^{-1}$ is well-defined, and $g^{-1}\in \aut(A_Q, \mathfrak c)$ since the inverse of an automorphism is an automorphism as well. Now define $h:= h'\circ g^{-1} = h'\circ (\pi_1\circ h')^{-1}$. We claim that $h\in \Nid$. First, as a composition of homomorphisms, $h$ is a homomorphism as well. Moreover, $\pi_1\circ h = (\pi_1\circ h')\circ (\pi_1\circ h')^{-1} = \id$. Also, we have that $g^{-1}$ is the identity on $\var(\mathfrak c)$, so $(\pi_2 \circ h)|_{\var(\mathfrak c)} = (\pi_2\circ h'\circ g^{-1})|_{\var(\mathfrak c)} = (\pi_2\circ h')|_{\var(\mathfrak c)}$. Since $\pi_2\circ h'$ satisfies $\mathfrak c$, we have that $\pi_2 \circ h$ does too. It follows that $h\in \Nid$, as claimed, which proves the containment in the second direction.
		
		To prove the claim, we show that for $h_1, h_2\in \Nid$ and $g_1, g_2\in \aut(A_Q, \mathfrak c)$, whenever $h_1\ne h_2$ or $g_1\ne g_2$, then $h_1\circ g_1\ne h_2\circ g_2$. Assume first that $g_1\ne g_2$. Since $\pi_1\circ h_1 =\pi_1\circ h_2 = \id$, we have $\pi_1\circ h_1\circ g_1 = g_1 \ne g_2 = \pi_1\circ h_2\circ g_2$. So $h_1\circ g_1$ and $h_2\circ g_2$ differ in the first coordinate. Now assume that $g_1= g_2 = g$ but $h_1\neq h_2$. Let $a\in \var(Q)$ be such that $h_1(a)\ne h_2(a)$. Since $g$ is a bijection, the inverse $g^{-1}$ exists and is a bijection as well. Set $a':= g^{-1}(a)$. Then $(h_1\circ g)(a') = (h_1\circ g \circ g^{-1})(a)= h_1(a) \ne h_2(a) = (h_2\circ g \circ g^{-1})(a) = (h_2\circ g)(a')$. So again $h_1\circ g_1 \ne h_2\circ g_2$.
	\end{proof}
	
	Next we show that we can compute $|\Naut|$ efficiently.
	
	\begin{claim}\label{clm:computing-size}
		Given an algorithm $\mathcal A$ for counting under prefix constraints for~$Q$, there is an algorithm computing $|\Naut|$ with the following properties:
		\begin{itemize}
			\item The preprocessing phase is given $D$ (but not $\mathfrak c$) and runs in time $O(|D|)$, making $O(1)$ calls to the preprocessing of $\mathcal A$ on databases of size $O(|D|)$. 
			\item The query phase is given a prefix constraint $\mathfrak c$ and computes $|\Naut|$ in time $O(1)$, making $O(1)$ calls to the query of $\mathcal A$.
		\end{itemize}
	\end{claim}
	\begin{proof}
		Given $\mathfrak c$, let $\mathfrak c^*$ be the prefix constraint on the same prefix $x_1, \ldots x_r$ as $\mathfrak c$, defined by $\mathfrak c^*(x_i):= [(x_i, y_i^1), (x_i, y_i^2)]$ where $\mathfrak c(x_i)= [y_i^1, y_i^2]$. Note that, since we assume $\dom(D)$ to be lexicographically ordered, $\mathfrak c^*$ is indeed a prefix constraint for the database $D$. By definition, $\hom(A_Q, D, {\mathfrak c}^*)$ is the set of homomorphisms $h$ from $A_Q$ to $D$ that satisfy $\mathfrak c^*$, or, equivalently, such that $\pi_2\circ h$ satisfies $\mathfrak{c}$ and $\pi_1\circ h$ is the identity over $\var(\mathfrak c)$.
		
		Note that $\Naut \subseteq \hom(A_Q, D, \mathfrak c^*)$ but there are homomorphisms $h\in \hom(A_Q, D, \mathfrak c^*)$ that are not in $\Naut$ because $\pi_1\circ h$ is not a bijection (see Example~\ref{ex:notABijection}). This motivates for every subset $\var(\mathfrak c)\subseteq T\subseteq \var(Q)$ the definition
		\begin{align*}
			\mathcal N_T := \{h \in \hom(A_Q, D, \mathfrak c^*) \mid (\pi_1\circ h)(\var(Q))\subseteq T\}.
		\end{align*}
		The set $\Naut$ consists exactly of homomorphisms $h\in \hom(A_Q, D, \mathfrak c^*)$ for which $(\pi_1\circ h)$ is a bijection, which by Observation~\ref{obs:automorphism} is equivalent to $(\pi_1\circ h)(\var(Q)) = \var(Q)$. 
		
		\begin{example}\label{ex:notABijection}
			We continue our running example.
			Given $\mathfrak c$ as before, we have that $\mathfrak c^*$ is the prefix constraint that assigns $(x,a)$ to $x$.
			An example for a homomorphism in $\hom(A_Q, D, \mathfrak c^*)$ but not in $\Naut$ is $((x,a),(y,b),(y,c))$. This homomorphism is also in $\mathcal N_{\{x,y\}}$ and $N_{\{x,y,z\}}$.
		\end{example}
		
		Using the fact that for all $h\in \hom(A_Q, D, \mathfrak c^*)$, the restriction of $\pi_1\circ h$ to $\var(\mathfrak c)$ is by definition the identity, we first compute the value $|\Naut|$ based on the $|\mathcal N_T|$ values. To this end, we use the abbreviations $V:= \var(Q)$ and $C:= \var(\mathfrak c)$, the obvious equality $|\Naut| = \mathcal N_V\setminus (\bigcup_{x\in V\setminus C} \mathcal N_{V\setminus \{x\}})$ and the inclusion-exclusion principle. We obtain
		\begin{align*}
			\left|\bigcup_{x\in V\setminus C} \mathcal N_{V\setminus \{x\}}\right| &= 
			\sum_{i=1}^{|V\setminus C|}\sum_{X\subseteq V\setminus C \colon |X|  = i}  (-1)^{i-1} |\mathcal N_{V\setminus X}|= \sum_{C\subseteq T\subsetneq V} (-1)^{|V\setminus T| -1} |\mathcal N_T|
		\end{align*}
		and 
		\begin{align*}
			|\Naut| = |\mathcal N_V| - \left|\bigcup_{x\in V\setminus C} \mathcal N_{V\setminus \{x\}}\right| = \sum_{C\subseteq T\subseteq V} (-1)^{|V\setminus T|} |\mathcal N_T|.
		\end{align*}
		Setting $v:= |\var(Q)|$, this gives 
		\begin{align}\label{eq:naut}
			|\Naut| = \sum_{\var(\mathfrak{c})\subseteq T\subseteq \var(Q)} (-1)^{v-|T|} |\mathcal N_T|.
		\end{align}
		(For the sake of simplicity in the rest of the proof, we will continue using the notation $v:= |\var(Q)|$ and $r:= |\var(\mathfrak{c})|$. Of course we have $0\le r\le v$.)
		
		The number of summands is bounded by an interger depending only on $v:=|\var(Q)|$, so if we can efficiently compute each $|\mathcal N_T|$ for $T\subseteq \var(Q)$ with $\var(\mathfrak{c})\subseteq T$, we are done.
		
		Now, for every $i\in [0,v]$, define $\mathcal N_{T,i}$ to be the set of homomorphisms $h\in\hom(A_Q, D, \mathfrak c^*)$ such that for exactly $i$ elements $a\in \var(Q)$ we have that $(\pi_1\circ h)(a)\in T$. With this notation, $\mathcal N_T = \mathcal N_{T,v}$. Then, for every $j\in [v+1]$, construct a new database $D_{T,j}$ over the same relation symbols as $D$. To this end, for every $a\in T$, introduce $j$ clones $a^{(1)}, \ldots, a^{(j)}$. Then the domain of $D_{T,j}$ is
		\begin{align*}
			\dom(D_{T,j}):= \{(a^{(k)},b)\mid (a,b)\in \dom(D), a\in T, k\in [j]\}  \cup \{(a,b)\in \dom(D)\mid a\notin T\}
		\end{align*}
		where we assume the lexicographic order between the new elements as follows:
		\begin{align*}
			(a^{(i)}, b) \prec (c^{(\ell)},d) \Leftrightarrow i< \ell \lor (i=\ell \land a\prec_L c) \lor (i=\ell \land a=c \land b \prec d).
		\end{align*}
		We add the elements in $\{(a,b)\in \dom(D)\mid a\notin T\}$ in an arbitrary way to this order.
		
		We then define a mapping $B:\dom(D)\rightarrow \mathcal P(\dom(D_{T,j}))$ where $\mathcal P(\dom(D_{T,j}))$ is the power set of $\dom(D_{T,j})$ as follows:
		\begin{align*}
			B(a,b) := \begin{cases}
				\{(a^{(k)},b)\mid k\in [j]\}, &\text{ if } a\in T\\
				\{(a,b)\}, & \text{ otherwise.}
			\end{cases}
		\end{align*}
		Then, for every relation symbol $R$ of $Q$ of arity $s$, define 
		\begin{align*}
			R^{D_{T,j}}:= \bigcup_{((a_1, b_1), \ldots, (a_s, b_s))\in R^D} B(a_1,b_1)\times \ldots \times B(a_s, b_s).
		\end{align*}
		Define a new prefix constraint $\mathfrak c^{**}$ on the variables $\var(\mathfrak c)$ by setting for every $a\in \var(\mathfrak c)$ the intervals $\mathfrak c^{**}(a)= [(a^{(1)}, y^1), (a^{(1)}, y^2)]$ where $\mathfrak c(a)= [y^1, y^2]$. 
		
		Note that for every $h\in \mathcal N_{T,i}$, since $(\pi_1\circ h)(\var(\mathfrak c))= \var(\mathfrak c) \subseteq T$, there are exactly $i-r$ elements $a\in \var(Q)\setminus \var(\mathfrak c)$ with $(\pi_1\circ h)(a) \in T$. For every $i\in [0,v]$ and $h\in \mathcal N_{T,i}$, each variable $a\in \var(Q)\setminus \var(\mathfrak c)$ with $(\pi_1\circ h)(a)\in T$ can be mapped to one of $j$ copies to get a homomorphism  in $\hom(A_Q, D_{T,j}, \mathfrak c^{**})$.
		Thus, every $h\in \mathcal N_{T,i}$ corresponds to exactly $j^{i-r}$ homomorphisms $h_{k_1, \ldots, k_{i-r}}$ in $\hom(A_Q, D_{T,j}, \mathfrak c^{**})$ for $(k_1, \ldots, k_{i-r})\in [j]^{i-r}$, defined by the following expression where $\ell\in [i-r]$:
		\begin{align*}
			h_{k_1, \ldots, k_{i-r}}(x) := \begin{cases}
				(x^{(1)}, b), & \text{ if $x\in \var(\mathfrak c)$ and } h(x)= (x,b),\\
				(a^{(k_\ell)},b) & \text{ if $x$ is the $\ell$-th element of $(\pi_1\circ h)^{-1}(T)\setminus \var(\mathfrak c)$ and $h(x) = (a,b)$},\\
				h(x), & \text{ otherwise, i.e., if $x\notin (\pi_1\circ h)^{-1}(T)$.}
			\end{cases}
		\end{align*}
		The following two points are immediate:
		\begin{enumerate}
			\item for all homomorphisms $h, h'\in \mathcal N_{T,i}$ and tuples $(k_1, \ldots, k_{i-r})$ and $(k_1', \ldots, k_{i-r}')$ in $[j]^{i-r}$ such that $(h,k_1, \ldots, k_{i-r}) \ne (h',k_1', \ldots, k_{i-r}')$, the resulting homomorphisms $h_{k_1, \ldots, k_{i-r}}$ and $h'_{k_1', \ldots, k_{i-r}'}$ in $\hom(A_Q, D_{T,j}, \mathfrak c^{**})$ are not equal; 
			\item\label{item:essential2} every homomorphism $H\in \hom(A_Q, D_{T,j}, \mathfrak c^{**})$ is a homomorphism $h_{k_1, \ldots, k_{i-r}}$ for exactly one $i\in [r,v]$, one $h\in \mathcal N_{T,i}$ and one tuple $(k_1, \ldots, k_{i-r}) \in[j]^{i-r}$. (To show this, we have $H= h_{k_1, \ldots, k_{i-r}}$, with each $k_\ell, \ell\in [i-r]$ determined as follows: if the $\ell$-th element of $(\pi_1\circ h)^{-1} (T)\setminus \var(\mathfrak{c})$ is $x$ with $H(x) = (a^{(k)}, b)$, then $k_\ell := k$.)
		\end{enumerate}
		
		Together, the points (1) and (2) imply, for every $j\in [v-r+1]$, the equality:
		\begin{align}\label{eq:vandermonder}
			\sum_{i=r}^{v} j^{i-r}|\mathcal N_{T,i}| = |\hom(A_Q, D_{T,j}, \mathfrak c^{**})|.
		\end{align}
		This is a linear system of $v-r+1$ equations with $v-r+1$ unknowns $|\mathcal N_{T,i}|$; the coefficients $j^{i-r}$, for $i\in [r,v]$ and $j\in [v-r+1]$, form a square Vandermonde matrix which is therefore invertible. For each fixed $T$, this allows to compute the values $|\mathcal N_{T,i}|$ (and in particular the value $|\mathcal N_{T,v}| = |\mathcal N_T|$ which interests us) in time $O((v-r+1)^3)$, which is constant time, if the $v-r+1$ values $|\hom(A_Q, D_{T,j}, \mathfrak c^{**})|$ are known.
		
		Since the values $|\hom(A_Q, D_{T,j}, \mathfrak c^{**})|$ are just counting queries under a prefix constraint, we can compute them using the algorithm $\mathcal A$. We note that, by the definition of the relations $R^{D_{T,j}}$, each database $D_{T,j}$ with $j\in [v+1]$ has size at most $j^\alpha|D|$ where $\alpha$ is the maximum arity of the relations, so that $|D_{T,j}| = O(|D|)$ and can be easily constructed from $D$ in linear time.
		
		To summarize, we can compute the required value $|\Naut|$ by the following two-phase algorithm:
		\begin{itemize}
			\item Preprocessing (independent of the prefix constraint): Compute the databases $D_{T,j}$ for all $\var(\mathfrak c)\subseteq T\subseteq \var(Q)$ and $j\in [v+1]$, and apply the preprocessing phase of $\mathcal A$ to each $D_{T,j}$; 
			\item Query phase: Read the prefix constraint $\mathfrak{c}$ and turn it into $\mathfrak{c}^{**}$; compute $|\hom(A_Q, D_{T,j}, \mathfrak c^{**})|$ by applying the query phase of $\mathcal A$ to $D_{T,j}$ and $\mathfrak c^{**}$ for all $j\in [v-r+1]$ and all $\var(\mathfrak c)\subseteq T\subseteq \var(Q)$; for each such $T$, compute (from the $|\hom(A_Q), D_{T,j}, \mathfrak c^{**})|$ values) the (unique) solution $|\mathcal N_{T,v}| = |\mathcal N_T|$ of the linear system of equations (\ref{eq:vandermonder}); finally, from the $|\mathcal N_T|$ values, compute $|\Naut|$ by equation~(\ref{eq:naut}).
		\end{itemize}
		
		Let us analyze the complexity of the above algorithm: the number of calls to the preprocessing phase and the query phase of $\mathcal A$ made by the above algorithm is the number of pairs $(T,j)$ involved, which depends only on $Q$. Therefore the preprocessing time of the algorithm is $O(|D|)$ and its query time is constant (both when ignoring the runtime of the calls to $\mathcal A$).
	\end{proof}
	
	We can now complete the proof of Lemma~\ref{lem:colorjoins}.
	From Claims \ref{clm:bijection} and \ref{clm:sizes-eq} we get that what we need to compute is 
	\begin{align*}
		|\hom(A_{Q^c}, D^c, \mathfrak c)| = |\Nid| = \frac{|\Naut|}{|\aut(A_Q, \mathfrak c)|}.
	\end{align*}
	From Claim~\ref{clm:computing-size}, $|\Naut|$ can be computed efficiently.
	Since we work in data complexity, we can easily compute $\aut(A_Q, \mathfrak c)$: the value only depends on $A_Q$ and the length $r$ of the prefix that $\mathfrak c$ is on, so we can simply compute it for every possible value of $r$. 
	This completes the proof of Lemma~\ref{lem:colorjoins}.\qed
	
	\subsection{Assembling the Proof}
	
	With all the components at hand, we can now prove Theorem~\ref{thm:selfjoins}.
	
	\begin{proof}[Proof of Theorem~\ref{thm:selfjoins}]
		By Proposition~\ref{prop:prefixcounting}, a lexicographic direct access algorithm for $Q$ and~$L$ implies an algorithm for counting under prefix constraints for $Q$ and $L$. By Lemma~\ref{lem:colorjoins}, this implies an algorithm for counting under prefix constraints for $Q^c$ and $L$. By Proposition~\ref{prop:prefixcounting}, this implies a lexicographic direct access algorithm for $Q^c$ and $L$. By Lemma~\ref{lem:color-sjf}, this finally implies a lexicographic direct access algorithm for $\Qsf$ and $L$. Each step increases the preprocessing time only by constant factors, and the access/counting is increased by a logarithmic factor whenever we use Proposition~\ref{prop:prefixcounting}. This proves the non-trivial direction of the theorem.
	\end{proof}
	
	\section{Tight Bounds for Direct Access}
	\label{sec:together}
	
	In this section, we formulate and prove the main result of this paper on direct access, showing that the incompatibility number determines the required preprocessing time, even for queries with self-joins.
	
	\begin{theorem}\label{thm:main}
		Let $Q$ be a join query, let $L$ be an ordering of its variables, and let ~$\iota$ be the incompatibility number of $Q$ and $L$.
		\begin{itemize}
			\item There is an algorithm that allows lexicographic direct access with respect to the order induced by $L$ with preprocessing time $O(|D|^\iota)$ and logarithmic access time.
			\item Assuming the \pb{Zero-Clique} Conjecture, there is no constant $\varepsilon > 0$ such that for all $\delta > 0$ there is an algorithm allowing lexicographic direct access with respect to the order induced by $L$ with preprocessing time $O(|D|^{\iota-\eps})$ and access time $O(|D|^\delta)$.
		\end{itemize}
	\end{theorem}
	\begin{proof}
		The first part was shown in Theorem~\ref{thm:cover-preprocess}. Let us now prove the second part first for $\iota > 1$. If there is an $\eps>0$ such that for all $\delta>0$ we have direct access for $Q$ and $L$ with preprocessing time $O(|D|^{\iota-\eps})$ and access time $O(|D|^\delta)$, then by Theorem~\ref{thm:selfjoins}, there is a similar algorithm for a self-join free version of $Q$.  
		By Lemma~\ref{lem:fractionalstar}, there is a $k \ge 2$ and $\eps'>0$ such that for all $\delta'>0$ the star query $\qstar_k$ with respect to a bad order admits direct access with preprocessing time $O(|D|^{k-\eps'})$ and access time $O(|D|^{\delta'})$. Combining Proposition~\ref{prop:testing} and Lemma~\ref{lem:reductiontosetdisjointness}, we obtain that $k$-Set-Disjointness has an algorithm with preprocessing time $O(\|I\|^{k-\eps'})$ and query time $\softO(\|\I\|^{\delta'}) = O(\|\I\|^{\delta''})$ for $\delta'' := 2\delta'$. The quantifiers on $\eps'$ and $\delta''$ match Theorem~\ref{thm:lowerksetdisjointness}, which implies that the \pb{Zero-$(k+1)$-Clique} Conjecture is false, in particular the \pb{Zero-Clique} Conjecture is false.
		
		Since the incompatibility number is always at least $1$, it remains to prove the second part for $\iota = 1$.
		Here we start from the well-known fact that searching in an unordered array requires linear time. More precisely, given an array $A[1..n]$ with $\{0,1\}$-entries, deciding whether $A$ contains a 0 requires time $\Omega(n)$, and this holds for deterministic algorithms as well as for randomized algorithms (with success probability at least 2/3). We reduce this problem to direct access for $Q$, in order to show a linear lower bound on the preprocessing time. To this end, given an array $A[1..n]$ with $\{0,1\}$-entries we construct a database $D$ with $\dom(D) = [2n]$. Set $B := \{i + A[i]\cdot n \mid i \in [n]\}$. For any relation symbol $R$ of $Q$, we set $R^D := \{(b,\ldots,b) \mid b \in B\}$. Note that every $(b,\ldots,b)$ with $b \in B$ is a solution to $Q(D)$, but there can also be other solutions. In any case, the lexicographically smallest solution to $Q(D)$ is the tuple $(b^*,\ldots, b^*)$ with $b^* := \min(B)$. Observe that $b^* \le n$ if and only if $A[1..n]$ has a 0-entry, so we can read off whether array $A$ has a 0-entry from $b^*$.
		Hence, by preprocessing $D$ and then using one direct access to determine the lexicographically smallest solution, we can decide whether the array $A[1..n]$ contains a 0-entry. Note that this reduction is implicit, meaning we do not need to explicitly compute the database $D$. Instead, when the preprocessing or direct access algorithm reads an entry of a relation $R^D$, we can compute this entry on the fly in constant time by reading an entry $A[i]$. Therefore, a direct access algorithm for $Q$ with preprocessing time $o(|D|)$ and access time $o(|D|)$ would imply that we can search in an unordered array in time $o(n)$. Since the latter is impossible, there is no direct access algorithm for $Q$ with preprocessing time $o(|D|)$ and access time $o(|D|)$. In particular, for any $\eps \in (0,1)$ and $\delta := 1-\eps$ there is no direct access algorithm with preprocessing time $O(|D|^{1-\eps})$ and access time $O(|D|^\delta)$. This lower bound is unconditional, proving the second part for $\iota = 1$.
	\end{proof}
	
	Note that for every $\delta,\delta' > 0$ we have $O((\log |D|)^\delta) \le O(|D|^{\delta'})$. Therefore, if for some $\varepsilon,\delta>0$ there was an algorithm with preprocessing time $O(|D|^{\iota-\varepsilon})$ and access time $O((\log |D|)^\delta)$, then for every $\delta'>0$ there would also be an algorithm with preprocessing time $O(|D|^{\iota-\varepsilon})$ and access time $O(|D|^{\delta'})$. Thus, Theorem~\ref{thm:main} implies the lower bound for polylogarithmic query time stated in the introduction.

	\section{Relaxed Orders and Projections}\label{sec:relaxations}

	So far, we have considered the case where the underlying order of the results is lexicographic and specified by the user as part of the problem definition.
	We now consider the relaxation of this requirement, when the order is either: (1) completely flexible; (2) has to be lexicographic, but the ordering of the variables is flexible; or (3) the answers must be ordered lexicographically by some of the variables, and breaking ties originating from the assignments to the other variables can be done in an arbitrary order. We will see that such relaxations allow to reduce the preprocessing time.
	We also briefly discuss how to support queries with projections.
	
	\subsection{Unconstrained Lexicographic Orders}
	
	From Theorem~\ref{thm:main} we can also get a lower bound for direct access for unconstrained lexicographic order based on fractional hypertree width.
	
	\begin{proposition}\label{prop:width-min-incompatibility}
		The fractional hypertree width of a hypergraph $H$ is the minimum incompatibility number of $H$ minimized over all orders $L$.
	\end{proposition}
	\begin{proof}
		First, we claim that for any hypertree decomposition $\mathcal{D}$, there exists an order $L$ of the variables, such that $\mathcal{D}$ and $L$ have no disruptive trio.
		Indeed, since $\mathcal{D}$ is acyclic, it has an elimination order, and the reverse of an elimination order is known to not have disruptive trios.
		From this claim we conclude that by considering all permutations $L$ and all hypertree decompositions that do not have a disruptive trio with respect to $L$, we get all hypertree decompositions of $H$. 
		
		Denote the set of all permutations of the vertices by $P$ and the set of all hypertree decompositions of $H$ by $D$.
		Given $\mathcal{D}\in D$ and $L\in P$, denote the fractional width of $D$ by $w_\mathcal{D}$, the set of all hypertree decompositions of $H$ that have no diruptive trio with respect to $L$ by $D_L$, and the incompatibility number of $\mathcal{D}$ and $L$ by $\iota(\mathcal{D},L)$.
		Denote also the fractional hypertree width of $H$ by $w$.
		Then,
		
		\begin{align*}
			w = \min_{\mathcal{D}\in D}{w_\mathcal{D}}
			= \min_{L\in P}{\min_{\mathcal{D}\in D_L}{w_\mathcal{D}}}
			= \min_{L\in P, \mathcal{D}\in D}{\iota(\mathcal{D},L)}
		\end{align*}
		Here, the first equality is the definition of fractional hypertree width, the second equality follows from the claim we just proved, and the third equality follows from Proposition~\ref{prop:incompatibility-min-width}.
	\end{proof}
	
	From Theorem~\ref{thm:main} and Proposition~\ref{prop:width-min-incompatibility}, we get the following corollary that essentially says that if we allow any lexicographic order for direct access, then the fractional hypertree width determines the exponent of the preprocessing time.
	\begin{corollary}\label{cor:fhtw}
		Let $Q$ be a join query of fractional hypertree width $w$.
		\begin{itemize}
			\item There is an ordering $L$ of the variables of $Q$ such that there is an algorithm that allows lexicographic direct access with respect to the order induced by $L$ with preprocessing time $O(|D|^w)$ and logarithmic access time.
			\item Assuming the \pb{Zero-Clique} Conjecture, there are no  order $L$ and constant $\varepsilon > 0$ such that for all $\delta > 0$ there is an algorithm allowing lexicographic direct access with respect to the order induced by $L$ with preprocessing time $O(|D|^{w-\eps})$ and access time $O(|D|^\delta)$.
		\end{itemize}
	\end{corollary}
	
	\subsection{Unconstrained Orders}
	
	As we have seen, the complexity of direct access for join queries depends strongly on the order in which we want the query results to be represented. Moreover, we know from Corollary~\ref{cor:fhtw} that if we insist on lexicographic orders, then the preprocessing time necessarily degrades with the fractional hypertree width.
	The next natural question is whether we can get a more efficient algorithm in case we have no constraints on the order at all.
	
	Let us discuss a model in which we do not require any particular order on the answers in direct access. That is, we only require that there is a bijection $b: [s] \rightarrow Q(D)$ where $s:= |Q(D)|$ such that, given a query integer $i\in [s]$ we can efficiently compute the answer $b(i)$. We call this model \emph{orderless direct access}. It is again convenient to extend $b$ to the domain $\mathbb{N}$, defining $b(i)$ to be some fixed error value whenever $i\notin [s]$. We again measure the preprocessing time, i.e., the time to prepare suitable data structures given a database $D$, and the query time, which is the time to compute the answer $b(i)$ given $i\in \mathbb{N}$ using these data structures. We will see that there are queries in which orderless direct access can be performed in a more efficient manner than direct access with a lexicographic order, proving the following proposition.
	
	\begin{proposition}
		Assuming the Zero-Clique Conjecture, there exists a query for which orderless direct access can be done with a lower complexity compared to direct access in any lexicographic order.
	\end{proposition}
	
	To prove this, we consider the $4$-cycle query 
	\begin{align*}
		Q^\circ(x_1, x_2, x_3, x_4) \datarule R_1(x_1, x_2), R_2(x_2, x_3), R_3(x_3, x_4), R_4(x_4, x_1).
	\end{align*}
	It is easy to verify that the query $Q^\circ$ has fractional hypertree width $2$. It follows from Corollary~\ref{cor:fhtw} that for any direct access algorithm with a lexicographic order and polylogarithmic access time, the preprocessing time must be essentially quadratic. We next show that this preprocessing time can be improved if we drop the requirement of having a lexicographic order.
	
	\begin{lemma}\label{lem:cycle}
		There is an algorithm for orderless direct access for $Q^\circ$ with logarithmic access time and preprocessing time $O(|D|^{3/2})$.
	\end{lemma}
	\begin{proof}
		We use a decomposition technique based on degree information as pioneered in~\cite{AlonYZ97}. To this end, we split every relation $R_i^D$ into a \emph{heavy} part $R_i^h$ and a \emph{light} part $R_i^\ell$ as follows:
		\begin{align*}
			R_i^\ell &:= \{(a,b)\in R_i^D \mid \deg_{R_i}(a)\le |R_i^D|^{1/2}\},\\
			R_i^h &:= \{(a,b)\in R_i^D \mid \deg_{R_i}(a)> |R_i^D|^{1/2}\},
		\end{align*}
		where $\deg_{R_i}(a)$ denotes the \emph{degree} of $a$ in $R_i^D$, i.e.,
		\begin{align*}
			\deg_{R_i}(a) := | \{b\mid (a,b)\in R_i^D\}|.
		\end{align*}
		Denote by $D^*$ the database we get from $D$ by adding all these relations. Note that $|D^*| = \Theta(|D|)$. Clearly, $D^*$ can be computed in time $O(|D|)$ by sorting the tuples in $D$ in linear time and computing the degrees.
		
		We have that $R_i^D = R_i^\ell \uplus R_i^h$ where $\uplus$ denotes disjoint union. Thus, we can rewrite the query $Q^\circ$ as follows, where $\bowtie$ denotes a join:
		\begin{align*}
			Q^\circ \equiv &\left(R_1^\ell(x_1, x_2)\uplus R_1^h(x_1, x_2)\right)\bowtie  \left(R_2^\ell(x_2, x_3)\uplus R_2^h(x_2, x_3)\right) \\ &\bowtie \left(R_3^\ell(x_3, x_4)\uplus R_3^h(x_3, x_4)\right)\bowtie \left(R_4^\ell(x_4, x_1)\uplus R_4^h(x_4, x_1)\right)\\
			\equiv & \biguplus_{(o1, o2, o3, o4)\in \{\ell, h\}^4} R_1^{o_1}(x_1, x_2)\bowtie R_2^ {o_2}(x_2, x_3)\bowtie R_3^{o_3}(x_3, x_4)\bowtie R_4^{o_4}(x_4, x_1)\\
			\equiv & \biguplus_{(o1, o2, o3, o4)\in \{\ell, h\}^4} Q^{(o1, o2, o3, o4)}(x_1, x_2, x_3, x_4),
		\end{align*}
		where for every tuple $(o1, o2, o3, o4)\in \{\ell, h\}^4$ we define
		\begin{align*}
			Q^{(o1, o2, o3, o4)}(x_1, x_2, x_3, x_4) := R_1^{o_1}(x_1, x_2)\bowtie R_2^ {o_2}(x_2, x_3)\bowtie R_3^{o_3}(x_3, x_4)\bowtie R_4^{o_4}(x_4, x_1).
		\end{align*}
		We can evaluate each individual $Q^{(o1, o2, o3, o4)}$ efficiently.
		\begin{claim}\label{clm:cycle}
			For all $(o_1, o_2, o_3, o_4)\in \{\ell, h\}^4$, there is an algorithm for direct access on $Q^{(o1, o2, o3, o4)}$ with logarithmic access time and preprocessing time $O(|D|^{3/2})$.
		\end{claim}
		\begin{proof}
			The idea is to rewrite $Q^{(o1, o2, o3, o4)}$ into an acyclic query on a database of size $O(|D|^{3/2})$. We consider three cases.
			
			\textbf{Case 1 -- there exists $i\in \{1,2\}$ such that $o_i = o_{i+2} = \ell$: } We assume that w.l.o.g.~$i=1$ since the other case is completely analogous.
			We rewrite $Q^{(\ell, o2, \ell, o4)}$ into an acyclic query as follows: let $S_1(x_4, x_1, x_2)\datarule R_1^{\ell}(x_1, x_2), R_4^{o_4}(x_4, x_1)$ and $S_2(x_2, x_3, x_4)\datarule R_2^{o_2}(x_2, x_3), R_3^{\ell}(x_3, x_4)$ and define a new database $D'$ with two relations $S_1^{D'}:=S_1(D^*)$ and $S_2^{D'}:=S_2(D^*)$. Finally, we define the query $Q'(x_1,x_2, x_3, x_4)\datarule S_1(x_4, x_1, x_2), S_2(x_2, x_3, x_4)$. Then, by definition, $Q^{(\ell, o2, \ell, o4)}(D^*) = Q'(D')$, so it suffices to have an algorithm with direct access for $Q'$ such as those from~\cite{CarmeliRandom, bb:thesis}. It only remains to show that $D'$ can be computed in time $O(|D^{3/2}|)$. We only show this for the computation of $S_1^{D'}$, as the case for $S_2^{D'}$ is completely analogous.
			In the join $R_1^\ell \bowtie R_4^{o_4}$, every tuple $(a,b)\in (R_4^{o_4})^{D^*}$ can by construction of $(R_1^\ell)^{D^*}$ only be extended in at most $|R_1^D|^{1/2}$ ways, so overall $|S_1^{D'}| \le |R_4^{D^*}| \cdot |R_1^{D^*}|^{1/2} \le |D|^{3/2}$ and the relation can again be computed efficiently.
			
			\textbf{Case 2 -- there exists $i\in \{1,2\}$ such that $o_i = o_{i+2} = h$: }
			We only consider the case $i=1$; the other case is analogous. 
			As before, we show how $Q^{(h, o2, h, o4)}$ can be rewritten into an acyclic query, but we regroup the atoms differently. We define $S_1'(x_1, x_2, x_3) \datarule R_1^{h}(x_1, x_2), R_2^{o_2}(x_2, x_3)$ and $S_2'(x_3, x_4, x_1) \datarule R_3^{h}(x_3, x_4), R_4^{o_4}(x_4, x_1)$. Then we define $Q'(x_1,x_2, x_3, x_4)\datarule S_1'(x_1, x_2, x_3), S_2'(x_3, x_4, x_1)$ and construct a database $D'$ by setting $S_1'^{D'}:= S_1(D^*)$ and $S_2'^{D'}:= S_2(D^*)$. Reasoning similarly to before, it only remains to show that $D'$ can be constructed in time $O(|D|^{3/2})$.
			To this end, remark that every element in $\pi_{x_1}(R_1^{h})^{D^*}$ appears in at least $|R_1^D|^{1/2}$ tuples of $R_1^D$, so $|\pi_{x_1}(R_1^{h})^{D^*}| \le |R_1^{D}|^{1/2}$. We get that $|S_1'^{D'}| \le |\pi_{x_1}(R_1^{h})^{D^*}| \cdot |R_2^{D^*}| \le |D|^{3/2}$ and the relation can easily be computed in the claimed time.
			
			\textbf{Case 3 -- there is no $i\in \{1,2\}$ such that $o_i = o_{i+2}$: } In that case, 
			$(o1, o2, o3, o4)$ must be a cyclic shift of $(\ell, \ell, h, h)$, so we only consider that specific index tuple, as all other cases are completely analogous. We decompose the same way as in Case 2. We only have to show that the relations $S_1'^{D'}$ and $S_2'^{D'}$ can be computed in time $O(|D|^{3/2})$. For $S_2'$, this follows immediately as in Case 2 since $S_2'(x_3, x_4, x_1) \datarule R_3^{h}(x_3, x_4), R_4^{h}(x_4, x_1)$. For $S_1'(x_1, x_2, x_3) \datarule R_1^{\ell}(x_1, x_2), R_2^{\ell}(x_2, x_3)$, the reasoning is analogous to Case 1.
			
			Applying one of the direct-access algorithms from~\cite{bb:thesis,CarmeliRandom} in all cases yields the claim for all $(o1, o2, o3, o4)$.
		\end{proof}
		
		We now complete the proof of Lemma~\ref{lem:cycle}. We order $\{\ell, h\}^4$ in an arbitrary way. Then we order the answers in $Q^\circ(D)$ by first ordering them with respect to which $Q^{(o_1, o_2, o_3, o_4)}(D')$ they appear in and then with respect to the order of the algorithm from Claim~\ref{clm:cycle}. This is the order which our algorithm for~$Q^\circ$ uses, working in the following way: in the preprocessing phase, we first compute $D'$ and then perform the preprocessing for all $Q^{(o_1, o_2, o_3, o_4)}$ on input $D'$. We also compute the values $|Q^{(o_1, o_2, o_3, o_4)}(D')|$ for all queries $Q^{(o_1, o_2, o_3, o_4)}$ which we can do with a simple binary search using the respective direct access algorithms. Using these answer counts, given $j\in \mathbb{N}$ in the query phase, we first compute in which $Q^{(o_1, o_2, o_3, o_4)}(D')$ the $j$th answer lies and compute the corresponding index $j'$ for which the $j'$th answer in $Q^{(o_1, o_2, o_3, o_4)}(D')$ is the $j$th answer in $Q^\circ(D)$. Finally, we use the direct access algorithm for $Q^{(o_1, o_2, o_3, o_4)}$ to compute the desired answer.
	\end{proof}
	
	We remark that the technique used in the proof of Lemma~\ref{lem:cycle} has a rich history, comprised of algorithms based on splitting using degree information. In fact, it is one of the basic techniques for worst-case optimal join algorithms (see e.g.~\cite{NgoPRR18}). Notice that simply applying a worst-case optimal algorithm during preprocessing is not possible in our setting since the output can be too big to allow this in the desired runtime bounds. Thus the rewriting into an acyclic instance that we do in Claim~\ref{clm:cycle} is a second crucial ingredient. 
	
	For finding cycles of fixed length in graphs, splitting on degree information was first used in~\cite{AlonYZ97}. For conjunctive queries, the technique is also part of different algorithms based on so-called submodular width, which is a measure capturing more general classes of tractable queries than fractional hypertree width for decision problems, see~\cite{Marx13,Khamis0S17}; for the special case of cycles, this is also discussed in~\cite{Scarcello18}. In this line of work, it is in particular also known that submodular width can be used for efficient enumeration algorithms~\cite{BerkholzS19}. Note that our splitting is slightly stronger than that generally guaranteed by the techniques based on submodular width as it guarantees that the subinstances that we split into are disjoint. For the algorithms in~\cite{Marx13,Khamis0S17,BerkholzS19}, this is generally not the case, so they cannot directly be adapted to work for direct access as in the proof of Lemma~\ref{lem:cycle}. In fact, it is an open question whether disjoint splitting is always possible for algorithms based on submodular width, see~\cite{KhamisCMNNOS20}, and in particular it is an open question to which extent submodular width can be used for counting algorithms. Since direct access algorithms directly imply counting algorithms by binary search, this means that the  question to which extent Lemma~\ref{lem:cycle} can be generalized to additional queries is also open.
	
	\subsection{Projections and Partial Lexicographic Orders}
	
	In this section, we discuss two extensions. First, we would like to support conjunctive queries (with projections) instead of only join queries. Second, we would like to support \emph{partial lexicographic orders}. These are specified by a permutation $L$ of a \emph{subset} of the non-projected variables, and the answers are ordered lexicographically by this permutation as before; the order between answers that agree on all variables in $L$ is not specified as part of the problem definition. This way, we define a preorder on the query answers and turning it into a (complete) order is left to the discretion of the algorithm. More formally, we say that an order $\preceq$ of the query answers is \emph{compatible} with a partial lexicographic order if it is a refinement of that preorder.
	
	To achieve these extensions, we use the connection between our work and the notion of elimination orders.
	Since an order $L$ of the variables of a query $q$ is the reverse of an elimination order if and only if $Q$ is acyclic and $L$ contains no disruptive trio with $Q$, the preprocessing we suggested in Theorem~\ref{thm:cover-preprocess} can be seen as constructing an equivalent query and corresponding database that have the reverse of $L$ as an elimination order.
	For more information on the connection between elimination orders and decompositions, see~\cite{faq}.
	
	To support projections, we use a folklore construction that lets us eliminate any prefix of an elimination order while constructing an equivalent corresponding database in linear time.
	One way of doing that is to eliminate one variable at a time: first, filter the one atom that contains this variable and all its neighbors according to all other atoms that contain this variable (this is the semijoin operator), and then remove the variable from the query and the corresponding columns from the relations (this is the projection operation).
	To support partial orders, we can complete them into a full order and use the previous algorithm.
	This leads us to the following definition of the incompatibility number of conjunctive queries and partial lexicographic orders.
	
	\begin{definition}\label{def:project-partial}
		Let $Q$ be a conjunctive query and $L$ an ordering of a subset of the free variables of $Q$.
		Let $\mathcal{L}^+_Q$ be the set of all orderings of the variables of $Q$ that start with $L$ and end with the projected variables.
		The incompatibility number of $Q$ and $L$, denoted~$\iota$, is defined as the minimum over all incompatibility numbers of $Q$ and an ordering in $\mathcal{L}^+_Q$.
	\end{definition}
	Using this definition, we get an extension of Theorem~\ref{thm:cover-preprocess}.
	
	\begin{theorem}\label{thm:project-partial}
		Given a conjunctive query and an ordering $L$ of a subset of its free variables with incompatibility number $\iota$, direct access with an order compatible with the partial lexicographic order of $L$ can be achieved with $O(|D|^\iota)$ preprocessing and logarithmic access time.
	\end{theorem}
	
	\begin{proof}
		Let $L'$ be the ordering of the query variables yielding the incompatibility number according to Definition~\ref{def:project-partial}.
		We perform the preprocessing of the algorithm of Theorem~\ref{thm:cover-preprocess} to construct a database $D'$ corresponding to a disruption-free decomposition for $Q$ with respect to $L'$. This takes $O(|D|^\iota)$ time.
		We then eliminate the projected variables to be left with an equivalent join query with no projections and no disruptive trios. 
		This can be done in $O(|D'|)$ since the projected variables appear at the end of an elimination order. 
		We can now use the existing direct-access algorithm from Theorem~\ref{thm:linearcase} on the prefix of $L'$ obtained by removing the projected variables. This takes a further $O(|D'|)$ during preprocessing followed by logarithmic access time.
	\end{proof}
	
	The next natural question is whether this algorithm provides optimal time guarantees for the task at hand. 
	In contrast to queries without projections and complete lexicographic orders, the answer is negative.
	An example for this is the $4$-cycle query $Q^\circ$ from the previous section. There we saw that if the partial lexicographic order is empty, then we can perform direct-access faster than any completion to a full lexicographic order. If all variables of this query were projected, it would have at most one answer, and the existence of this answer can be determined in time $O(|D|^{3/2})$ since the algorithm of Lemma~\ref{lem:cycle} in particular lets us decide the existence of an answer. This runtime is again faster than any completion to a full order and so better than what we get from Theorem~\ref{thm:project-partial}.
	
	We remark that, when only considering linear preprocessing and logarithmic delay, all queries and partial orders that can be solved within this time guarantee, can be solved this way through a completion to a full lexicographic order~\cite{CarmeliTGKR20}. In contrast, as we have seen here, this is no longer the case when considering larger preprocessing times.

	\section{Lower Bounds for Enumeration of Cyclic Joins}\label{sec:enumeration}
	
	In this section we show another application of the \pb{Zero-Clique} Conjecture. Specifically, we prove enumeration lower bounds for cyclic (i.e.~non-acyclic) joins, assuming the hardness of \pb{Zero-Clique}. These lower bounds reprove a dichotomy from~\cite{bb:thesis} under a different complexity assumption.
	
	First, in Section~\ref{sec:set-intersect-enum} we prove a lower bound for a variation of \pb{$k$-Set-Intersection}. In Section~\ref{sec:lw-joins} we prove hardness of enumerating Loomis-Whitney joins, and then we use a reduction from Loomis-Whitney joins to general cyclic joins in Section~\ref{sec:cyclic-enum}.
	
	\subsection{Hardness of Set-Intersection-Enumeration}
	\label{sec:set-intersect-enum}
	
	We start by defining an offline variant of \pb{$k$-Set-Intersection}, in which we want to enumerate all elements of all queried intersections.
	
	\begin{definition} \label{def:setintersection_enum}
		In the \pb{$k$-Set-Intersection-Enumeration} problem, we are given an instance $\I$ consisting of a universe $U$ and families $\mathcal{A}_1, \ldots, \mathcal{A}_k \subseteq 2^U$, where we denote the sets in family $\A_i$ by $S_{i,1},\ldots,S_{i,|\A_i|}$. We are also given a set of queries $Q$, where each query specifies indices $(j_1,\ldots,j_k)$. The goal is to enumerate all answers to all queries, that is, to enumerate all pairs $(q,u)$ such that $q = (j_1,\ldots,j_k) \in Q$ and $u \in S_{1,j_1}\cap \ldots \cap S_{k,j_k}$.
	\end{definition}
	
	As usual for enumeration problems, we assume algorithms for \pb{$k$-Set-Intersection-Enumeration} work in two phases: in a first phase, we preprocess the input into an index data structure. In a second phase, we enumerate the answers one after the other, measuring the efficiency of the algorithm in the \emph{delay}, i.e., the time between two consecutive answers.
	
	We prove the following analogue of Theorem~\ref{thm:intersection}.
	
	\begin{lemma}\label{lem:enumerationintersection}
		Assuming the \pb{Zero-$(k+1)$-Clique} Conjecture, for every constant $\varepsilon> 0$ there exists a constant $\delta > 0$ such that no randomized algorithm solves \pb{$k$-Set-Intersection-Enumeration} on universe size $|U|=n$ with at most $n$ sets in preprocessing time $O(n^{k+1-\varepsilon})$ and delay $O(n^\delta)$.
	\end{lemma}
	
	We prove Lemma~\ref{lem:enumerationintersection} by contradiction in the rest of this subsection. So assume that \pb{$k$-Set-Intersection-Enumeration} on universe size $|U|=n$ can be solved in preprocessing time $O(n^{k+1-\varepsilon})$ and delay $O(n^\delta)$ for some $\eps > 0$ and $\delta := \frac{\eps}{4k}$.
	
	We closely follow the proof of Theorem~\ref{thm:intersection}, and we only describe the differences. We need an additional source of randomness in the new edge weights $w'$, therefore for every $v \in V_1$ we pick a random $y_v \in \Fp$, we subtract $y_v$ from every edge weight $w'(v,u)$ for each $u \in V_{k+1}$, and we add $y_v$ to every edge weight $w'(v,u)$ for each $u \in V_2$. This does not change the weight of a clique $(v_1,\ldots,v_{k+1}) \in V_1 \times \ldots \times V_{k+1}$.
	In more detail, we now choose uniformly and independently at random from $\Fp$:
	\begin{itemize}
		\item one value $x$,
		\item for all $v \in V_{k+1}$ and all $j \in [k-1]$ a value $y_v^j$, and 
		\item for each $v \in V_1$ a value $y_v$.
	\end{itemize}
	We then define the new weight function $w'$ by setting for any $i,j \in [k+1]$ and any $v \in V_i$ and $u \in V_j$:
	\begin{align} \label{eq:def_new_weights_two}
		w'(v, u) = x \cdot w(v, u) + 
		\begin{cases} y_{u}^1 - y_v, & \text { if } i=1, j=k+1\\ 
			y_{u}^i - y_u^{i-1}, & \text { if } 2 \le i < k, j = k+1\\ 
			- y_u^{k-1}, & \text { if } i=k, j=k+1\\ 
			y_v, & \text{ if } i=1, j=2\\
			0, & \text{ otherwise }
		\end{cases}
	\end{align}
	With this minor modification of $w'$, we follow the rest of the construction verbatim, in particular we pick the same parameter $\rho := \frac{\eps}{2k}$, and for every weight interval tuple $(I_0,\ldots,I_k) \in S$ we construct the instance of \pb{$k$-Set-Intersection-Enumeration} given by 
	\begin{align*}
		\mathcal{A}_i := \{S_{i,v}\mid v\in V_i\} \;\text{ where }\; S_{i,v} := \{ u\in V_{k+1}\mid w'(v,u)\in I_i\}
	\end{align*}
	for $i \in [k]$, and the set of queries
	\begin{align*}
		Q := \{(v_1,\ldots,v_k) \in V_1\times \ldots \times V_k \mid w'(v_1,\ldots,v_k) \in I_0 \}.
	\end{align*}
	Note that we have for every $i\in [k]$ that $|\mathcal A_i| = |V_i|$, so $\sum_i |\A_i|\le n$.
	We then run the assumed \pb{$k$-Set-Intersection-Enumeration} algorithm on $\I=(\A_1,\ldots,\A_k)$ and $Q$. If this enumeration at some point returns an index corresponding to a zero-clique, then we return this zero-clique. If the enumeration ends without producing a zero-clique, then we return 'no'.
	
	By the same arguments as in the proof of Theorem~\ref{thm:intersection}, since every zero-clique has a weight interval tuple in $S$, 
	every zero-clique appears as an answer to some query $q \in Q$. Thus, by enumerating all answers we enumerate all zero-cliques. In particular, if there exists a zero-clique, then the algorithm will find one. Since we filter the results for zero-cliques, it also holds that if there is no zero-clique then the algorithm will return 'no'. This proves correctness.
	
	It remains to analyze the running time. 
	In total over all $|S| = O(n^{k\rho})$ instances, we need preprocessing time $O(n^{k+1-\eps+k\rho}) = O(n^{k+1-\eps/2})$ for our choice of $\rho = \frac{\eps}{2k}$.
	To analyze the total delay, consider the following claim.
	
	\begin{claim}\label{clm:exptectedfalsepositives}
		With probability at least $.99$, the number of cliques $C = (v_1,\ldots,v_{k+1}) \in V_1 \times \ldots \times V_{k+1}$ such that $C$ is not a zero-clique and the weight interval tuple of $C$ lies in $S$ is bounded by $O(n^{k+1-\rho})$.
	\end{claim}
	
	Let us first finish the correctness proof, and later prove the claim. 
	Each clique as in the claim is a ``false positive'', that is, an answer to our queries that does not correspond to a zero-clique. It follows that with probability at least $.99$, all our \pb{$k$-Set-Intersection-Enumeration} instances together enumerate $O(n^{k+1-\rho})$ answers before seeing the first answer that corresponds to a zero-clique. 
	The total query time is then bounded by $O(n^{k+1-\rho+\delta}) = O(n^{k+1-\rho/2})$ for our choice of $\rho = \frac{\eps}{2k}$ and $\delta = \frac{\eps}{4k}$.
	Together with the preprocessing, this solves \pb{Zero-$(k+1)$-Clique} in time $O(n^{k+1-\eps/2}) + O(n^{k+1-\rho/2}) = O(n^{k+1-\rho/2})$, which contradicts the \pb{Zero-$(k+1)$-Clique} Conjecture.
	
	Note that the time bound here only holds with probability $.99$. We can make the time bound hold deterministically, by aborting the algorithm after an appropriate time $O(n^{k+1-\rho/2})$. Then the running time is deterministic, and the algorithm still finds a zero-clique with probability at least $.99$. We can boost the success probability by repeating this algorithm.
	
	It remains to prove the claim.
	
	\begin{proof}[Proof of Claim~\ref{clm:exptectedfalsepositives}]
		The proof is similar to the one of Claim~\ref{clm:secondcutoff}.
		Let $(v_1,\ldots,v_k,v) \in V_1\times \ldots \times V_{k+1}$ be a tuple that is not a zero-clique.
		We use the shorthand notation $w'_i := w'(v_i,v_{k+1})$ for $i \in [k]$ and $w'_0 := w'(v_1,\ldots,v_k)$. The values $w_0,\ldots,w_k$ are defined analogously for the original edge weight function $w$.
		Expanding the definition of $w'$ (see (\ref{eq:def_new_weights_two})) we obtain
		\begin{align}
			w'_0 &= x \cdot w_0 + y_{v_1},\notag \\
			w'_1 &= x \cdot w_1 + y_v^1 - y_{v_1},\label{eq:bijectiontwo}\\
			w'_i &= x \cdot w_i + y_v^i - y_v^{i-1} \;\text{for each }i\in \{2, \ldots, k-1\}\notag\\
			w'_k &= x \cdot w_k - y_v^{k-1},\notag
		\end{align}
		
		We claim that for fixed $w_0,\ldots,w_k$ the above equations induce a bijection sending $(w'_0,\ldots,w'_k)$ to $(x,y_{v_1},y_v^1,\ldots,y_v^{k-1})$. Indeed, by (\ref{eq:bijection}) we can compute $(w'_0,\ldots,w'_k)$ given $(x,y_{v_1},y_v^1,\ldots,y_v^{k-1})$. For the other direction, we use that $(v_1,\ldots,v_k,v)$ is no zero-clique, that is, $0 \ne w(v_1,\ldots,v_k,v) = w_0+\ldots+w_k$. 
		By adding up all equations~(\ref{eq:bijection}), since each term $y_v^i$ and $y_{v_1}$ appears once positively and once negatively, we obtain $w'_0+\ldots+w'_k = x \cdot (w_0+\ldots+w_k)$. Hence, given $(w'_0,\ldots,w'_k)$ we can compute $x = (w'_0+\ldots+w'_k)/(w_0+\ldots+w_k)$. Note that here we do not divide by~0, since $w_0+\ldots+w_k \ne 0$. Next we can compute $y_{v_1} = x \cdot w_0 - w'_0$ by rearranging the first equation of (\ref{eq:bijection}) (recall that $w_0$ is fixed, we are given $w'_0$, and we just computed $x$, so each variable on the right hand side is known). Similarly, from $x$ and $y_{v_1}$ we can compute $y_v^1 = x \cdot w_1 - y_{v_1} - w'_1$, and for any $1 < i < k$ we can compute $y_v^i = x \cdot w_i - y_v^{i-1} - w'_i$. This computes $(x,y_{v_1},y_v^1,\ldots,y_v^{k-1})$ given $(w'_0,\ldots,w'_k)$.
		
		Since we showed that there is a bijection sending $(w'_0,\ldots,w'_k)$ to $(x,y_{v_1},y_v^1,\ldots,y_v^{k-1})$, and since $x,y_{v_1},y_v^1,\ldots,y_v^{k-1}$ are all chosen independently and uniformly random from $\mathbb{F}_p$, we obtain that $(w'_0,\ldots,w'_k)$ is uniformly random in $\Fp^{k+1}$. Since each weight interval $I_i$ is of length $O(p n^{-\rho})$, we have that $w'_i \in I_i$ holds with probability $O(\frac{p n^{-\rho}}p) = O(n^{-\rho})$, for each $i \in \{0,1,\ldots,k\}$. In total, $(v_1,\ldots,v_k,v)$ has a specific weight interval tuple $(I_0,\ldots,I_k)$ with probability $O(n^{-(k+1)\rho})$. Since there are $O(n^{k\rho})$ tuples in $S$, it follows that $(v_1,\ldots,v_k,v)$ has a weight interval tuple in $S$ with probability $O(n^{-\rho})$.
		
		Since there are $O(n^{k+1})$ choices for $(v_1,\ldots,v_k,v)$, it follows that the expected number of non-zero-cliques with a weight interval tuple in $S$ is $O(n^{k+1-\rho})$. The claim now follows from an application of Markov's inequality.
	\end{proof}

	\subsection{Hardness of Loomis-Whitney Joins}
	\label{sec:lw-joins}
	
	In this section, we prove the hardness of enumerating answers to Loomis-Whitney joins based on the hardness results from the previous section. Loomis-Whitney joins are queries of the form
	\[LW_k(x_1, \ldots, x_k) \datarule R_1(\bar X_1), \ldots, R_k(\bar X_k)\]
	where for each $i\in[k]$, we have that $\bar X_i = x_1, \ldots x_{i-1}, x_{i+1} \ldots, x_k$.
	When $k=3$, $LW_3$ is simply the (edge-colored) triangle query:
	\[LW_{3}(x_1, x_2, x_3) \datarule R_1(x_2, x_3), R_2(x_1, x_3), R_3(x_1, x_2)\]
	Loomis-Whitney joins are thus generalizations of the triangle join where instead of $3$ nodes in a graph where every $2$ share an edge, we are looking for $k$ nodes in a hypergraph where every $k-1$ nodes share an edge.
	Loomis-Whitney joins are of special interest because, as we explain in Section~\ref{sec:cyclic-enum}, they are in a well-defined sense the obstructions to a query being acyclic.
	
	There is a naive way of designing constant delay algorithms for Loomis-Whitney joins: one can materialize the query result during preprocessing and then simply read them from memory in the enumeration phase. This solution requires preprocessing time $O(|D|^{1+1/(k-1)})$ using the algorithm by Ngo et al.~\cite{NgoPRR18} for Loomis-Whitney joins or, since $LW_k$ has fractional cover size $1+1/(k-1)$, any worst-case optimal join algorithm from Theorem~\ref{thm:worst-case-joins} will give the same guarantees. In this section, we show that this naive approach is likely optimal for enumeration of Loomis-Whitney joins.
	
	\begin{theorem}\label{thm:LW-hardness}
		For every $k\ge 3$ and every $\varepsilon > 0$ there is a $\delta> 0$ such that there is no enumeration algorithm for $LW_k$ with preprocessing time $O(|D|^{1+\frac{1}{k-1}-\varepsilon})$ and delay $O(|D|^\delta)$, assuming the \pb{Zero-$k$-Clique} Conjecture. 
	\end{theorem}
	\begin{proof}
		We show that a fast enumeration algorithm for $LW_{k}$ yields a fast algorithm for \pb{$(k-1)$-Set-Intersection-Enumeration}, which by Lemma~\ref{lem:enumerationintersection} breaks the \pb{Zero-$k$-Clique} Conjecture.
		
		Assume that there are $k \ge 3$ and $\eps>0$ such that for all $\delta>0$ there is an enumeration algorithm for $LW_{k}$ with preprocessing time $O(|D|^{1+1/{(k-1)}-\varepsilon})$ and delay $O(|D|^\delta)$. Consider an instance $\I$ of \pb{$(k-1)$-Set-Intersection-Enumeration} with universe size $|U|=n$ and $O(n)$ sets as in the statement of Lemma~\ref{lem:enumerationintersection}. That is, we are given families $\mathcal{A}_1, \ldots, \mathcal{A}_{k-1} \subseteq 2^U$, where we write $\A_i = \{S_{i,1},\ldots,S_{i,|\A_i|}\}$ and $\sum_i |\mathcal A_i| \le n$, and a set of queries $Q$, where each query is of the form $(j_1,\ldots,j_{k-1})$ with $j_i \in [|\A_i|]$. We construct a database $D$ as follows. The domain of $D$ is $U\cup [\max_{i\in [k-1]}|\mathcal A_i|]$. 
		Note that $\max_{i\in [k-1]}|\mathcal A_i|\le \sum_i |\mathcal A_i| \le n$, so the overall domain size of $D$ is~$O(n)$.
		
		For each $i \in [k-1]$, let $i^+$ be the unique value in $[k-1]$ with $i+1 \equiv i^+ \bmod {(k-1)}$.
		Note that the atom $R_{i}(\bar X_{i})$ of $LW_{k}$ contains the variables $x_{i^+}$ and $x_{k}$.
		We define the relations:
		\begin{align*}
			R_{i}'^{D} &:= \{(j, v) \mid v \in S_{i^+,j}\} \;\;\text{ for each } i \in [k-1], \\
			R_{k}^D &:= Q.
		\end{align*}
		When interpreted on $D$, the join query
		\begin{align*}
			LW'_{k} \datarule R_{1}'(x_{1^+}, x_{k}), \ldots, R_{k-1}'(x_{{(k-1)}^+}, x_{k}), R_{k}(x_1, \ldots, x_{k-1})
		\end{align*}
		computes the conjunction $x_{k}\in S_{1^+, x_{1^+}}\land  \ldots \land x_{k}\in S_{{(k-1)}^+, x_{(k-1)^+}}$ and $q= (x_1, \ldots, x_{k-1})\in Q$ because the join query (on $D$) formally means $\bigwedge_{i\in [k-1]} x_k \in S_{i^+, x_{i^+}} \land (x_1, \ldots, x_{k-1})\in Q$. Therefore, it computes exactly the answers to $q\in Q$ for the \pb{$(k-1)$-Set-Intersection-Enumeration} instance. To embed $LW'_{k}$ into $LW_{k}$, we create $(k-1)$-ary relations $R_j^D$, for $j\in [k-1]$, by simply adding the columns for the missing attributes $x_i, i\in [k-1]\setminus \{j, j^+\}$ in the relations~$R_j'^D$. To this end, we extend each $2$-tuple in $R_j'$ by all possible combinations of assignments to $[n]$ to the $k-3$ attributes in $[k-1]\setminus \{j, j^+\}$. Clearly, 
		\begin{align*}
			LW_{k}'(D) = LW_{k}(D),
		\end{align*}
		so running the assumed enumeration algorithm for $LW_{k}$ on the instance $D$ solves \pb{$(k-1)$-Set-Intersection-Enumeration}.
		
		It remains to analyze the running time of the resulting algorithm. First observe that all relations in $D$ have arity $k-1$. Since the domain size of $D$ is $O(n)$, the relation sizes are at most $O(n^{k-1})$, and so $|D| = O(n^{k-1})$. Moreover, $D$ can easily be constructed in time $O(n^{k-1})$. Thus, the overall preprocessing time is $O(n^{k-1})$ plus the preprocessing time of the algorithm for $LW_{k}$, which is
		\begin{align*}
			O(n^{k-1} + |D|^{1+1/{(k-1)}-\varepsilon}) = O(n^{k-\varepsilon'}),
		\end{align*}
		for $\eps' := \min\{1, (k-1) \eps\}$.
		For any desired $\delta' > 0$, by setting $\delta := \delta'/(k-1)$ we obtain delay $O(|D|^{\delta}) = O(n^{(k-1) \delta}) = O(n^{\delta'})$. By Lemma~\ref{lem:enumerationintersection}, this violates the \pb{Zero-$k$-Clique} Conjecture.
	\end{proof}
	
	\subsection{Hardness of Cyclic Joins}\label{sec:cyclic-enum}
	
	In this section, we describe the implication of hardness of Loomis-Whitney joins on the hardness of other cyclic joins, and derive an enumeration lower bound for self-join free cyclic joins based of \pb{Zero-$k$-Clique}.
	A known characterizarion of cyclic hypergraphs is that they contain a chordless cycle or a non-conformal clique (i.e., a set of pairwise neighbors that are not contained in an edge), see e.g.~the survey~\cite{Brault-Baron16}. By considering a minimal non-conformal clique, Brault-Baron~\cite{bb:thesis} used this property to claim that every cyclic query contains either a chordless cycle or a set of $k$ variables such that there is no atom that contains all of them, but for every subset of size $k-1$ there is such an atom. In other words, by removing all occurrences of some variables, and removing some atoms whose variables are contained in other atoms, we can obtain either a Loomis-Whitney join or a chordless cycle join.
	In case of a chordless cycle, we can always use it to solve a chordless cycle of length $3$ (also called a triangle), which can be seen as a Loomis-Whitney join of size $k=3$.
	Stated in different words, Brault-Baron~\cite{bb:thesis} essentially proved the following lemma:
	
	\begin{lemma}\label{lemma:cyclic-LW-reduction}
		Let $Q$ be a self-join free cyclic join. There exists $k\ge 3$ such that there is an exact reduction from $LW_k$ to $Q$.
	\end{lemma}
	
	Since enumeration in linear preprocessing time and constant delay is closed under exact reductions,
	by combining Lemma~\ref{lemma:cyclic-LW-reduction} with Theorem~\ref{thm:LW-hardness}, we conclude the enumeration hardness of cyclic joins based on the hardness of Zero-$k$-Clique: they have no enumeration algorithm with linear preprocessing time and constant delay.
	
	\begin{theorem}
		Let $Q$ be a self-join free cyclic join. Assuming the \pb{Zero-Clique} Conjecture, there exists an $\eps >0$ such that there is no enumeration algorithm for $Q$ with preprocessing time $O(|D|^{1+\eps})$ and delay $O(|D|^\eps)$. 
	\end{theorem}
	\begin{proof}
		Let $k$ be the constant from Lemma~\ref{lemma:cyclic-LW-reduction}. Then apply Theorem~\ref{thm:LW-hardness} with $\varepsilon := \min\left(\frac{1}{2(k-1)}, \delta\right)$.
	\end{proof}
	
	\section{Conclusion}
	
	In this paper, we identified a parameter that we call the incompatibility number of a join query and an order of its variables. We proposed an algorithm for direct access to the answers of a join query in the lexicographic order induced by the variable order, using a reduction to the acyclic case without disruptive trios which had been studied in~\cite{CarmeliTGKR20}. The exponent of the preprocessing phase of this algorithm is exactly the incompatibility number. 
	We then complemented this upper bound by a matching lower bound that shows that the incompatibility number cannot be beaten as the exponent of the preprocessing for any join query, assuming the Zero-Clique Conjecture from fine-grained complexity theory.
	Due to the relationship of the incompatibility number and fractional hypertree width, it follows that efficient direct access in any lexicographic order requires preprocessing that is exponential in the fractional hypertree width of the query.
	As part of the lower bound proof, we established the complexity equivalence between queries with the same structure with and without self-joins for direct access.
	
	With only a minor change to the proof of our direct-access lower bounds, we also showed lower bounds for enumeration. First, we showed that for Loomis-Whitney joins we cannot significantly improve upon a simple algorithm that computes all answers during preprocessing, assuming the Zero-Clique Conjecture. Then, under the same assumption, we showed that
	acyclic queries are the only self-join free join queries whose answers can be enumerated in constant delay after linear preprocessing. This gives further evidence to a dichotomy that was already shown using a different complexity hypothesis.
	
	We have also presented several extensions of our main results on direct access where we weakened the order requirements or allowed projections. In these settings, we have seen that our approach does \emph{not} yield optimal algorithms. It would be interesting to understand the complexity landscape for these extensions better, and we leave this as an open question for future work.
	
	On a more technical level, we have seen that the Zero-Clique Conjecture is tightly linked to \pb{$k$-Set-Intersection} and \pb{$k$-Set-Disjointness}, which can be seen as variants of a join query. Thus, it seems plausible that there are additional applications of the Zero-Clique Conjecture in database theory, and it would be interesting to further explore such connections.

	\bibliographystyle{ACM-Reference-Format}
	\bibliography{direct-access}

\end{document}